\documentclass{amsart}

\usepackage[english]{babel}
\usepackage{amssymb}
\usepackage{mathtools}
\usepackage{enumitem}
\usepackage{mathrsfs}
\usepackage[utopia]{mathdesign}
\usepackage[breaklinks]{hyperref}
\usepackage{xcolor}
\usepackage{hyperref}
\usepackage{fullpage}
\usepackage{caption}
\usepackage[section]{placeins}
\usepackage{tabularx}

\newcommand{\C}{\mathcal{C}}
\newcommand{\loc}{\textrm{loc}}
\newcommand{\N}{\mathbb{N}}
\newcommand{\R}{\mathbb{R}}
\newcommand{\tot}{\textrm{total}}

\renewcommand{\Im}{\textrm{Im}}
\renewcommand{\vec}[1]{{\bf #1 } }
\newcommand{\citep}[1]{\cite{#1}}
\newcommand{\sys}[1]{\left \{ \begin{aligned} #1 \end{aligned} \right. }

\newcommand{\ini}{\mathrm{in}}

\DeclareMathOperator*{\argmax}{arg\,max}

\newcommand{\nouveau}[1]{{#1}}
\newcommand{\nv}[1]{{#1}}

\newcommand{\Vb}{{\underline{V}}}

\newcommand{\rhob}{\underline{\rho}}

\newcommand{\bl}{\pm}

\newtheorem{theorem}{Theorem}[section]
\newtheorem{lemma}[theorem]{Lemma}

\newtheorem{proposition}[theorem]{Proposition}
\newtheorem{conjecture}[theorem]{Conjecture}

\newtheorem{definition}[theorem]{Definition}

\theoremstyle{remark}
\newtheorem{remark}[theorem]{Remark}

\numberwithin{equation}{section}

\newcounter{hyp}

\AtBeginDocument{
   \def\MR#1{}
}

\begin{document}
\title{Numerical study of the sharp stratification limit towards bilayer models}
\date{\today}
\author{Théo Fradin}
\address{Univ. Bordeaux, CNRS, Bordeaux INP, IMB, UMR 5251, F-33400 Talence, France \newline \indent Univ Rennes, CNRS, IRMAR - UMR 6625, F-35000 Rennes, France }
\email{theo.fradin@math.u-bordeaux.fr}
\maketitle

\begin{abstract}
In the study of oceanic flows at the geophysical scale, the phenomenon of  density stratification plays a central role in the dynamics of the system. Two categories of mathematical models are commonly used to describe the role played by the density stratification: on the one hand, continuously stratified models - such as the stratified Euler equations in a strip, considered in the present article - offer an accurate description of vertical effects, but come with a high level of complexity, both at the theoretical and numerical levels. On the other hand, bilayer models approximate the stratification by a piecewise constant profile. In the latter case, the main point is to study the evolution of the free interface between both layers, which leads to a substantially simplified model. In the present article, we compare both approaches  in the framework of the linearized stratified Euler equations around density profiles that are close to piecewise constant profiles, and prove the convergence towards the bilayer Euler equations. However, in the presence of a shear flow, bilayer models have a range of validity limited by the presence of Kelvin-Helmholtz instabilities. In this case, we  use a suitable normal modes decomposition to compute numerically the dispersion relation of this linearized model, and provide numerical evidence that the Kelvin-Helmholtz instabilities limit the applicability of two widely used bilayer models, namely the bilayer Euler equations and the bilayer shallow-water equations.
\end{abstract}
\section{Introduction}
In this article, we study the phenomenon of vertical density stratification in the ocean, and its effect on gravity waves. At the geophysical scale, density variations in the ocean are often considered as a small perturbation of a smooth, strictly decreasing density profile $\rhob$ that only depends on the vertical variable $z$; we then speak of a continuous and (strictly) stable stratification. One can write the incompressible stratified Euler equations to describe the evolution of the water velocity and density\nv{. I}n this study these equations are posed in the 2D domain $\mathbb{T}_L \times [-H,0]$, where $\mathbb{T}_L$ is the one-dimensional torus of length $2 \pi L$, for some $L > 0$, and where $H> 0$ is the depth of the domain; our setting is summarized in Figure \ref{fig:setting}. These equations can \nv{naturally} be extended to other domains, including 3D ones. However, our numerical strategy is implemented for the strip $\mathbb{T}_L \times [-H,0]$ so that we restrict to this case for the sake of conciseness. Because we focus on the stratification phenomenon, we neglect other ingredients, despite their major role in the description of the ocean at the geophysical scale (for instance, Coriolis forcing, wind stress, non-flat topography, vertical boundaries). \\

At the non-linear level, local well-posedness of the stratified Euler equations in finite regularity Sobolev spaces has been proven recently, see \cite{Desjardins2019} and \cite{Fradin2024} for studies close to our setting, and for further references. The study of the qualitative behavior of these equations, even at the linear level, is \nv{challenging}, and often reduced to specific density profiles $\rhob$. One typical profile is when $\rhob$ is a decreasing exponential with height, which \nv{is referred to as} homogeneous stratification; in this case, the so-called Brunt-Väisälä frequency does not depend on $z$, which simplifies parts of the analysis. In the present study, we do not make this assumption. Also, the numerical resolution of the stratified Euler equations can be costly (see however \cite{CaulfieldVieweg2025} for a study on the large time behavior using direct numerical simulations, and references \nv{therein}).\\

It is thus natural to write reduced models that still capture the dynamics of the system at hand. It is well known in oceanography (see \cite[Part I, Section 3]{Vallis2017}) that a ubiquitous density profile consists of two layers of (almost) constant density. These two layers are separated by a thin layer of thickness $\delta > 0$, called the pycnocline, where the density gradient is very large, see the profile $\rhob$ in Figure \ref{fig:shear:sharp} - we \nv{refer to} such a profile \nv{as} a sharp stratification profile. In this case it is very common to consider bilayer models, where two layers of constant density (with the density of the upper layer smaller than the density of the lower layer) are separated by an interface. In the irrotational case, such models can be reduced to the study of the evolution of quantities that do not depend on $z$, such as the bilayer irrotational Euler equations (sometimes called the interfacial waves equations), or the bilayer shallow-water equations (see \cite[Chapter B, Section 6]{Duchene2022a}). Then their numerical resolution is much cheaper, and 
their modal analysis at the linear level can be carried out analytically.\\

In this study we ask whether bilayer models provide a good approximation of the stratified Euler equations, with continuous density variations, in the so-called sharp stratification limit, i.e. when $\delta \to 0$. Throughout this article, we consider linear models, and we distinguish two distinct frameworks.\\

In the case where there is no shear velocity at equilibrium between both layers, we briefly show that the linear bilayer Euler equations are well-posed in finite-regularity Sobolev spaces. Then we rigorously prove the convergence of solutions of the linear stratified Euler equations towards the linear bilayer Euler equations, in the sharp stratification limit $\delta \to 0$. More precisely, we provide quantitative estimates between the solutions of the aforementioned continuous and bilayer models, in this linear case and without a shear velocity between both layers; this what we call {\it full justification}, see Section \ref{subsection:sharp_strat:sans_shear}. \\

However, this simple setting fails to take shear into account, which naturally arises either from the non-linear dynamics or from the presence of a background current. This is why we also consider the stratified Euler equations, linearized around a horizontal velocity profile $\Vb$ that only depends on the vertical variable $z$, and no vertical velocity; such a profile is called a shear flow. We then study the profile $\Vb_{\delta}$ given in Figure \ref{fig:shear:sharp} of a sharp shear profile: in the limit $\delta \to 0$, the profile $\Vb_{\delta}$ converges towards a piecewise constant function of height, with one discontinuity, at the interface between both layers. In the homogeneous case (i.e. with no density variations), the so-called Kelvin-Helmholtz instabilities associated to a shear flow of the form of $\Vb_{\delta}$ have been partially described through a rigorous spectral analysis, see \cite{Lin2003, KumarOzanski2025}.\\

In the limit $\delta \to 0$, i.e. for piecewise constant density profile and shear flow, the expected system is the bilayer Euler equations. This system suffers from the Kelvin-Helmholtz instabilities. In this bilayer setting, these instabilities are strong enough to prevent the well-posedness in Sobolev spaces (see \cite{IguchiTanakaTani1997}, \cite{LebeauKamotski2005}, as well as \cite[Chapter A, Section 3.5]{Duchene2022a} and references \nv{therein}). Recall that this is in contrast with well-posedness results on the stratified Euler equations, and in particular with \cite{Fradin2024}, which includes a non-zero shear flow. However the latter result is not uniform in $\delta$. Still, well-posedness of the bilayer Euler equations holds in analytic function spaces, (see for instance \cite{SulemSulemBardos1981} and \cite{CaflischOrellana1986} for the vortex sheet problem, which is closely related to the bilayer Euler equations, and in particular suffers from the Kelvin-Helmholtz instabilities as well).\\

An important question is to accurately describe the Kelvin-Helmholtz instabilities. \nouveau{These instabilities are fairly well understood in the bilayer setting.} In \cite{Lannes2013a}, the author shows that the (non-linear) bilayer Euler equations including surface tension are well-posed in Sobolev spaces. Moreover, the amount of surface tension that is sufficient to ensure well-posedness uniformly in the shallow-water parameter vanishes in the shallow-water limit, consistently with the fact that the non-linear bilayer shallow-water equations are well-posed, even without surface tension, as a quasilinear strictly hyperbolic system, under some hyperbolicity criterion (see \cite{Ovsjannikov1979}, \nouveau{\cite{GuyenneLannesSaut2010}} as well as \cite[Chapter B, Section 6.2]{Duchene2022a} and references \nv{therein}). It should be \nv{noted} that surface tension is relevant for water-air interfaces, when the air is considered to be an incompressible fluid, but is not expected to be present in water-water interfaces.\\

Considering now the link between the stratified Euler equations and the bilayer Euler equations, the author in \cite{James2001} justified the sharp stratification limit for traveling waves. Recently, the authors in \cite{Adim2024} studied the sharp stratification limit for the hydrostatic stratified Euler equations (closely related to the stratified Euler equations, considered in the present study). Their analysis strongly relies on the presence of eddy diffusivity, which regularizes the system. In order to prove their convergence result, they write a new reduced model that takes into account some effects in the pycnocline, we refer to \cite{Adim2024} for more details on their model, as well as other physical models that aim at describing phenomena happening inside a thin pycnocline.  \\

In the present study, we provide a comparison of the Kelvin-Helmholtz instabilities in both the stratified Euler equations, linearized around the shear flow $\Vb_{\delta}$, and bilayer models, thanks to a numerical description of the dispersion relation of the continuous model. We provide numerical evidence that the Kelvin-Helmholtz instabilities are present in the linearized stratified Euler equations. Moreover, to some extent, they are well described by the bilayer Euler equations, linearized around a shear flow. As a consequence, we conjecture that the Kelvin-Helmholtz instabilities prevent any full justification of the sharp stratification limit from the stratified Euler equations in the presence of a shear flow as the one in Figure \ref{fig:shear:sharp} (see also \eqref{eqn:def:sharp:shear}) towards the bilayer Euler equations, in the context of finite-regularity Sobolev spaces. Note that this fact is expected, as the bilayer Euler equations with a shear flow are ill-posed in Sobolev spaces. However our findings also indicate that the sharp stratification limit towards the bilayer shallow-water equations cannot be fully justified in the setting of Sobolev spaces, despite the fact that the latter model is well-posed in Sobolev spaces. \\

The setting of a sharp stratification as in \eqref{eqn:def:sharp_strat} with two layers separated by a pycnocline of thickness $\delta > 0$ falls within the framework of diffuse interfaces, in which one aims to study the case where a scalar quantity (e.g. density, volume fraction, concentration of some chemical) is not piecewise constant, with a sharp interface separating both regions, but rather varies continuously (see for instance the introduction in \cite{SibleyNoldKalliadasis2013} and references \nv{therein}). However, the stratified Euler equations that we consider in the present study are in the inviscid setting, and as such are different from the classical models studied in the diffuse interface framework. A large variety of such models strongly rely on diffusion of some quantity in the system, see for instance \cite{AbelsGarkePoiatti2024} (and references \nv{therein}) in which the authors study the Cahn-Hilliard and Allen-Cahn equations coupled with the Navier-Stokes equations. The sharp interface limit is considered for these models in \cite{SibleyNoldKalliadasis2013,AbelsFischerMoser2023,Schaubeck2014}. In the case where the scalar considered is the density, one can add capillarity effects for a compressible viscous fluid, described by the so-called Navier-Stokes-Korteweg equations. Such equations include both viscous and dispersive effects, and have been studied, alongside the sharp interface limit, in \cite{DreyerGiesselmannKraus2012,Daube2016}. When one neglects viscosity, the equations are the Euler-Korteweg equations, see \cite{BenzoniGavageDanchin2005,BenzoniGavageChiron2018,BenzoniGavage2011}. If these describe an inviscid fluid, they include regularizing dispersive effects, which are absent from the system \eqref{eqn:euler:isopycnal:Vb} at hand here. \\

The outline of the paper is as follows. In Section \ref{section:normal_modes}, we introduce the stratified Euler equations. We give a formulation in the so-called isopycnal coordinates, which, roughly speaking, aim at describing the level sets of the density (called isopycnals) rather than the density itself. Although this formulation is equivalent to the more standard Eulerian coordinates, it is well-suited to compare the solutions in the continuous and bilayer cases: in the latter situation, the interface separating the two layers can be thought of as an isopycnal. Such coordinates prove to be particularly useful in the non-linear setting, which is outside of the scope of the present study, but was considered in \cite{Adim2024} - in the hydrostatic setting and with an additional regularization. We then introduce a normal mode decomposition, adapted to the stratified Euler equations, and reformulate the equations into a system of coupled equations on each mode. We provide an energy estimate on this coupled system of evolution equations, and highlight the crucial role of the dispersion in this case. \\

In Section \ref{section:ebc}, we briefly describe the linear bilayer Euler equations, with and without a shear flow. We also give a simplified formulation in the irrotational case (or more precisely, when the vorticity is concentrated on the interface), which we use to perform simulations. Finally, we state a well-posedness result, in the absence of a shear flow.\\ 

In Section \ref{section:num_scheme}, we introduce our numerical strategy for solving the stratified Euler equations. Namely, we consider the (coupled) evolution \nv{of} a finite number of the aforementioned modes, as a vertical discretization. We then use truncated Fourier series as a horizontal discretization and a Runge-Kutta method for the time discretization.  We also explain how this discretization is particularly convenient to compute an approximation of the dispersion relation of the stratified Euler equations. We then prove in Section \ref{subsection:num_scheme:cvce} the convergence of the (semi-discrete) numerical scheme, and provide in Section \ref{subsection:illustrations} numerical illustrations, which in particular show how the well-known Saint-Andrew cross pattern for the propagation of internal waves is affected by the presence of a pycnocline of small but non-zero thickness. Finally, in Section \ref{subsection:num:bl}, we describe our numerical strategy for the linear, irrotational bilayer Euler equations.  \\

Section \ref{section:sharp_strat} is the core of our study. In Section \ref{subsection:sharp_strat:sans_shear}, we quantify the error between solutions of the stratified Euler equations (with $\Vb = 0$) and of the bilayer Euler equations \eqref{eqn:euler:bilayer}, in Sobolev spaces. The proof of the statements in this subsection are given in Appendix \ref{section:apdx:preuve}, and relies on energy estimates on the difference between solutions of both models. In order to study such a difference, isopycnal coordinates are crucial, as they transfer the natural framework in the bilayer case (i.e. the study of the evolution of the interface between both layers) to the stratified case, by shifting the focus from the density field to the isopycnals. Then we provide a numerical approximation of the dispersion relation of the stratified Euler equations with the profile $\Vb_{\delta}$ of Figure \ref{fig:shear:sharp}, in Section \ref{subsection:sharp_strat:shear}. More precisely, we plot the phase velocities associated to the first horizontal frequencies, and in particular highlight the presence of the Kelvin-Helmholtz instabilities.  Some consistency checks are provided in Appendix \ref{section:apdx:approx}. We then give some quantitative results on the features of the dispersion relation computed in the previous section. Overall, we find that the dispersion relations of the stratified Euler equations in the sharp stratification limit and the bilayer Euler equations, both with a shear flow, share several common features. In particular, the Kelvin-Helmholtz instabilities are present in the stratified Euler equations, and are well approximated by the instabilities present in the bilayer Euler equations. Finally in Section \ref{subsection:sharp_strat:shear:csq}, we point out the consequences of the presence of the Kelvin-Helmholtz instabilities in the stratified Euler equations, for the full justification of bilayer models:  \nv{indeed,} our findings suggest that we can find growing modes whose growth rate goes to infinity when $\delta$ goes to $0$. We discuss the implications for the full justification of the bilayer Euler equations, but also the bilayer shallow-water equations, from the stratified Euler equations, in the sharp stratification limit.\\

Eventually, in Section \ref{section:conclusion}, we summarize our findings and discuss some perspectives.  
 
\begin{figure}
\begin{minipage}[t]{0.49\textwidth}
\includegraphics[width=\textwidth]{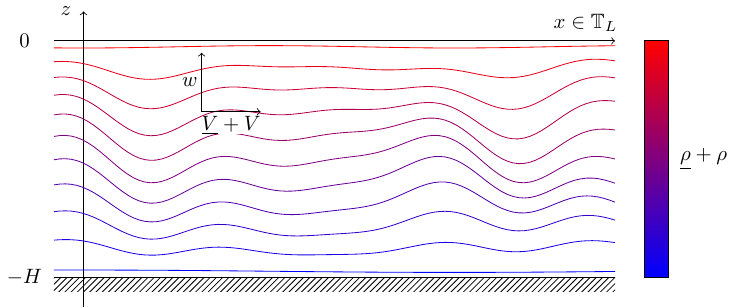}
\captionof{figure}{Idealized setting. Lines are isopycnals; their color represents the value of the density $\rhob + \rho$ of the fluid, which is a perturbation of the profile $\rhob$. The horizontal velocity is $\Vb + V$, a perturbation of the profile $\Vb$, and the vertical velocity is $w$.}
\label{fig:setting}
\end{minipage}
\begin{minipage}[t]{0.49\textwidth}
\includegraphics[width=\textwidth]{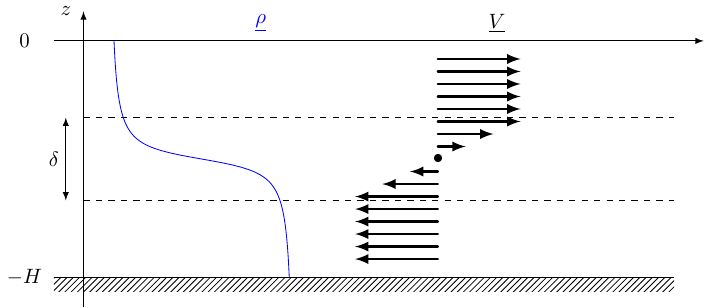}
\captionof{figure}{\label{fig:shear:sharp} The sharp stratification profile $\rhob$ from \eqref{eqn:def:sharp_strat} and a sharp shear flow $\Vb$, with parameter $\delta > 0$. When $\delta \to 0$, these converge towards piecewise constant functions.}
\end{minipage}
\end{figure}


\section{Notations}
\label{section:notations}
In this study we use the following notations.
\subsection{\nouveau{Functional setting}}
We denote by $L^2(\Omega)$ the space of functions that are square integrable on $\Omega$, for $\Omega = \mathbb T_L$ or $\Omega\coloneq\mathbb T_L \times [-H,0]$. It is equipped with the scalar product $\langle f,g \rangle \coloneq \int_{\Omega} fg$. We denote the associated norm by $|f|_{L^2(\mathbb T_L)}^2 \coloneq \langle f,f \rangle$ for $f$ defined on $\mathbb T_L$ and by $\Vert f\Vert_{L^2(\mathbb T_L \times [-H,0])}^2 \coloneq \langle f,f \rangle$ for $f$ defined on $\mathbb T_L \times [-H,0]$. We denote by $L^2_{\omega}$ the weighted $L^2$ space with weight $\omega$, i.e. with scalar product $\langle f,g\rangle_{L^2_{\omega}} \coloneq  \langle f,\omega g\rangle_{L^2}$.\\
Similarly, we denote by $\ell^2(\N)$ the space of sequences indexed on the space of non-negative integers, which are square integrable. It is equipped with the scalar product $\langle (u_n)_n, (v_n)_n \rangle \coloneq  \sum_{n = 0}^{\infty} u_n v_n$ and the associated norm $|(u_n)_n|_{\ell^2}$.\\
The same definitions go for $\ell^2(\N^*)$, the space of square integrable sequences indexed on positive integers. When no confusion is possible we simply write $\ell^2$ either for $\ell^2(\N)$ or $\ell^2(\N^*)$, and denote $(u_n)_{n \in \N}$ simply by $u$. For $s \in \N$, we denote by $H^s(\mathbb T_L)$ the $L^2$ based Sobolev space of regularity $s$, and $H^{s,k}(\mathbb T_L \times [-H,0])$ the anisotropic Sobolev space, i.e. the space of $L^2$ functions $f$ such that
$$\Vert f \Vert_{H^{s,k}}^2 \coloneq  \sum\limits_{\begin{aligned}k_1 + k_2 \leq s \\ k_2 \leq k\end{aligned}}  \Vert \partial_x^{k_1}\partial_z^{k_2}f \Vert_{L^2}^2 < +\infty.$$
We write $u \in \ell^2(\N^*,H^s)$ for $(u_n)_{n\in \N^*}$ such that, for any $n \in \N^*$, $u_n \in H^s$ and 
$$ \Vert u \Vert_{\ell^2 H^s}^2 \coloneq  \sum_n \vert u_n \vert_{H^s}^2 $$ is finite. Similarly, we write $W^{k,\infty}([-H,0])$ for the space of functions whose $k^{th}$ first derivatives are in $L^{\infty}([-H,0])$, and
$$\Vert f \Vert_{W^{k,\infty}} \coloneq \sum_{j=0}^k \Vert \partial_r^k f \Vert_{L^{\infty}}.$$
We use the notation 
$$f(\delta) \approx g(\delta)$$
to denote that there exists $C > 0$, which does not depend on $\delta$, and $\epsilon({\delta})$ with $\epsilon({\delta}) \underset{\delta \to 0}{\longrightarrow}0$ such that
$$f(\delta) = Cg(\delta)(1 + \epsilon(\delta)).$$
\nouveau{
\subsection{Variables and Parameters}
We write the numerical parameters that we use in the present study in the following table, alongside with a reference to their definition.\\

\noindent\begin{tabularx}{\textwidth}{|c|X|}
\hline 
$L,H$ & Length and depth of the strip $\mathbb T_L \times [-H,0]$ \\
\hline 
$\delta$ & Thickness of the pycnocline (see Figure \ref{fig:shear:sharp}, \eqref{eqn:def:sharp_strat}, \eqref{eqn:def:sharp:shear})\\
\hline
$\mu$ & Shallow-water parameter $\mu \coloneq H^2/L^2$ (see Section \ref{subsection:csq:blsw}) \\
\hline
$r_*$ & Depth of the pycnocline at rest (see Section \ref{section:sharp_strat})\\
\hline
$\mathtt{m_r}$ & Number of vertical points to compute the first vertical modes $(f_n)$ (see Section \ref{subsection:num_scheme_description})\\
\hline
$\mathtt N$ & Number of (vertical) modes $f_n$ used for the numerical scheme (see Section \ref{subsection:num_scheme_description})\\
\hline 
$\mathtt K$ & Number of (horizontal) Fourier modes used for the numerical scheme (see Section \ref{subsection:num_scheme_description})\\
\hline 
$\Delta_t$ & Time step (see Section \ref{subsection:num:bl})\\
\hline
\end{tabularx}
\par\vspace{1em}

\noindent We also write the notations for the vertical modes and coupling operators, which are not standard notations. \\

\noindent\begin{tabularx}{\textwidth}{|c|X|}
\hline 
$(f_n)_{n\in \N^*}$ & Orthonormal basis of $L^2_{\rhob N^2}([-H,0])$, obtained from a Sturm-Liouville problem (see \eqref{eqn:SL})\\
\hline
$(g_n)_{n\in \N}$ & Orthonormal basis of $L^2_{\rhob}([-H,0])$ (see \eqref{eqn:def:gn})\\
\hline
$M, A^{1}, \dots, A^{4}$ & Operators on $\ell^2$ defining the coupling for the evolution of the modes $(V_n, \eta_n)_{n \in \N^*}$ (see \eqref{eqn:def:M}-\eqref{eqn:def:Mnm}, \eqref{eqn:def:A})\\
\hline
$M^{(\mathtt N)}, A^{1,(\mathtt N)}, \dots, A^{4,(\mathtt N)}$ & Matrices used in the numerical scheme, defining the coupling for the evolution of the first modes $(V_n, \eta_n)_{n \leq \mathtt N}$ (see \eqref{eqn:def:Mnm:approx}, \eqref{eqn:def:A:approx})\\
\hline 
\end{tabularx}
}
\section{The linear stratified Euler equations, and its normal modes}
\label{section:normal_modes}
In this section we introduce the stratified Euler equations, as well as the isopycnal change of coordinates. This is done in Section \ref{subsection:isopyc}, where we write the stratified Euler equations linearized around a zero shear flow \eqref{eqn:euler:isopycnal:no_shear}. In Section \ref{subsection:normal_modes}, we write a normal mode decomposition that is adapted to the density profile $\rhob$ we consider. In order to find these normal modes, we first write the Taylor-Goldstein equation \eqref{eqn:taylor_goldstein:reduit}, which describes internal waves, and is derived directly from the stratified Euler equations. From the Taylor-Goldstein equation, a Sturm-Liouville problem arises naturally, and it yields the desired normal modes. Using this decomposition, the stratified Euler equations reduce to the system of coupled equations \eqref{eqn:modes:couple}, and we study the coupling operator in Proposition \ref{prop:M}. Then in Section \ref{subsection:extension_shear}, we provide the same modal decomposition for the stratified Euler equations linearized around a shear flow $\Vb$, so that it also reduces to a system of coupled equations. We provide an energy estimate on this system in Proposition \ref{prop:nrj:boussinesq:couple:shear}, highlighting in particular the role of the dispersion in the stability of the system.  
\subsection{The stratified Euler equations in isopycnal coordinates}
\label{subsection:isopyc}
In this section we write the Euler equations in our context, and explain the isopycnal change of coordinates. 
We decompose the density of the fluid as 
$$\rho_{\tot}(t,x,z) \coloneq  \rhob(z)+\rho(t,x,z),$$
where $t \in [0,T]$ (for some $T > 0$) denotes the time variable, $x \in \mathbb T_L$ is the horizontal space variable in the 1D torus of length $2\pi L$, and $z \in [-H,0]$ is the vertical space variable, with $H > 0$ the depth of the strip. The fluid domain is assumed to be the strip $\mathbb T_L \times [-H,0]$, so that in particular we consider a flat bottom and make the rigid lid approximation. The function $\rhob$ only depends on $z$ and is the density at equilibrium of the fluid at hand.\\
Then, the total density of the fluid $\rho_{\tot}$ is entirely described by the deviation $\rho$ of the density from the equilibrium $\rhob$. In particular, $\rho$ might depend on $t,x$ and $z$. 
We write $V$ and $w$ for the horizontal and vertical velocities of the fluid respectively; the notations and the setting are summarized in Figure \ref{fig:setting} (although in the slightly more general case where we consider a non-zero shear flow $\Vb$, see Section \ref{subsection:extension_shear}).
Then, the (non-linear) Euler equations in this context read
\begin{equation}
\label{eqn:euler:euleriennes}
\sys{\partial_t V + V\partial_x V +  w \partial_z V  + \frac{1}{\rhob + \rho} \partial_x P &= 0,\\
	 \partial_t w + V \partial_x w + w \partial_z w  + \frac{1}{\rhob + \rho} \partial_z P + \frac{g\rho}{\rhob + \rho} &= 0,\\
	\partial_t \rho + V \partial_x \rho + w \partial_z \rho + w \rhob' &= 0,\\
	\partial_x V + \partial_z w &=0,} \qquad \text{in } (0,T) \times \mathbb T_L \times (-H,0);
\end{equation}
here, $P$ denotes the pressure of the fluid. It is understood as the Lagrange multiplier of the divergence-free constraint, and as such it can be expressed from the unknowns $V,w,\rho$ through an elliptic equation, which we do not write here. We refer to Appendix \ref{section:apdx:preuve} for its precise definition, in the linear setting. The system \eqref{eqn:euler:euleriennes} is completed with the boundary conditions
\begin{equation}
\label{eqn:euler:euleriennes:bc}
\begin{aligned}
w_{|z=-H}&=0,  \qquad w_{|z=0} &= 0, \\
\rho_{|z=-H} &= 0, \qquad \rho_{|z=0} &= 0.
\end{aligned}
\end{equation}
Throughout the present study, the density profile $\rhob$ is assumed to be smooth and strictly decreasing, more precisely, there exists $c_* > 0$ such that 
\begin{equation}
\label{eqn:stable}
- \rhob' \geq c_*,
\end{equation}
where the prime denotes the (vertical) derivative. Such an assumption is standard for the study of oceanic flows at geophysical scales (see \cite{Vallis2017}), as it prevents the Rayleigh-Taylor instability - which takes place at too small time and space scales to be relevant for the large-scale circulation. Using such an assumption, the well-posedness of \eqref{eqn:euler:euleriennes} together with the boundary conditions \eqref{eqn:euler:euleriennes:bc} and suitable initial conditions in Sobolev spaces is studied in \cite{Desjardins2019} and \cite{Fradin2024}.\\
\nouveau{
\begin{remark}
The variables $t,x,$ and $z$, the unknowns $V,w,\rho,$ and $P$, the density profile $\rhob$ (and a velocity profile $\Vb$ that we introduce in Section \ref{subsection:extension_shear}), and the parameters $H$ and $L$ are all dimensionless quantities.
The characteristic scales used to nondimensionalize these quantities depend on the physical situation under consideration, so that working entirely in these dimensionless quantities allows us to keep the discussion as general as possible. The same discussion goes for the bilayer Euler equations, which we introduce in Section \ref{section:ebc}.
\end{remark}
}

\nouveau{
\begin{remark}
The (non-linear) stratified Euler equations \eqref{eqn:euler:euleriennes} together with the boundary conditions \eqref{eqn:euler:euleriennes:bc} formally admit a Hamiltonian formulation, see for instance \cite{Benjamin1985}. In particular, if $(V,w,\rho)$ is a smooth solution of \eqref{eqn:euler:euleriennes}-\eqref{eqn:euler:euleriennes:bc}, we can define the energy
$$\mathcal E(t) \coloneq \frac12 \int_{\mathbb T_L \times [-H,0]}\left( (\rhob(z) + \rho(t,x,z))(|V(t,x,z)|^2 + |w(t,x,z)|^2) + \frac{g}{- \rhob'(z)} |\rho(t,x,z)|^2\right)\mathrm dx \mathrm dz,$$
which is conserved, i.e. for any $t \geq 0$ such that the solution $(V,w,\rho)$ is defined, we get
$$\mathcal E(t) = \mathcal E(0).$$
Although not sufficient to prove the well-posedness of \eqref{eqn:euler:euleriennes}-\eqref{eqn:euler:euleriennes:bc}, this property is at the core of the so-called energy method; see for instance \cite{Desjardins2019,Fradin2024} for an application of the energy method to \eqref{eqn:euler:euleriennes}-\eqref{eqn:euler:euleriennes:bc}, both in the 2D and 3D case.
\end{remark}}
Let us linearize \eqref{eqn:euler:euleriennes} around the trivial equilibrium $(V,w,\rho,P) = (0,0,0,0)$, and apply the so-called isopycnal change of coordinates, which we now describe. We assume that the total density is stably stratified, that is 
\begin{equation}
\label{eqn:stable_tot}
-\partial_z(\rhob + \rho) \geq c_* > 0.
\end{equation}
Thus, given $(t,x) \in [0,T]\times \mathbb T_L$, the function $z \mapsto \rhob(z) + \rho(t,x,z)$ admits an inverse, so that we can choose $\eta$ depending on $t,x$ and $r \in [-H,0]$ such that
\begin{equation}
\label{eqn:def:eta}
\begin{aligned}
\rhob(r + \eta(t,x,r)) + \rho(t,x,r + \eta(t,x,r)) &= \rhob(r);
\end{aligned}
\end{equation}
the precise form $r + \eta(t,x,r)$ of the inverse of $\rhob + \rho$ follows from the fact that $\rhob$ only depends on $r$, see \cite{Fradin2024} for more details. The linearized Euler equations written in isopycnal coordinates then read
\begin{equation}
\label{eqn:euler:isopycnal:no_shear}
\sys{\partial_t V + \frac{1}{\rhob} \partial_x P &= 0,\\
	\partial_t w + \frac{1}{\rhob} \partial_r P - \frac{g\rhob'}{\rhob} \eta &= 0,\\
	\partial_t \eta - w &= 0,\\
	\partial_x V + \partial_r w &=0,} \qquad \text{in } (0,T) \times S;
\end{equation}
where $S \coloneq  \mathbb T_L \times (-H,0)$ is the domain of the fluid in isopycnal coordinates. We refer once again to \cite{Fradin2024} for details on the linearization process and the isopycnal change of coordinates. The isopycnal coordinates are well-suited to study the sharp stratification limit, see Section \ref{section:sharp_strat}.
Finally, the system \eqref{eqn:euler:isopycnal:no_shear} is completed with the boundary conditions
\begin{equation}
\label{eqn:euler:iso:bc}
\sys{&w_{|r=-H} = w_{|r=0} = 0, \\
&\eta_{|r=-H} = \eta_{|r=0} = 0,}
\end{equation}
stemming from the impermeability condition at the flat surface and flat bottom, and the assumption that the surface and the bottom are flat and isopycnals (i.e. level-sets of the density $\rhob + \rho$). Both assumptions are propagated by the flow. Finally, we impose initial conditions
\begin{equation}
\label{eqn:euler:iso:ci}
\sys{&V_{|t=0} = V_{\ini}, \qquad & w_{|t=0} = w_{\ini}, \\
&\eta_{|t=0} = \eta_{\ini},}
\end{equation}
that satisfy the incompressibility condition in \eqref{eqn:euler:isopycnal:no_shear} and the boundary conditions \eqref{eqn:euler:iso:bc}. The well-posedness of the system \eqref{eqn:euler:isopycnal:no_shear} together with \eqref{eqn:euler:iso:bc} and initial conditions \eqref{eqn:euler:iso:ci} in suitable Sobolev spaces is studied in Appendix \ref{section:apdx:preuve} (see Proposition \ref{prop:nrj}). See also  \cite{Fradin2024} in the non-linear setting and \cite{Duchene2022} for a similar study with an additional diffusion term.

\nouveau{
\begin{remark}
\label{rk:nrj:stab}
For $(V,w,\eta)$ a smooth solution to \eqref{eqn:euler:isopycnal:no_shear}-\eqref{eqn:euler:iso:bc}, the energy
$$\mathcal E(t) \coloneq \int_{\mathbb T_L \times [-H,0]} \frac12 \left( \rhob(r)(|V(t,x,r)|^2 + |w(t,x,r)|^2) + g|\rhob'(r)| |\eta|^2(t,x,r)\right)\mathrm dx \mathrm dr$$
is conserved, i.e. $\mathcal E(t) = \mathcal E(0)$, for $t\geq 0$. In particular, there is no unstable mode, in the following sense: a solution to \eqref{eqn:euler:isopycnal:no_shear}-\eqref{eqn:euler:iso:bc} of the form
$$ (V,w,\eta)(t,x,r) = e^{-i \omega t}(V_0(x,r),w_0(x,r),\rho_0(x,r))$$
satisfies $\Im(\omega)=0$. See also Definition \ref{def:dispersion_relation}.
\end{remark}
}

\subsection{Normal modes} 
\label{subsection:normal_modes}
Given a solution $(V,w,\eta)$ to \eqref{eqn:euler:isopycnal:no_shear}, simple algebraic manipulations show that $w$ satisfies the Taylor-Goldstein equation
\begin{equation}
\label{eqn:taylor_goldstein:reduit}
\partial_t^2 \left(\frac{1}{\rhob}\partial_r(\rhob \partial_r w) + \partial_x^2 w \right) + N^2 \partial_x^2 w =0,
\end{equation} 
where the quantity 
\begin{equation}
\label{eqn:def:N2}
N^2 \coloneq  - g \frac{\rhob'}{\rhob}
\end{equation}
is the (square of the) Brunt-Väisälä frequency. Note that $N^2$ depends on $r$ a priori. We reproduce the derivation of \eqref{eqn:taylor_goldstein:reduit} in Appendix \ref{section:apdx:derivation}, for the sake of completeness. \\

Complemented with Dirichlet boundary conditions, \eqref{eqn:taylor_goldstein:reduit} is the starting point of the study of internal waves. However, solving \eqref{eqn:taylor_goldstein:reduit} is a difficult task in general. Roughly speaking, a standard strategy (see for instance \cite{Desjardins2019}) is to start with the study of the first and last terms in \eqref{eqn:taylor_goldstein:reduit}. More precisely, we can define $(c_n,f_n)_{n\geq1}$ as the eigenvalues and eigenfunctions of the following Sturm-Liouville problem
\begin{equation}
\label{eqn:SL}
(\rhob f_n')' = - \rhob N^2 \frac{1}{c_n^2} f_n \qquad \text{ in } (-H,0),
\end{equation}
with Dirichlet boundary conditions at $r=-H$ and $r=0$. We also define $(g_n)_{n\geq 0}$ by
\begin{equation}
\label{eqn:def:gn}
g_n \coloneq  \sys{ &c_n f'_n \qquad &\text{ if  }  n\geq 1,\\
			&\left(\int_{-H}^0 \frac{1}{\rhob}\right)^{-\frac12} \frac{1}{\rhob}  \qquad &\text{ if } n = 0. } 
\end{equation}

There is a vast literature studying the modes for internal waves, so that we only cite a few studies that are close to our framework, and refer to them for further references. In \cite{Banerjee2005}, the author studies the extension of \eqref{eqn:taylor_goldstein:reduit} to the case of a non-zero shear flow. In \cite[Chapter 5]{GerkemaZimmerman2008} the authors describe the modes from the Sturm-Liouville problem \eqref{eqn:SL} in several configurations. The modal decomposition from eigenfunctions of the extension of \eqref{eqn:taylor_goldstein:reduit} to the case of a non-zero shear flow, or from the simpler \eqref{eqn:SL} can give an insight on numerous features of the ocean, see for instance stability considerations in \cite{Jones1967,BanksDrazinZaturska1976, BarrosVoloch2020}, turbulence theory in \cite{fuFlierl1980} and an extension to the non-flat bottom case in \cite{LaCasceGroeskamp2020}. See also the simulations from \cite[Chapter 5]{GerkemaZimmerman2008} that use the modal decomposition from the Sturm-Liouville problem \eqref{eqn:SL}, as is the case of the present study, although they restrict to specific stratification profiles (piecewise constant or linear), but also consider Coriolis forcing.\\

We sum up the basic properties of $(f_n)_n, (g_n)_n, (c_n)_n$ in the following lemma. 
\begin{lemma}
\label{lemma:SL}
Assume that there exists $c_* > 0$ such that $\rhob \geq c_*$, and such that \eqref{eqn:stable} holds. Let $(f_n)_{n\geq1}, (g_n)_{n\geq 0}, (c_n)_{n\geq1}$ be defined from \eqref{eqn:SL} and \eqref{eqn:def:gn} as above. In particular, $(c_n)_{n\geq1}$ is a decreasing sequence by convention.
\begin{itemize}
\item The family $(f_n)_{n\geq 1}$ is an orthonormal basis of $L^2_{\rhob N^2}([-H,0])$, the weighted $L^2$-space with weight $\rhob N^2$.
\item There exists $C > 0$ a constant that does not depend on $n$, such that, for all $n \in \N^*$:
	\begin{equation}
	\label{eqn:fn:borne_sup}
		\vert f_n \vert_{L^{\infty}([-H,0])} \leq C.
	\end{equation}
\item The family $(g_n)_{n\geq 0}$ is an orthonormal basis of $L^2_{\rhob}([-H,0])$.
\item For $n \in \N^*$, $c_n > 0$, and we have the asymptotic behavior for $(c_n)_n$, for some constant $c > 0$ 
\begin{equation}
\label{eqn:c_n:behaviour}
n c_n \underset{n \to \infty}{\longrightarrow} c.
\end{equation}
In particular, the series $\sum_{n \geq 1} c_n^2$ converges.
\end{itemize}
\end{lemma}
\begin{proof}
The proof of Lemma \ref{lemma:SL} relies on the standard Sturm-Liouville theory, see for instance \cite{Desjardins2019} for the present context, and \cite{Atkinson1987} for \eqref{eqn:c_n:behaviour}. We only write the proof of \eqref{eqn:fn:borne_sup} for the sake of completeness. Let $n \in \N^*$. Recall that $f_n$ is a smooth function, as it is the solution of a (regular)  ODE, so that the following computations make sense. Multiplying \eqref{eqn:SL} by $f_n$, integrating over $[-H,0]$, and using $\vert f_n \vert_{L^2_{\rhob N^2}} = 1$,  we get
\begin{equation}
\label{eqn:SL:temp:1}
\vert f'_n \vert_{L^2_{\rhob}}^2 = \frac{1}{c_n^2}.
\end{equation}
Now multiplying \eqref{eqn:SL} by $\rhob f_n'$, we get
\begin{equation}
\label{eqn:SL:temp:2}
\left( (\rhob f'_n)^2\right)' + \frac{\rhob^2 N^2}{c_n^2} (f_n^2)' = 0.
\end{equation}
Let $r \in (-H,0)$. We integrate \eqref{eqn:SL:temp:2} over $[-H,r]$ to get
\begin{equation}
\label{eqn:SL:temp:3}
(\rhob f'_n)^2(r) + \frac{\rhob^2(r) N^2(r)}{c_n^2} f_n^2(r) = (\rhob f'_n)^2(-H) + \frac{1}{c_n^2} \int_{-H}^r \left(\rhob^2(r) N^2(r)\right)' f_n^2(r) dr.
\end{equation}
Because $\rhob$ and $N^2$ are both bounded from above and below by positive constants and $\vert f_n \vert_{L^2_{\rhob N^2}} = 1$, \eqref{eqn:SL:temp:3} implies
\begin{equation}
\label{eqn:SL:temp:4}
f_n^2(r) \leq C (1+c_n^2 (\rhob f'_n)^2(-H)),
\end{equation}
for some constant $C > 0$. We now bound the second term in \eqref{eqn:SL:temp:4} from above uniformly in $n$. We integrate \eqref{eqn:SL:temp:3} from $-H$ to $0$, and use \eqref{eqn:SL:temp:1} and $\vert f_n \vert_{L^2_{\rhob N^2}} = 1$ to get
\begin{equation}
\label{eqn:SL:temp:5}
(\rhob f'_n)^2(-H) \leq \frac{C}{c_n^2}.
\end{equation}
Combining \eqref{eqn:SL:temp:4} and \eqref{eqn:SL:temp:5}, we get \eqref{eqn:fn:borne_sup}.
\end{proof}
We now decompose $V,P,w,\eta$ along these bases, and write $(V_n)_{n\in \N},(P_n)_{n\in \N}, (w_n)_{n\in \N^*}, (\eta_n)_{n\in \N^*}$ for their coefficients, respectively. Plugging these decompositions into \eqref{eqn:euler:isopycnal:no_shear}, small computations (postponed to Appendix \ref{section:apdx:derivation}) show that the coefficients of a solution $(V,w,\eta,P)$ to \eqref{eqn:euler:isopycnal:no_shear} satisfy, for $n \geq 1$:
\begin{equation}
\label{eqn:modes:couple}
\sys{ \left((I- M \partial^2_x) \partial_t V\right)_n + c_n \partial_x \eta_n &= 0,\\
	\partial_t \eta_n + c_n \partial_x V_n &=0,} \qquad \text{ in } (0,T) \times \mathbb T_L,
\end{equation}
where I is the identity operator on $\ell^2(\N^*)$, and $M$ is the linear operator on $\ell^2(\N^*)$ defined by
\begin{equation}
\label{eqn:def:M}
\begin{aligned}
M :  \ell^2(\N^*) &\longrightarrow \ell^2(\N^*) \\
	 (V_n)_{n\in \N^*} &\longmapsto \left(\sum_{m \in \N^*}M_{n,m} V_m\right)_{n \in \N^*},
\end{aligned}
\end{equation}
with, for $(n,m) \in (\N^*)^2$,
\begin{equation}
\label{eqn:def:Mnm}
M_{n,m} \coloneq  c_n c_m \int_{-H}^{0} f_n(r) f_m(r) \rhob(r) dr. 
\end{equation}
\begin{remark}
 Note that \eqref{eqn:modes:couple} is a closed system of equations on the coefficients on $V$ and $\eta$, and we do not need to solve for the coefficients of $P$ and $w$. However, if of interest, the coefficients  of $P$ and $w$ can be recovered from those of $V$ and $\eta$ from the elliptic equation satisfied by the pressure (see \eqref{eqn:elliptic:1}) and the divergence-free condition, respectively.
\end{remark}
\begin{remark}
\label{rk:coef_zero}
We do not write an equation for the evolution of $V_0$. Indeed, taking the scalar product with weight $\rhob$ of the divergence-free condition in \eqref{eqn:euler:isopycnal:no_shear} with $g_0$ yields
$$\partial_x V_0 = 0.$$
Then, taking the scalar product with weight $\rhob$ of the first equation in \eqref{eqn:euler:isopycnal:no_shear} with $g_0$ and integrating over $\mathbb{T}_L$, we get
$$ \partial_t V_0 = 0,$$
so that $V_0$ is a constant, and therefore is prescribed by the initial data. Conversely, the value of $V_0$ does not play any role in the evolution of the other modes, given by \eqref{eqn:modes:couple}.
\end{remark}
Basic properties of the operator $M$ are summarized in the following proposition.
\begin{proposition}
\label{prop:M}
Assume that there exist constants $c^*, c_* > 0$ such that $c_* \leq \rhob \leq c^*$ and $-\rhob' \geq c_*$. Then there exists $C > 0$ such that the following holds. The operator $M$ is a linear, bounded operator on $\ell^2$, and for $(u_n)_n \in \ell^2$
\begin{equation}
\label{eqn:M:bounded}
\vert M (u_n)_n \vert_{\ell^2} \leq C \vert (c_n u_n)_n \vert_{\ell^2}. 
\end{equation}
The operator $M$ is symmetric, i.e. for $(u_n)_n, (v_n)_n \in \ell^2$
\begin{equation}
\label{eqn:M:sym}
\langle M (u_n)_n, (v_n)_n\rangle = \langle (u_n)_n, M (v_n)_n \rangle.
\end{equation}
The operator $M$ has an ellipticity property; more precisely, for $(u_n)_n \in \ell^2$, 
\begin{equation}
\label{eqn:M:positive}
\langle M (u_n)_n, (u_n)_n \rangle \geq  \frac{1}{C} \vert c_n u_n \vert_{\ell^2}^2.
\end{equation}
As a consequence, it admits a unique positive square root denoted by $M^{\frac12}$, i.e. such that $(M^{\frac12})^2 = M$, which satisfies
\begin{equation}
\label{eqn:M:sqrt:positive}
\frac{1}{C} \vert (c_n u_n)_n \vert_{\ell^2} \leq \vert M^{\frac12} u \vert_{\ell^2} \leq  C \vert (c_n u_n)_n \vert_{\ell^2}.
\end{equation}
Finally, in the special case where $N^2$ does not depend on the vertical variable $r$, then $M$ is a diagonal operator, namely
\begin{equation}
\label{eqn:M:diag}
\begin{aligned}
M_{n,m} &= \frac{c_n^2}{N^2} \qquad \text{ if } n = m,\\
M_{n,m} &= 0 \qquad \text{ if } n \neq m.
\end{aligned}
\end{equation}
\end{proposition}
In the case where $N^2$ is independent of $r$, \eqref{eqn:M:diag} together with \eqref{eqn:modes:couple} show that the evolution of the modes $(V_n,\eta_n)_n$ is decoupled. This case has been well-studied, see for instance \cite[Chapter 7]{Vallis2017} and references within.  In particular, internal waves propagate in this case according to the Saint Andrew's cross pattern, see Figure \ref{fig:saint_andrew}. However, for a general stratification profile $\rhob$, the Brunt-Väisälä frequency $N^2$ depends on $r$, and the operator $M$ is not diagonal. Thus the evolution of the modes $((V_n,\eta_n))_{n}$ through \eqref{eqn:modes:couple} is coupled. This coupling is referred to as a dispersive coupling, see \cite{Desjardins2019}. 
\begin{proof}
For $n,m \in \N^*,$ the definition \eqref{eqn:def:Mnm} of $M_{n,m}$ and the fact that $(f_n)_n$ is an orthonormal basis of $L^2_{\rhob N^2}$ yield
\begin{equation}
\label{eqn:M:borne_coefs}
|M_{n,m}| \leq \left\vert \frac{1}{N^2} \right\vert_{L^{\infty}([-H,0])} c_n c_m.
\end{equation}
Thus, there exists a constant $ C_0 > 0$ such that, for any $n \in \N^*$:
$$|(M u)_n|^2 \leq \sum_m |M_{n,m} u_m|^2 \leq C_0 c_n^2 |(c_m u_m)_{m\in \N^*}|_{\ell^2}^2. $$
Because of \eqref{eqn:c_n:behaviour}, the series $\sum_{n\in \N^*} c_n^2$ converges, and thus there exists $C > 0$ such that \eqref{eqn:M:bounded} holds.\\
\nv{The symmetry property} \eqref{eqn:M:sym} is an immediate consequence of \eqref{eqn:def:Mnm}. For \eqref{eqn:M:positive}, recall that, for $n \in \N^*$, $|f_n|$ is bounded uniformly in $n$, i.e. \eqref{eqn:fn:borne_sup} holds. Then, we write, for $u \in \ell^2$, 
\begin{equation}
\label{eqn:M:temp}
\begin{aligned}
\left\langle Mu,u \right\rangle &= \sum_{m,n} c_n c_m u_n u_m \int_{-H}^0 f_n f_m \rhob \\
&= \sum_{n} c_nu_n   \sum_{m} \int_{-H}^0 c_mu_m f_n f_m \rhob \\
&= \sum_{n} c_nu_n  \int_{-H}^0 \sum_{m} c_m  u_m f_n f_m \rhob, \\
\end{aligned}
\end{equation}
where we used the dominated convergence theorem for the last equality, since $f_n$ and $f_m$ are bounded uniformly in $n$ and $m$, and $\sum_m |u_m c_m|$ converges. The same computations yield  
\begin{equation}
\label{eqn:M:temp:2}
\begin{aligned}
\sum_{n} c_nu_n   \int_{-H}^0 f_n\sum_{m}  c_mu_m  f_m \rhob &=   \int_{-H}^0 \left(\sum_{m}  c_m u_m f_m\right)^2 \rhob. \\
\end{aligned}
\end{equation}
Using the definition \eqref{eqn:def:N2}, we get $N^2 \geq c_*/c^* > 0$, so that we can write
\begin{equation}
\label{eqn:M:temp:3}
\begin{aligned}
\int_{-H}^0 \left(\sum_{m}  c_m u_m f_m\right)^2 \rhob  &= \int_{-H}^0 \left(\sum_{m}  c_m u_m f_m\right)^2 \frac{\rhob N^2}{ N^2} \\
&\geq c_*/c^* \sum_{m,n} c_n c_m u_n u_m \int_{-H}^0 f_n f_m \rhob N^2\\
&\geq c_*/c^* \vert c_n u_n \vert_{\ell^2}^2
\end{aligned}
\end{equation}
where we used once again the dominated convergence theorem and the fact that $(f_n)_n$ is an orthonormal basis of $L^2_{\rhob N^2}$. This yields \eqref{eqn:M:positive}. The existence and uniqueness of a positive square root for a bounded positive operator on a Hilbert space is standard, see for instance \cite[Problem 121]{Halmos1982}. \\
 When $N^2$ is constant (and positive), \eqref{eqn:M:diag} follows from \eqref{eqn:def:Mnm} and the fact that $(f_n)_n$ is an orthonormal basis of $L^2_{\rhob N^2}$. 
\end{proof}
\subsection{Extension to a non-zero shear flow}
\label{subsection:extension_shear}
We now write the stratified Euler equations in isopycnal coordinates, linearized around a shear flow. A shear flow is an equilibrium of the stratified Euler equations \eqref{eqn:euler:euleriennes} with $V = \Vb$ that only depends on the vertical variable. In the case $\Vb = 0$, one finds the system \eqref{eqn:euler:isopycnal:no_shear}, so that the present section is a generalization of the previous one. We assume that $\Vb$ is a smooth function of $r$ only; the Euler equations in isopycnal coordinates, linearized around a shear flow read
\begin{equation}
\label{eqn:euler:isopycnal:Vb}
\sys{\partial_t V + \Vb \partial_x V +  w \Vb' + \frac{1}{\rhob} \partial_x P &= 0,\\
	\partial_t w + \Vb \partial_x w  + \frac{1}{\rhob} \partial_r P - \frac{g\rhob'}{\rhob} \eta &= 0,\\
	\partial_t \eta + \Vb \partial_x \eta - w &= 0,\\
	\partial_x V + \partial_r w &=0,} \qquad \text{in } [0,T] \times S,
\end{equation}
completed with the boundary conditions \eqref{eqn:euler:iso:bc}. 
\nouveau{
\begin{remark}
Recall the energy $\mathcal E$ defined in Remark \ref{rk:nrj:stab}. One can check that, for a smooth solution of \eqref{eqn:euler:isopycnal:Vb} with boundary conditions \eqref{eqn:euler:iso:bc}, the energy is not conserved:
$$\frac{d}{dt} \mathcal E(t) = -\int_{\mathbb T_L \times [-H,0]} \rhob \Vb' w V.$$
Hence, in the presence of a shear flow, the equations \eqref{eqn:euler:isopycnal:Vb}-\eqref{eqn:euler:iso:bc} may admit unstable modes (in the sense of Remark \ref{rk:nrj:stab}, see also Definition \ref{def:dispersion_relation}). The purpose of Section \ref{subsection:sharp_strat:shear} is to provide a description of these instabilities, in the sharp stratification limit, through a numerical investigation. 
\end{remark}
}

From \eqref{eqn:euler:isopycnal:Vb} one can derive the Taylor-Goldstein equation, which is a generalization of \eqref{eqn:taylor_goldstein:reduit} in the presence of a shear flow:
\begin{equation}
\label{eqn:taylor_goldstein}
(\partial_t + \Vb \partial_x)^2 \left( \frac{1}{\rhob}\partial_r(\rhob \partial_r w) + \partial_x^2 w \right) - (\partial_t + \Vb \partial_x) \frac{1}{\rhob}(\rhob \Vb')' \partial_x w + N^2 \partial_x^2 w =0;
\end{equation}
we postpone these computations to Appendix \ref{section:apdx:derivation}. The boundary conditions \eqref{eqn:euler:iso:bc} on $w$ complete this equation.\\
As for \eqref{eqn:taylor_goldstein:reduit}, the equation \eqref{eqn:taylor_goldstein} contains the Sturm-Liouville operator $\frac{1}{\rhob} \partial_r \left(\rhob \partial_r \cdot \right)$. We thus perform the same modal decomposition as in Section \ref{subsection:normal_modes}; the computations are postponed to Appendix \ref{section:apdx:derivation}. The evolution of the modes $((V_n,\eta_n))_n$ is coupled, and is given by the following system, for $n \geq 1$:
\begin{equation}
\label{eqn:modes:couple:shear}
\sys{ &((I- M \partial^2_x) \partial_t V)_n + c_n \partial_x \eta_n + \left((2 A^1 + A^2) \partial_x V\right)_n  - (A^3 \partial_x^3 V)_n = 0,\\
	&\partial_t \eta_n + c_n \partial_x V_n +  (A^4 \partial_x \eta)_n =0,}
\end{equation}
where $M$ is defined through \eqref{eqn:def:M}-\eqref{eqn:def:Mnm} and $A^1,A^2,A^3,A^4$ are defined from their coefficients through
\begin{equation}
\label{eqn:def:A}
\begin{aligned}
A^1_{n,m} &\coloneqq  \int_{-H}^0 \Vb(r) g_n(r)g_m(r) \rhob(r)dr,
&
A^2_{n,m} &\coloneqq -\frac{c_m}{c_n}\int_{-H}^0 \Vb(r) f_n(r)f_m(r) \rhob(r)N^2(r)dr, \\
A^3_{n,m} &\coloneqq  c_n c_m\int_{-H}^0 \Vb(r) f_n(r)f_m(r) \rhob(r)dr, 
&
A^4_{n,m} &\coloneqq  \int_{-H}^0 \Vb(r) f_n(r)f_m(r) \rhob(r)N^2(r)dr.
\end{aligned}
\end{equation}

\begin{remark}
\label{rk:A2}
Notice that, with an integration by parts and \eqref{eqn:SL}, we can write
\begin{equation}
\label{eqn:A1plusA2}
A^1_{n,m} + A^2_{n,m} = - c_m\int_{-H}^0 \Vb' c_n f_n' f_m \rhob,
\end{equation}
which yields, thanks to Cauchy-Schwarz inequality and the fact that $(c_n f'_n)_n$ and $(f_n)_n$ are orthonormal families for some weighted scalar product, that there exists a constant $C > 0$ such that
\begin{equation}
\label{eqn:Aun_Adeux}
|A^{1}_{n,m} + A^{2}_{n,m} | \leq C c_m.
\end{equation}
We use this property in Proposition \ref{prop:nrj:boussinesq:couple:shear}. Note also that it does not hold for $A^{2}_{n,m}$ when considered alone.\\

Owing to \eqref{eqn:A1plusA2}, we obtain an alternative formula for the first equation in \eqref{eqn:modes:couple:shear}:
\begin{equation}
\label{eqn:modes:couple:shear:V2}
((I- M \partial^2_x) \partial_t V)_n + c_n \partial_x \eta_n + ( (A^1 + \tilde{A}^2) \partial_x V)_n  - (A^3 \partial_x^3 V)_n = 0,
\end{equation}
where $\tilde{A}_2$ is defined through 
\begin{equation}
\label{eqn:def:Atilde}\tilde{A}^2_{n,m} = A^1_{n,m} + A^2_{n,m} = - c_m\int_{-H}^0 \Vb' c_n f_n' f_m \rhob.
\end{equation}
The advantage of $\tilde{A}^2_{n,m}$ over $A^2_{n,m}$ is that the terms the terms in \eqref{eqn:modes:couple:shear:V2} are either skew-symmetric or bounded (see the proof of Proposition \ref{prop:nrj:boussinesq:couple:shear} for more details), so that we can derive suitable energy estimates. Note however that \eqref{eqn:def:Atilde} requires that $\Vb'$ is bounded, which is not the case uniformly to the small parameter $\delta$ that we introduce in Section \ref{section:sharp_strat}. \\

The advantage of $A^2_{n,m}$ over $\tilde{A}^2_{n,m}$ is that it can be computed from $A^{4}_{n,m}$, up to a multiplication by the scalar $c_m/c_n$. Thus, computing $A^{1}_{n,m}, A^{2}_{n,m}, A^{3}_{n,m}, A^{4}_{n,m}$ requires to compute three integrals, instead of four if one replaces $A^{2}_{n,m}$ by $\tilde{A}^2_{n,m}$. From the computational perspective, this shortens the time required to compute these coefficients, which happens to be significant in the simulations we show in Section \ref{subsubsection:sharp_strat:shear:quantitative}.
\end{remark}
\begin{remark}
\label{rk:coef_zero:shear}
We do not write an evolution equation on $V_0$. The same computations as in Remark \ref{rk:coef_zero} yield
\begin{equation}
\label{eqn:rk:zero_shear:temp}
\sys{\partial_x V_0 &= 0,\\
		\partial_t V_0 + \partial_x P_0 + \int_{-H}^0 \left(\partial_x(\Vb V) + w\Vb' g_0\rhob\right)  &=0.}
\end{equation}
Now \eqref{eqn:def:gn} \nv{implies that} $g_0 \rhob$ is a constant. Hence, integrating by parts in the last term of the integral in \eqref{eqn:rk:zero_shear:temp} and using the divergence-free condition in \eqref{eqn:euler:isopycnal:Vb} and the boundary conditions \eqref{eqn:euler:iso:bc} yields
$$\partial_t V_0 + \partial_x \left(P_0 + \int_{-H}^0 \Vb V (1+g_0 \rhob) \right) =0. $$
As $V_0$ does not depend on $x$, integrating the previous equation on $\mathbb T_L$ yields $\partial_t V_0 = 0$.  Therefore, as in the previous section, $V_0$ is prescribed by the initial data. Conversely, the value of $V_0$ does not play any role in the evolution of the other modes, given by \eqref{eqn:modes:couple:shear}.
\end{remark}

We conclude this section with an existence and stability result on the equations \eqref{eqn:modes:couple:shear} satisfied by the modes $((V_n,\eta_n))_{n\in \N^*}$.
\begin{proposition}
\label{prop:nrj:boussinesq:couple:shear}
Let $s \geq 0$, $\rhob, \Vb \in W^{1,\infty}([-H,0])$ such that there exist constants $c^*, c_* > 0$ such \mbox{that $c_* \leq \rhob \leq c^*$} and $-\rhob' \geq c_*$. Then there exists $C > 0$ such that the following holds. \mbox{Let $(V_{\ini},\eta_{\ini})\in\ell^2(\N^*,H^s(\mathbb T_L))$} such that 
$$ (c_n \partial_x V_{\ini, n})_{n \in \N^*} \in \ell^2(\N^*,H^s(\mathbb T_L)).$$
Then there exists a unique \mbox{solution $(V,\eta) \in C^0([0,+\infty), \ell^2(\N^*,H^s(\mathbb T_L))$} to \eqref{eqn:modes:couple:shear} with initial condition  $(V_{\ini,n},\eta_{\ini,n})_{n\in \N^*}$, and it satisfies the energy estimate, for $t > 0$:
\begin{equation}
\label{eqn:nrj:boussinesq:couple:shear}
\begin{aligned}
&\vert (V_n(t,\cdot))_{n \in \N^*} \vert_{\ell^2 H^s}^2 + \vert (c_n \partial_x V_n(t,\cdot))_{n \in \N^*} \vert_{\ell^2H^s}^2 + \vert (\eta_n(t,\cdot))_{n\in \N^*} \vert_{\ell^2 H^s}^2  \\
&\leq \left( \vert (V_{\ini,n})_{n \in \N^*} \vert_{\ell^2 H^s}^2 + \vert (c_n \partial_x V_{\ini,n})_{n \in \N^*} \vert_{\ell^2H^s}^2 + \vert (\eta_{\ini,n})_{n\in \N^*} \vert_{\ell^2 H^s}^2\right)e^{Ct}. 
\end{aligned}
\end{equation} 
\end{proposition} 
\begin{remark}
\label{rk:incomp:CI}
Note that, given the coefficients $(V_{\ini,n})_{n\in \N^*}$ of $V \in L^2(S)$, we can use the incompressibility condition in \eqref{eqn:euler:isopycnal:Vb} to define the coefficients $(w_{\ini,n})_{n\in \N^*}$ of $w_{\ini}$ along the basis $(f_n)_{n\in \N^*}$ by
\begin{equation}
\label{eqn:rk:incomp:temp:1}
w_{\ini,n} \coloneq  -c_n \partial_x V_{\ini,n}.
\end{equation}
Thus, the assumption on $V_{\ini}$ in Proposition \ref{prop:nrj:boussinesq:couple:shear} implies that $ (w_{\ini,n})_{n\in \N^*} \in \ell^2$, so that
\begin{equation}
\label{eqn:rk:incomp:temp:2}w_{\ini} \coloneq  \sum_{n=1}^{\infty} w_{\ini,n} f_n \in H^{s,0}(S).
\end{equation}
In other words, the assumption on $(V_{\ini,n})_{n \in \N^*}$ in Proposition \ref{prop:nrj:boussinesq:couple:shear} together with some $V_{\ini,0}$
 which does not depend on $x$ ensures that one can reconstruct $V_{\ini} \in H^{s,0}(S)$ from $(V_{\ini,n})_{n\in \N}$ and $w_{\ini} \in H^{s,0}(S)$ from $(V_{\ini,n})_{n\in \N}$ through \eqref{eqn:rk:incomp:temp:1} and \eqref{eqn:rk:incomp:temp:2}, such that $(V_{\ini},w_{\ini})$ satisfies the incompressibility condition in \eqref{eqn:euler:isopycnal:Vb}.
\end{remark}
\begin{proof}

Let us start with the case $s=0$. Integrating \eqref{eqn:modes:couple:shear} against $V_n,\eta_n$ and summing over $n \in \N^*$, we find that the terms in  $A^1, \dots, A^4$ are skew symmetric, except for the term $A^1 + A^2$. This yields
\begin{equation}
\label{eqn:nrj:boussinesq:temp:shear}
\begin{aligned}
&\frac12 \frac{d}{dt} \sum_{n\in \N^*} \left( \vert V_n \vert_{L^2(\mathbb T_L)}^2 + \vert (M^{\frac 12} \partial_x V)_n\vert_{L^2(\mathbb T_L)}^2 + \vert \eta_n \vert_{L^2(\mathbb T_L)}^2 \right) \\
&+ \sum\limits_{n\in \N^*}c_n\left( \int_{\mathbb T_L}\partial_x \eta_n  V_n + \int_{\mathbb T_L} \partial_x V_n \eta_n \right) = -\sum_{n \in \N^*} \int_{\mathbb T_L}((A^1 + A^2) \partial_x V)_n V_n.
\end{aligned}
\end{equation} 
For each term of the second sum on the left-hand side in \eqref{eqn:nrj:boussinesq:temp:shear}, we integrate by parts in the last integral, so that the two integrals cancel out each other. For the right-hand side of \eqref{eqn:nrj:boussinesq:temp:shear}, we use \eqref{eqn:Aun_Adeux}, where $C$ does not depend on $n$ or $m$. Using Cauchy-Schwarz inequality on the right-hand side of \eqref{eqn:nrj:boussinesq:temp:shear}, \eqref{eqn:M:sqrt:positive} on both sides and Grönwall's inequality, we get \eqref{eqn:nrj:boussinesq:couple:shear} with $s=0$. \\

Let now $s \geq 0$. Applying $\partial_x^s$ to \eqref{eqn:modes:couple:shear}, we find that $(\partial_x^s V_n, \partial_x^s \eta_n)_{n\in \N^*}$ satisfies \eqref{eqn:modes:couple:shear}. Therefore, applying the previous $L^2-$ energy estimate to $(\partial_x^s V_n, \partial_x^s \eta_n)_{n\in \N^*}$ yields \eqref{eqn:nrj:boussinesq:couple:shear}. From such an energy estimate the uniqueness of the solution follows, and its existence is a standard result that we omit here.
\end{proof}
\section{The linear bilayer Euler equations}
\label{section:ebc}

In this section we present the bilayer Euler equations. We start with the general description of the linear bilayer Euler equations in the presence of a shear flow in Section \ref{subsection:ebc:shear}. Although we do not directly use this formulation in the present study, it serves as an introduction for next two sections: in Section \ref{subsection:ebc:shear:irrot}, we write a reformulation in the irrotational setting, which we use to perform numerical simulations. Finally, in Section \ref{subsection:ebc:no_shear}, we write the linear bilayer Euler equations without a shear flow. In this case, we state a well-posedness result in Sobolev spaces. This result will prove useful in the comparison between these linear bilayer Euler equations in the absence of a shear flow, and the linear stratified Euler equations without a shear flow, which we provide in Section \ref{subsection:sharp_strat:sans_shear}. 

\subsection{The linear bilayer Euler equations with a shear flow}
\label{subsection:ebc:shear}
We start with the description of the linear bilayer Euler equations with a shear flow. The setting is summarized in Figure \ref{fig:euler:bc:notations}.  
\begin{figure}
\includegraphics[width = 0.8\textwidth]{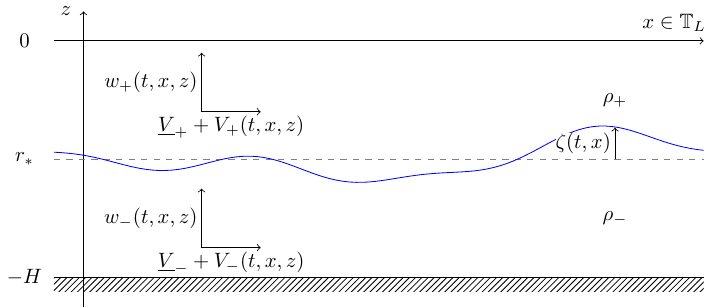}
\caption{\label{fig:euler:bc:notations}The bilayer setting. The unknowns $V_{\pm},w_{\pm}$ are the fluid velocities in the upper and lower layer, as well as the interface $\zeta$ between both layers.}
\end{figure}
These equations read
\begin{equation}
\label{eqn:euler:bilayer:shear}
\sys{
\partial_t V_{\bl} + \Vb_{\bl} \partial_x V_{\bl} + \frac{1}{\rho_{\bl}} \partial_x P_{\bl} + g\partial_x \zeta &= 0,\\
 \partial_t w_{\bl} + \Vb_{\bl} \partial_x V_{\bl} + \frac{1}{\rho_{\bl}} \partial_r P_{\bl} &= 0,\\
\partial_x V_{\bl} + \partial_r w_{\bl} &= 0,\\} \qquad \text{ in } S_+ \cup S_-,
\end{equation}
where we use the notations
\begin{equation}
\label{eqn:notation:bl:shear}
V_{\bl}(t,x,r) \coloneq  \sys{&V_+(t,x,r) \text{ if } r > r_*,\\
						&V_-(t,x,r) \text{ if } r < r_*,}
\end{equation}
respectively $\Vb_{\bl}, w_{\bl}$ and $P_{\bl}$, for some $r_* \in (-H,0)$, which denotes the position of the interface at rest. The quantities $\Vb_{+}$ and $\Vb_{-}$ are two given constants. Here, the indices $+$ and $-$ refer to quantities defined respectively in the upper layer $S_+ \coloneq  \mathbb T_L \times \{ r_* < r < 0\}$ and lower layer $S_- \coloneq  \mathbb T_L \times \{-H < r < r_*\}$. The quantities  $\rho_+$ and $\rho_-$ are two constants with $\rho_+ < \rho_-$. We also write a kinematic equation on the motion on the interface (parametrized by $r_*+\zeta$) between the upper layer $S_+$ and the lower layer $S_-$, which reads
\begin{equation}
\label{eqn:euler_bilayer:zeta:shear}
\partial_t \zeta = -\Vb_{+} \partial_x \zeta + (w_{+})_{|r=r_*}=- \Vb_{-}\partial_x \zeta + (w_{-})_{|r=r_*},
\end{equation}
which also encapsulates the continuity of the normal component of the velocity across the interface. The pressure is also continuous at the interface; we thus have two constraints on the interface, i.e.
\begin{equation}
\label{eqn:continuite:pression_w_interface:shear}
\sys{
(P_+)_{|r=r_*} &= (P_-)_{|r=r_*}, \\
-\Vb_{+} \partial_x \zeta + (w_{+})_{|r=r_*}&=- \Vb_{-}\partial_x \zeta + (w_{-})_{|r=r_*}.}
\end{equation}
The equations \eqref{eqn:euler:bilayer:shear} are completed with the boundary conditions
\begin{equation}
\label{eqn:euler:bilayer:bc:shear}
w_{-|r=-H} = w_{+|r=0} = 0,
\end{equation}
which are the impermeability of the flat bottom and lid. The equations \eqref{eqn:euler:bilayer:shear} are also completed with the initial conditions
\begin{equation}
\label{eqn:euler:bilayer:ci:shear}
V_{\bl |t=0} = V_{\bl,\ini}, \qquad  w_{\bl |t=0} = w_{\bl,\ini},\qquad  \zeta_{|t=0} = \zeta_{\ini}.
\end{equation}
The initial data \eqref{eqn:euler:bilayer:ci:shear} should be consistent with \eqref{eqn:continuite:pression_w_interface:shear}, the incompressibility condition in \eqref{eqn:euler:bilayer:shear} and the boundary conditions \eqref{eqn:euler:bilayer:bc:shear} , namely
\begin{equation}
\label{hyp:euler_bc:shear}
\sys{ \partial_x V_{\bl,\ini} + \partial_r w_{\bl,\ini} &= 0, \\
\Vb_+ \partial_x \zeta_{\ini} - w_{+,\ini|r=r_*} &= \Vb_-  \partial_x \zeta_{\ini} - w_{-,\ini|r=r_*},\\
w_{-,\ini|r=-H} &= w_{+,\ini|r=0} = 0.}
\end{equation}
\subsection{The irrotational case, with a shear flow}
\label{subsection:ebc:shear:irrot}
In this section we consider \eqref{eqn:euler:bilayer:shear} in the irrotational case. We use this formulation for the numerical simulations of the bilayer Euler equations, see Section \ref{subsection:num:bl}. \\

In this section, we assume that $r_*=-H/2$, as this simplifies some formula below; these can be found in more generality in \cite[Chapter A, Section 3.3]{Duchene2022a}. For the same reason, we choose $\Vb_+ = v, \Vb_- = -v$, for some $v \in \R$. The irrotationality of the flow implies that $V_{\bl}$ can be written as
$$V_{\bl} = \nabla_{x,r} \phi_{\bl},$$
where $\phi_+$ and $\phi_-$ are called the velocity potentials of $V_+$ and $V_-$. If such an assumption is satisfied by the initial data $V_{\bl, \ini}$, then it is satisfied for a solution of \eqref{eqn:euler:bilayer:shear} at all times. Let
$$\psi \coloneq  \frac{\rho_-}{\rho_+ + \rho_-} \phi_{-|r=r_*} - \frac{\rho_+}{\rho_+ + \rho_-} \phi_{+|r=r_*} .$$
Then, the equations \eqref{eqn:euler:bilayer:shear} can be recast as  evolution equations on $\zeta$ and $\psi$ only, see for instance \cite[Chapter A, Section 3.3]{Duchene2022a}. We write them directly as equations on the Fourier coefficients $\hat{\zeta}_k, \hat{\psi}_k$ of $\zeta$ and $\psi$. For $k \in \frac{1}{L} \mathbb Z$, these equations read
\begin{equation}
\label{eqn:interfacial_waves}
\sys{
\partial_t \hat{\zeta}_k + ik d(k) \hat{\zeta}_k - b(k) \hat{\psi}_k &= 0, \\
\partial_t \hat{\psi}_k + a(k) \hat{\zeta}_k + ikd(k) \hat{\psi}_k &= 0,}
\end{equation}
where
\begin{equation}
\begin{aligned}
a(k) &\coloneq  g \frac{\rho_- - \rho_+}{\rho_- + \rho_+} - \nouveau{\frac{\rho_- \rho_+}{(\rho_- + \rho_+)^2\tanh(|k|H/2)}} v^2 |k| \\
b(k) &\coloneq  \tanh(|k|H/2) |k|,\\
d(k) &\coloneq  \frac{\rho_+ - \rho_-}{\rho_- + \rho_+}v. 
\end{aligned}
\end{equation}
\subsection{The linear bilayer Euler equations without a shear flow}
\label{subsection:ebc:no_shear}
We now write the linear bilayer Euler equations, without a shear velocity (and without the irrotationality assumption), i.e. we specify \eqref{eqn:euler:bilayer:shear} with $\Vb_{\bl}=0$. These equations read
\begin{equation}
\label{eqn:euler:bilayer}
\sys{
\partial_t V_{\bl} + \frac{1}{\rho_{\bl}} \partial_x P_{\bl} + g\partial_x \zeta &= 0,\\
 \partial_t w_{\bl} + \frac{1}{\rho_{\bl}} \partial_r P_{\bl} &= 0,\\
\partial_x V_{\bl} + \partial_r w_{\bl} &= 0,\\} \qquad \text{ in } S_+ \cup S_-.
\end{equation}
Recall that the quantities  $\rho_+$ and $\rho_-$ are two constants with $\rho_+ < \rho_-$. We also write the kinematic equation on the motion on the interface (parameterized by $r_*+\zeta$) between the upper layer $S_+$ and the lower layer $S_-$, which reads
\begin{equation}
\label{eqn:euler_bilayer:zeta}
\partial_t \zeta = (w_{+})_{|r=r_*}=(w_{-})_{|r=r_*}.
\end{equation}
The pressure is continuous at the interface, and as is $w_{\bl}$, i.e. 
\begin{equation}
\label{eqn:continuite:pression_w_interface}
\sys{
(P_+)_{|r=r_*} &= (P_-)_{|r=r_*}, \\
(w_{+})_{|r=r_*}&=(w_{-})_{|r=r_*};}
\end{equation}the latter condition stems from the incompressibility condition on the interface. Thus, \eqref{eqn:euler_bilayer:zeta} is unambiguous. Note that we do not assume that the fluid velocity is irrotational in the upper and lower layer, and so the equations \eqref{eqn:euler:bilayer} do not reduce to a system on unknowns defined on the interface $r=r_*$ (see for instance Section \ref{subsection:ebc:shear:irrot} and \cite{Duchene2022a} for an account on the bilayer Euler equations in this irrotational case).
The equations \eqref{eqn:euler:bilayer} are completed with the boundary conditions
\begin{equation}
\label{eqn:euler:bilayer:bc}
w_{-|r=-H} = w_{+|r=0} = 0,
\end{equation}
which are the impermeability of the flat bottom and lid. The equations \eqref{eqn:euler:bilayer} are also completed with the initial conditions
\begin{equation}
\label{eqn:euler:bilayer:ci}
V_{\bl |t=0} = V_{\bl,\ini}, \qquad  w_{\bl |t=0} = w_{\bl,\ini}, \qquad \zeta_{|t=0} = \zeta_{\ini}.
\end{equation}
The initial data \eqref{eqn:euler:bilayer:ci} should be consistent with the incompressibility condition in \eqref{eqn:euler:bilayer} and the boundary conditions \eqref{eqn:euler:bilayer:bc}, namely
\begin{equation}
\label{hyp:euler_bc}
\sys{ \partial_x V_{\bl,\ini} + \partial_r w_{\bl,\ini} &= 0, \\
w_{+,\ini|r=r_*} &= w_{-,\ini|r=r_*},\\
w_{-,\ini|r=-H} &= w_{+,\ini|r=0} = 0.}
\end{equation}
We now state a result on the existence and uniqueness of solutions of \eqref{eqn:euler:bilayer}\nv{, which we} prove in Appendix \ref{subsection:euler:bilayer}.
\begin{proposition}
\label{prop:nrj:bilayer}
Assume $\rho_+ < \rho_-$. Let $s \geq 1$, $(V_{\bl,\ini},w_{\bl,\ini},\zeta_{\ini}) \in H^{s,0}(S) \times H^{s,1}(S) \times H^s(\mathbb T_L)$, satisfying \eqref{hyp:euler_bc}. There exists a unique solution $(V_{\bl},w_{\bl},\zeta) \in C^0\left([0,+\infty),H^{s,0}(S) \times H^{s,1}(S) \times H^s(\mathbb T_L)\right)$ to \eqref{eqn:euler:bilayer}, satisfying \eqref{eqn:euler_bilayer:zeta}, \eqref{eqn:continuite:pression_w_interface} and the boundary and initial conditions \eqref{eqn:euler:bilayer:bc} and \eqref{eqn:euler:bilayer:ci}.
\end{proposition}
\section{Numerical scheme}
\label{section:num_scheme}
We start by describing our strategy for solving \eqref{eqn:euler:isopycnal:Vb} numerically in Section \ref{subsection:num_scheme_description}\nv{, which} is based on the modal decomposition introduced in Section \ref{section:normal_modes}. \nv{In Section \ref{subsection:num_scheme:cvce}, we prove the convergence of our numerical scheme, and in Section~\ref{subsection:illustrations} we provide some illustrations on internal waves.}  Finally, in Section \ref{subsection:num:bl}, we present our numerical strategy regarding the irrotational bilayer Euler equations \eqref{eqn:interfacial_waves}.
\subsection{Numerical scheme for the stratified Euler equations}
\label{subsection:num_scheme_description}
Our numerical strategy to compute solutions of \eqref{eqn:euler:isopycnal:Vb} is as follows.\\
\paragraph{\bf Solving the Sturm-Liouville problem.}
First, let $\rhob$ be the density profile under consideration. Then, the Brunt-Väisälä frequency $N^2$ is given analytically or given via a finite-difference approximation, through \eqref{eqn:def:N2}; we use the latter solution. We solve the eigenvalue problem \eqref{eqn:SL} with Dirichlet boundary conditions, in order to find an approximation on $((c_n,f_n,g_n))_{1 \leq n \leq \mathtt{N}}$, for $\mathtt{N} \in \N^*$ some integer. We introduce a uniform discretization of $[-H,0]$: for $\mathtt{m_r} \in \N^*$, we write $r_j = -j h$, with $j \in \{0,\dots,\mathtt{m_r} \}$, and $h \coloneq  H/\mathtt{m_r}$.\\

Using a finite-difference scheme, solving the eigenvalue problem \eqref{eqn:SL} with Dirichlet boundary conditions boils down to solving the following generalized eigenvalue problem
\begin{equation}
\label{eqn:gen_eigv_pb}
\frac{1}{h^2}S Y = - \lambda D Y,
\end{equation}
where
\begin{equation}
\label{eqn:SL:matrix_finite_diff:1}
S =
\begin{pmatrix}
\alpha_0 & \beta_0 & 0 & \cdots & 0 \\
\beta_0 & \alpha_1 & \beta_1 & \ddots & \vdots \\
0 & \beta_1 & \alpha_2 & \ddots  & 0 \\
\vdots & \ddots & \ddots & \ddots & \beta_{\mathtt{m_r}-1} \\
0 & \cdots & 0 & \beta_{\mathtt{m_r}-1} & \alpha_{\mathtt{m_r}}
\end{pmatrix}, \qquad
D = \begin{pmatrix}
\gamma_0 & 0 & 0 & \cdots & 0 \\
0 & \gamma_1 & 0 & \ddots & \vdots \\
0 & 0 & \gamma_2 & \ddots  & 0 \\
\vdots & \ddots & \ddots & \ddots & 0 \\
0 & \cdots & 0 & 0 & \gamma_{\mathtt{m_r}}
\end{pmatrix},
\end{equation}
where, for $j\in \{0,\dots,\mathtt{m_r} \}$:
\begin{equation}
\label{eqn:SL:matrix_finite_diff:2}
\begin{aligned}
\alpha_j &=
-\left(\rhob\!\big(r_j-\tfrac{h}{2}\big)+\rhob\!\big(r_j+\tfrac{h}{2}\big)\right), \qquad &
\beta_j  &= \rhob\!\big(r_j+\tfrac{h}{2}\big),\\
\gamma_j &= \rhob(r_j) N^2(r_j).
\end{aligned}
\end{equation}
Solving for the first eigenvalues and eigenvectors, we obtain $(\frac{1}{c_n^2})_{1 \leq n \leq \mathtt{N}}$ and a family of vectors $(f_n^j)_{1 \leq n \leq \mathtt{N}, 0 \leq j \leq \mathtt{m_r}}$ approximating $(f_n(r_j))_{1 \leq n \leq \mathtt{N}, 0 \leq j \leq \mathtt{m_r}}$. Then using the definition \eqref{eqn:def:gn}, we obtain $(g_n)_{1\leq n \leq \mathtt{N}}$, using the finite difference approximation 
$$g_n(r_j) \approx g_n^j \coloneq  \frac{c_n}{2h}(f_n^{j+1} - f_n^{j-1}).$$

\paragraph{\bf Restriction to a finite number of (vertical) modes.}
We now show how to discretize the system \eqref{eqn:modes:couple:shear}\nouveau{, using the modes from the Sturm-liouville problem \eqref{eqn:SL} studied in Section \ref{subsection:normal_modes}, which we find numerically using \eqref{eqn:gen_eigv_pb}}. Given the density profile $\rhob$ and the shear flow $\Vb$, we can compute the matrices $M^{(\mathtt{N})} \coloneq  (M^{(\mathtt N)}_{n,m})_{n,m \leq \mathtt{N}}, A^{j, (\mathtt{N})} \coloneq  (A^{j,(\mathtt N)}_{n,m})_{n,m\leq \mathtt{N}}$ for $j \in \{1, \dots, 4\}$ that approximate the operators $M, A^j$, through the definitions of the coefficients in \eqref{eqn:def:Mnm} and \eqref{eqn:def:A}, where the integrals are replaced by a quadrature formula; we use the trapezoid rule. We get:
\begin{equation}
\label{eqn:def:Mnm:approx}
M^{(\mathtt N)}_{n,m} \coloneq  c_n c_m \frac{h}{2} \sum_{j=0}^{\mathtt m_r-1} (\rhob(r_j)f_n^j f_m^j + \rhob(r_{j+1})f_n^{j+1} f_m^{j+1}),
\end{equation}
as well as
\begin{equation}
\label{eqn:def:A:approx}
\begin{aligned}
A^{1,(\mathtt N)}_{n,m} &\coloneq  \frac{h}{2} \sum_{j=0}^{\mathtt m_r-1}\left( \Vb(r_j)\rhob(r_j) g_n^jg_m^j + \Vb(r_{j+1}) \rhob(r_{j+1}) g_n^{j+1}g_m^{j+1}\right), \\
A^{2,(\mathtt N)}_{n,m} &\coloneq  -\frac{c_m}{c_n}A^{4,\mathtt N}_{n,m}, \\
A^{3,(\mathtt N)}_{n,m} &\coloneq  c_n c_m \frac{h}{2} \sum_{j=0}^{\mathtt m_r-1}\left( \Vb(r_j)\rhob(r_j) f_n^jf_m^j + \Vb(r_{j+1}) \rhob(r_{j+1}) f_n^{j+1}f_m^{j+1}\right),\\
A^{4,(\mathtt N)}_{n,m} &\coloneq  \frac{h}{2} \sum_{j=0}^{\mathtt m_r-1}\left( \Vb(r_j)\rhob(r_j) N^2(r_j) f_n^jf_m^j + \Vb(r_{j+1}) \rhob(r_{j+1}) N^2(r_{j+1})f_n^{j+1}f_m^{j+1}\right),
\end{aligned}
\end{equation}
where we used that $A^2$ and $A^4$ are defined from the same integral, see Remark \ref{rk:A2}.
We thus consider the following system, that we call semi-discrete in the vertical variable as it only depends on a finite number of vertical modes,
\begin{equation}
\label{eqn:modes:couple:finite}
\sys{\left((I^{(\mathtt N)} - M^{(\mathtt N)} \partial_x^2)\partial_tV^{(\mathtt N)}\right)_n + c_n \partial_x \eta^{(\mathtt N)}_n + \left((2 A^{1,(\mathtt N)} + A^{2,(\mathtt N)} )\partial_x V^{(\mathtt N)}\right)_n  - (A^{3,(\mathtt N)} \partial_x^3 V^{(\mathtt N)})_n &= 0,\\
\partial_t \eta^{(\mathtt N)}_n + c_n \partial_x V^{(\mathtt N)}_n  + (A^{4,(\mathtt N)} \partial_x \eta^{(\mathtt N)})_n&= 0,
} 
\end{equation}
for $1 \leq n \leq \mathtt N$. Let now $((V_{\ini,n},\eta_{\ini,n}))_{n\in \N^*} \in \ell^2(\N^*,H^{s}(\mathbb T_L))$ such that $(c_n \partial_x V_{\ini,n})_{n\in \N^*} \in \ell^2(\N^*, H^{s}(\mathbb T_L))$, for some $s \in \N$. We complete  \eqref{eqn:modes:couple:finite} with the initial condition 
\begin{equation}
\label{eqn:modes:couple:finite:CI}
(V^{(\mathtt N)}_n,\eta^{(\mathtt N)}_n)(0,\cdot) = (V_{\ini,n},\eta_{\ini,n}),
\end{equation}
for $1 \leq n \leq \mathtt N$. 
\\

\paragraph{\bf Restriction to a finite number of (horizontal) Fourier modes.}
We approximate solutions $(V^{(\mathtt N)}_n, \eta^{(\mathtt N)}_n)_{n \in \{1, \dots, \mathtt N\}}$ of \eqref{eqn:modes:couple:finite} by the Fourier series of their first Fourier coefficients with respect to $x \in \mathbb T_L$; namely, take $\mathtt{K} \in \N^*$. Denoting by $\hat{V}^{(\mathtt N)}_{n,k}$ and $\hat{\eta}^{(\mathtt N)}_{n,k}$ the $k^{th}$ Fourier coefficient of $V^{(\mathtt N)}_n$ and $\eta^{(\mathtt N)}_n$ for $k \in \{ -\mathtt{K}, \dots, \mathtt{K}\}$, we obtain the following semi-discrete equations from \eqref{eqn:modes:couple:finite}:  
\begin{equation}
\label{eqn:EDO:modes}
\frac{d}{dt} \vec{\hat{Y}}_k(t) = B_k \vec{\hat{Y}}_k(t),
\end{equation} 
for $k \in \{-\mathtt{K}, \dots, \mathtt{K}\}$ and where 
$$(\vec{\hat{Y}}_k)_{n\leq \mathtt{N}} \coloneq  (\hat{V}^{(\mathtt N)}_{1,k}, \dots, \hat{V}^{(\mathtt N)}_{\mathtt N,k}, \hat{\eta}^{(\mathtt N)}_{1,k}, \dots \hat{\eta}^{(\mathtt N)}_{\mathtt N,k})^T,$$
and $B \coloneq  \begin{pmatrix} B^{1,1} & B^{1,2} \\ B^{2,1} & B^{2,2} \end{pmatrix}$ is defined by blocks, through
\begin{equation}
\label{eqn:def:B}
\begin{aligned}
B^{1,1} &\coloneq  ik (I^{(\mathtt{N})} +M^{(\mathtt{N})}k^2)^{-1} (2 A^{1,(\mathtt{N})} + A^{2,(\mathtt{N})}) +ik^3 (I^{(\mathtt N)} +M^{(\mathtt N)} k^2)^{-1} A^{3,(\mathtt{N})},  &B^{1,2} &\coloneq  (I^{(\mathtt N)} +M^{(\mathtt N)}k^2)^{-1} i k \text{diag}(c_n), \\
B^{2,1} &\coloneq  ik \text{diag}(c_n),  &B^{2,2} &\coloneq  ik A^{4,(\mathtt{N})};
\end{aligned} 
\end{equation}
here, $I^{(\mathtt N)}$ is the identity matrix, and $\text{diag}(c_n)$ denotes the $\mathtt N \times \mathtt N$ diagonal matrix whose $n^{th}$ diagonal coefficient is $c_n$. \\

\paragraph{\bf Initial conditions and time evolution.}
On the one hand, we use \eqref{eqn:EDO:modes} with suitable initial conditions to compute approximate solutions of \eqref{eqn:modes:couple:shear} (see Sections \ref{subsection:illustrations} and \ref{subsection:sharp_strat:sans_shear}). To this end, take an initial condition $(V_{\ini},\eta_{\ini})$ defined on the strip $S$. This initial condition should be sufficiently regular, i.e. satisfy the assumptions of Proposition \ref{prop:nrj:boussinesq:couple:shear} for some regularity index $s \geq 0$, and such that the first vertical mode $V_{\ini,0}$ of $V_{\ini}$ is independent of $x$.  We compute the $\mathtt N$ first coefficients $(\eta_{\ini, n})_{1\leq n \leq \mathtt N}$ of $\eta_{\ini}$ the basis $(f_n)_{n\in \N^*}$, and the $\mathtt N+1$ first coefficients $(V_{\ini, n})_{0\leq n \leq \mathtt N}$ of $V_{\ini}$ along the basis $(g_n)_{n \in \N}$. Recall the coefficient $V_{0}$ of the solution of \eqref{eqn:euler:isopycnal:Vb} is constant in time; it thus remains to compute the evolution of the coefficients $V_n$ for $n \in \{1, \dots, \mathtt N\}$. We consider the first Fourier coefficients of $(\eta_{\ini, n}, V_{\ini,n})_{1\leq n \leq \mathtt N}$, and we thus write, for $k \in \{-\mathtt K, \dots, \mathtt K\}$,
\begin{equation}
\label{eqn:EDO:modes:CI:coefs}
\vec{\hat{Y}}_{k,\ini} \coloneq  (\hat{V}_{\ini,1,k}, \dots, \hat{V}_{\ini,\mathtt N,k}, \hat{\eta}_{\ini,1,k}, \dots \hat{\eta}_{\ini,\mathtt N,k})^T.
\end{equation}
The system \eqref{eqn:EDO:modes} is thus completed with the initial condition
\begin{equation}
\label{eqn:EDO:modes:CI}
\vec{\hat{Y}}_k(0) = {\vec{\hat{Y}}}_{k,\ini}. 
\end{equation}
In order to solve \eqref{eqn:EDO:modes} together with the initial condition \eqref{eqn:EDO:modes:CI}, we need to perform a time discretization. We use a Runge-Kutta method of order $4$. Note however that other choices are possible. For instance, one could compute a numerical approximation of the exponential of the matrices $B_k$, for $k \in \{-\mathtt K, \dots, \mathtt K \}$. \\

\paragraph{\bf Dispersion relation.}
On the other hand, we use \eqref{eqn:EDO:modes} to compute an approximation of the dispersion relation of  \eqref{eqn:euler:isopycnal:Vb}, see Section \ref{subsection:sharp_strat:shear}. \nv{We} give a more precise definition of the dispersion relation in our setting.
\begin{definition}
\label{def:dispersion_relation}
Let $k \in \frac{1}{L} \mathbb Z$, that we call a {\rm frequency}. We define the {\rm phase velocities for the frequency k} as the set $\C(k)$ of complex numbers $c$ so that there exists a solution to \eqref{eqn:modes:couple:shear} of the form 
\begin{equation}
\label{eqn:def:dispersion:mode}
(e^{ik(x-ct)} \hat{V}_n, e^{ik(x-ct)} \hat{\eta}_n)_{n\in \N^*},
\end{equation}
with $\hat{V}_n, \hat{\eta}_n \in \ell^2(\N^*)$. We say that $(k,\omega(k))$ satisfies the {\rm dispersion relation} of \eqref{eqn:euler:isopycnal:Vb} if $c(k)\coloneq \omega(k)/k \in \C(k)$. If there exists $c \in \mathcal C(k)$ with $\Im(c) >0$, then the associated solution grows exponentially \nouveau{in time}. In this case, we say that $k$ is an {\rm unstable} frequency. Otherwise, we say that $k$ is {\rm stable}. In Section \ref{section:sharp_strat}, we consider a density profile and a shear flow that depend on a parameter $\delta$; in this case, we write $\mathcal{C}_{\delta}(k)$ to specify the dependency of $\mathcal{C}_{\delta}(k)$ on $\delta$.
\end{definition}
Let $\mathtt N, \mathtt K \in \N^*$, and $k \in \{ -\mathtt K, \dots, \mathtt K \}.$ We define $\mathcal C^{\mathtt N}(k)$ as the set of complex numbers $c$ such that there exist $\hat{V}^{\mathtt N}, \hat{\eta}^{\mathtt N} \in \mathbb C^{\mathtt N}$ such that 
$$(e^{ik(x-ct)} \hat{V}^{\mathtt N}_n, e^{ik(x-ct)} \hat{\eta}^{\mathtt N}_n)_{n\in \{1, \dots, \mathtt N\}}$$
is a solution of \eqref{eqn:EDO:modes}. Note that $\mathcal C^{\mathtt N}(k)$ can be computed from the eigenvalues of $B_k$ defined in \eqref{eqn:def:B}, namely
\begin{equation}
\label{eqn:def:mathcalC}
\mathcal{C}^{\mathtt N}(k) = \{ -i\lambda / k, \lambda \text{ eigenvalue of }B_k\}.
\end{equation}
We use \eqref{eqn:def:mathcalC} to compute a numerical approximation of the dispersion relation (more precisely, of the set $\mathcal C(k)$ of phase velocities, for $k \in \{ - \mathtt K, \dots, \mathtt K\}$) of \eqref{eqn:euler:isopycnal:Vb}.\\

The requirement associated to $\hat{V}_n, \hat{\eta}_n \in \ell^2(\N^*)$ in Definition \ref{def:dispersion_relation} is automatically satisfied by any eigenvector associated with an eigenvalue of $B_k$, when computing $\mathcal{C}^{\mathtt N}(k)$. This difference is linked to the distinction between eigenvalues and spectrum of unbounded operators, and falls outside of the scope of this study. 
\subsection{Convergence of the scheme}
\label{subsection:num_scheme:cvce}
In this section we prove the convergence of the numerical scheme given in Section \ref{subsection:num_scheme_description} for solving \eqref{eqn:modes:couple:shear}. More precisely, as \eqref{eqn:modes:couple:shear} is a linear system of equations with coefficients that do not depend on $x \in \mathbb T_L$, we do not study the time and horizontal space discretization, as these are fairly standard. For the same reason, we do not consider the matrices $M^{(\mathtt N)}, A^{1,(\mathtt N)}, \dots, A^{4,(\mathtt N)}$ given by the quadrature formulas, but rather the direct truncation of the operators $M, A^{1}, \dots, A^4$, namely
\begin{equation}
\label{eqn:def:M_A:tronques}
M^{(\mathtt N)}_{n,m} = M_{n,m}, \qquad A^{j,(\mathtt N)}_{n,m} = A^j_{n,m},
\end{equation}
for $j \in \{1,2,3,4\}$, $n,m \in \{1, \dots, \mathtt N\}$ and according to the definitions \eqref{eqn:def:Mnm} and \eqref{eqn:def:A} of $M_{n,m}$ and $A^{j}_{n,m}$. \\
 
The following proposition states that there exists a unique solution $(V^{(\mathtt N)}_n,\eta^{(\mathtt N)}_n)_{1 \leq n \leq \mathtt N}$ to \eqref{eqn:modes:couple:finite} defined with \eqref{eqn:def:M_A:tronques}, with initial conditions \eqref{eqn:modes:couple:finite:CI}, and provides an upper bound on its difference with the first $\mathtt N$ modes of $(V_n,\eta_n)_{n \in \N^*}$, solution to \eqref{eqn:modes:couple} with initial condition given by $(V_{\ini,n},\eta_{\ini,n})_{n\in \N^*}$.
\begin{proposition}
\label{prop:num_scheme:cvce}
Let $\mathtt N \in \N^*$, $s \in \N$, $\rhob, \Vb \in W^{2,\infty}([-H,0])$ such that there exist constants $c^*, c_* > 0$ such \mbox{that $c_* \leq \rhob \leq c^*$} and $-\rhob' \geq c_*$. Let $(V_{\ini},\eta_{\ini}) \in \ell^2(\N^*,H^{s+3}(\mathbb T_L))$ such that $(c_n \partial_x V_{\ini})_{n\in \N^*} \in \ell^2(\N^*, H^{s+3}(\mathbb T_L))$, and denote by $C_0$ an upper bound on its $\ell^2(\N^*, H^{s+3}(\mathbb T_L))$-norm. Under the assumptions of Proposition \ref{prop:nrj:boussinesq:couple:shear}, there exists a unique solution $((V^{(\mathtt N)}_n,\eta^{(\mathtt N)}_n))_{1\leq n \leq \mathtt N} \in C^0([0,\infty), H^{s+3}(\mathbb T_L)^{\mathtt N})$ to \eqref{eqn:modes:couple:finite} with initial condition \eqref{eqn:modes:couple:finite:CI}. Consider also the unique solution $(V_n,\eta_n)_{n \in \N^*} \in C^0([0,+\infty), \ell^2(\N^*,H^{s+3}(\mathbb T_L)))$ to \eqref{eqn:modes:couple:shear} with initial condition $(V_{\ini},\eta_{\ini})$ given by Proposition~\ref{prop:nrj:boussinesq:couple:shear}. Then, for $T > 0$, there exists $ C > 0$ depending on $C_0$ and $T$ such that
\begin{equation}
\label{eqn:nrj:modes:diff}
\begin{aligned}
&\sup_{t\in [0,T]} \left(\vert (V_n - V^{(\mathtt N)}_n)(t,\cdot)|_{\ell^2H^{s}} + \vert c_n \partial_x (V_n - V^{(\mathtt N)}_n)(t,\cdot)|_{\ell^2H^{s}} + \vert (\eta_n - \eta^{(\mathtt N)}_n)(t,\cdot)\vert_{\ell^2 H^{s}}\right) \\
&\leq C\left( \vert V_{\ini,n} - V^{(\mathtt N)}_{\ini,n}|_{\ell^2 H^{s}} + \vert c_n \partial_x (V_{\ini,n} - V^{(\mathtt N)}_{\ini,n})|_{\ell^2H^{s}}+ \vert \eta_{\ini,n} - \eta^{(\mathtt N)}_{\ini, n}\vert_{\ell^2 H^{s}} +\nouveau{\epsilon(\mathtt N)}\right),
\end{aligned}
\end{equation}
where 
$$\epsilon(\mathtt N) \underset{\mathtt N \to 0}{\longrightarrow} 0,$$
and with the convention $(V^{(\mathtt N)}_n,\eta^{(\mathtt N)}_n)=(0,0)$ for $n > \mathtt N$.
\end{proposition}
\begin{remark}
\label{rk:cvce_scheme}
In the right-hand side of \eqref{eqn:nrj:modes:diff}, $(V_{\ini,n} - V^{(\mathtt N)}_{\ini,n}, \eta_{\ini,n} - \eta^{(\mathtt N)}_{\ini, n}) = (0,0)$ for $n \leq \mathtt N$, so that the first three terms on the right-hand side of \eqref{eqn:nrj:modes:diff} account for the convergence rate of the modal decomposition of the unknowns $V$ and $\eta$ in \eqref{eqn:euler:isopycnal:no_shear} along the basis $(g_n)_n$ and $(f_n)_n$. This convergence rate can be linked to the vertical regularity of $V$ and $\eta$, see \cite[Proposition 3]{Desjardins2019}.
\end{remark}
\begin{remark}
\label{rk:better_cvce}
In the case $\Vb = 0$, the operators $A^1, \dots, A^4$ in \eqref{eqn:modes:couple:shear} and the matrices $A^{1,(\mathtt N)}, \dots, A^{4,(\mathtt N)}$ in \eqref{eqn:modes:couple:finite} vanish. In this case, \eqref{eqn:nrj:modes:diff} can be replaced with the more precise estimate
\begin{equation}
\label{eqn:nrj:modes:diff:sans_shear}
\begin{aligned}
&\vert (V_n - V^{(\mathtt N)}_n)(t,\cdot)|_{\ell^2 H^{s}} + \vert c_n \partial_x (V_n - V^{(\mathtt N)}_n)(t,\cdot)|_{\ell^2H^{s}} + \vert (\eta_n - \eta^{(\mathtt N)}_n)(t,\cdot)\vert_{\ell^2 H^{s}}\\
 \leq &\vert V_{\ini,n} - V^{(\mathtt N)}_{\ini,n}|_{\ell^2 H^{s}} + \vert c_n \partial_x (V_{\ini,n} - V^{(\mathtt N)}_{\ini,n})|_{\ell^2H^{s}}+ \vert \eta_{\ini,n} - \eta^{(\mathtt N)}_{\ini, n}\vert_{\ell^2 H^{s}} + C\frac{t}{\mathtt N},
 \end{aligned}
\end{equation}
where $C$ does not depend on $t$. There are two differences between \eqref{eqn:nrj:modes:diff} and \eqref{eqn:nrj:modes:diff:sans_shear}. The first one is that the constant $C$ in \eqref{eqn:nrj:modes:diff:sans_shear} does not depend on time. The second difference between \eqref{eqn:nrj:modes:diff:sans_shear} and \eqref{eqn:nrj:modes:diff} is that in \eqref{eqn:nrj:modes:diff:sans_shear}, the error is of size $\frac{t}{\mathtt N}$; indeed, when $\Vb = 0$, the only contribution to this error term is the truncation of the operator $M$ (see \eqref{eqn:modes:diff:temp:3}), and because of \eqref{eqn:def:Mnm} and \eqref{eqn:c_n:behaviour}, we get
$$M_{n,m} \approx \frac{1}{nm}.$$
\end{remark}
\begin{proof}
The proof for the existence of $V^{(\mathtt N)},\eta^{(\mathtt N)}$ is the same as the proof of Proposition \ref{prop:nrj:boussinesq:couple:shear} and we omit it. In particular, it yields that there exists $C > 0$ independent of $\mathtt N$ such that, for $t \in [0,T]$:
\begin{equation}
\label{eqn:modes:diff:temp:0}
\Vert (V^{(\mathtt N)}_n(t,\cdot))_n \Vert_{\ell^2 H^{s+3}} + \Vert (c_n\partial_x V^{(\mathtt N)}_n(t,\cdot))_n \Vert_{\ell^2 H^{s+3}} + \Vert (\eta^{(\mathtt N)}_n(t,\cdot))_n \Vert_{\ell^2 H^{s+3}} \leq Ce^{CT}.
\end{equation}
We now proceed to prove \eqref{eqn:nrj:modes:diff} with $s=0$. With the convention $(V^{(\mathtt N)}_n,\eta^{(\mathtt N)}_n)=(0,0)$ for $n > \mathtt N$, we write
\begin{equation}
\label{eqn:def:Vtilde}
(\tilde{V}_n,\tilde{\eta}_n) \coloneq  (V_n - V^{(\mathtt N)}_n,\eta_n - \eta^{(\mathtt N)}_n),
\end{equation} 
for $n \in \N^*$. Taking the difference between \eqref{eqn:modes:couple:shear} and \eqref{eqn:modes:couple:finite} we get
\begin{equation}
\label{eqn:modes:diff}
\sys{
&\left((I - M) \partial^2_x)\partial_t\tilde{V}\right)_n + c_n \partial_x \tilde{\eta}_n + ((2 A^1 + A^2) \partial_x \tilde{V})_n - (A^3 \partial_x^3 \tilde{V})_n\\
&= \left((M-M^{(\mathtt N)}) \partial^2_x \partial_t V^{(\mathtt N)}\right)_n + \left(2(A^{1, (\mathtt N)} - A^{1}) + A^{2, (\mathtt N)} - A^2) \partial_x V^{(\mathtt N)}\right)_n - \left(( A^{3, (\mathtt N)} - A^3) \partial_x^3 V^{(\mathtt N)}\right)_n,\\
&\partial_t \tilde{\eta}_n + c_n \partial_x \tilde{V}_n + (A^{4} \partial_x \tilde{\eta})_n= ((A^{4,(\mathtt N)} - A^4) \partial_x \eta^{(\mathtt N)})_n.
} 
\end{equation}
Multiplying \eqref{eqn:modes:diff} by $\tilde{V}_n$ and $\tilde{\eta}_n$ for $n \in \N^*$, integrating over $\mathbb T_L$ and summing over $\N^*$ we get, similarly to the proof of Proposition \ref{prop:nrj:boussinesq:couple:shear}:
\begin{equation}
\label{eqn:modes:diff:temp}
\begin{aligned}
&\frac12 \frac{d}{dt} \left( \vert (\tilde{V}_n)_{n \in \N} \vert_{\ell^2 L^2}^2 + \vert (M^{\frac12} \partial_x \tilde{V})_{n \in \N} \vert_{\ell^2L^2}^2 + \vert (\tilde{\eta}_n)_{n\in \N} \vert_{\ell^2 L^2}^2 \right)\\
& =  \sum_{n=1}^{\infty} \int_{\mathbb T_L}\left((M - M^{(\mathtt N)}) \partial_x^2 \partial_t V^{(\mathtt N)}\right)_n \tilde{V}_n + \sum_{n=1}^{\infty} \int_{\mathbb T_L}((A^{1, (\mathtt N)} + A^{2, (\mathtt N)} - (A^{1} + A^2)) \partial_x V^{(\mathtt N)})_n \tilde{V}_n +\sum_{n=1}^{\infty} \int_{\mathbb T_L}(A^{1, (\mathtt N)} - A^{1})\partial_x V^{(\mathtt N)})_n \tilde{V}_n \\
&- \sum_{n=1}^{\infty} \int_{\mathbb T_L}((A^{3, (\mathtt N)} - A^3) \partial_x^3 V^{(\mathtt N)})_n \tilde{V}_n 
+ \sum_{n=1}^{\infty} \int_{\mathbb T_L}((A^{4,(\mathtt N)} - A^4) \partial_x \eta^{(\mathtt N)})_n \tilde{\eta}_n. 
\end{aligned}
\end{equation}
We now bound the terms on the right-hand side of \eqref{eqn:modes:diff:temp} from above. We only treat the first and second terms, the other ones are treated the same way. Because $(M-M^{(\mathtt N)})_{n,m}$ is $0$ if $n\leq \mathtt N$ and $m \leq \mathtt N$, and $M_{n,m}$ otherwise, and $V^{(\mathtt N)}_n = 0$ when $n > \mathtt N$, we can write
$$\sum_{n=1}^{\infty} \int_{\mathbb T_L}\left((M - M^{(\mathtt N)}) \partial_x^2 \partial_t V^{(\mathtt N)}\right)_n \tilde{V}_n = \sum_{n=\mathtt N+1}^{\infty} \sum_{m=1}^{\mathtt N}\int_{\mathbb T_L} M_{n,m} \partial_x^2 \partial_t V^{(\mathtt N)}_m \tilde{V}_n.$$
We then use \eqref{eqn:M:borne_coefs} and \eqref{eqn:c_n:behaviour} as well as Cauchy-Schwarz inequality to write
\begin{equation}
\label{eqn:modes:diff:temp2}
\sum_{n=1}^{\infty}\int_{\mathbb T_L} \left((M - M^{(\mathtt N)}) \partial_x^2 \partial_t V^{(\mathtt N)}\right)_n \tilde{V}_n \leq \frac{C}{\mathtt N} \vert c_m \partial_x \partial_t V^{(\mathtt N)}_m \vert_{\ell^2 H^{1}} \vert \tilde{V} \vert_{\ell^2 L^2}.
\end{equation}
We now need the following lemma to bound $\partial_t V^{(\mathtt N)}$ from above using \eqref{eqn:modes:couple:finite}.
\begin{lemma}
\label{lemma:M:inverse}
Let $s \geq 0$, $f \in  H^s(\mathbb{T}_L)^{\mathtt N}$, there exists a unique solution $u \in  H^s(\mathbb T_L)^{\mathtt N}$ to the equation
\begin{equation}
\label{eqn:M:inverse}
(I - M^{(\mathtt N)} \partial^2_x) u = f,
\end{equation}
and $u$ satisfies
\begin{equation}
\label{eqn:M:inverse:borne}
\Vert u \Vert_{\ell^2H^s} + \Vert (c_n \partial_x u_n)_n \Vert_{\ell^2 H^s} \leq C \Vert f \Vert_{\ell^2H^s},
\end{equation}
with $C > 0$, a constant that does not depend on $\mathtt N$. 
\end{lemma}
\begin{proof}
This a standard elliptic regularity result, given \eqref{eqn:M:sqrt:positive} (see for instance \cite[Part II, Section 6.3, Theorem 2]{Evans})
\end{proof}
Using Lemma \ref{lemma:M:inverse} together with the first equation in \eqref{eqn:modes:couple:finite}, \eqref{eqn:modes:diff:temp2} yields 
\begin{equation}
\label{eqn:modes:diff:temp:3}
\int \sum_{n=1}^{\infty} \left((M - M^{(\mathtt N)}) \partial_x^2 \partial_t V^{(\mathtt N)}\right)_n \tilde{V}_n \leq \frac{C}{\mathtt N} (\Vert V^{(\mathtt N)} \Vert_{\ell^2 H^{2}} + \Vert c_n \partial_x V_n^{(\mathtt N)} \Vert_{\ell^2 H^3} + \Vert \eta^{(\mathtt N)} \Vert_{\ell^2 H^{2}}) \Vert \tilde{V} \Vert_{\ell^2 L^2}.
\end{equation}
We now turn to the second term on the right-hand side of \eqref{eqn:modes:diff:temp}. We use the identity \eqref{eqn:A1plusA2}, and perform an integration by parts to get, for $n \geq \mathtt N+1$ and $m \leq \mathtt N$:
\begin{equation}
\label{eqn:A1A2:ipp}
A^{1,(\mathtt N)}_{n,m} + A^{2,(\mathtt N)}_{n,m}- (A^{1}_{n,m} + A^{2}_{n,m}) = c_n\int_{-H}^0 c_m (\Vb'' f_m \rhob + \Vb' \rhob' f_m + \Vb' \rhob \frac{1}{c_m}g_m ) f_n ,
\end{equation}
where we used \eqref{eqn:def:gn}. We can now use that $c_n \in \ell^2$ and Fubini's theorem together with Cauchy-Schwarz inequality to write, for $n \in \N^*$,
\begin{equation}
\label{eqn:modes:diff:temp:4}
\begin{aligned}
&\sum_{n=\mathtt N+1}^{\infty}\int_{\mathbb T_L}((A^{1,(\mathtt N)}_{n,m} + A^{2,(\mathtt N)}_{n,m}- (A^{1}_{n,m} + A^{2}_{n,m})) \partial_x V^{(\mathtt N)})_n\tilde{V}_n \\
&= \int_{-H}^0 \int_{\mathbb T_L} \left(\sum_{m=1}^{\mathtt N}c_m(\Vb'' f_m \rhob + \Vb' \rhob' f_m + \Vb' \rhob \frac{1}{c_m}g_m ) \partial_x V^{(\mathtt N)}_m\right) \left( \sum_{n=\mathtt N+1}^{\infty} c_n f_n \tilde{V}_n \right)\\
& \leq C\mathtt N  \Vert c_m\partial_x V^{(\mathtt N)}_m \Vert_{\ell^2 L^2}  \frac{1}{\mathtt N} \left\Vert \sum_{n=\mathtt N+1}^{\infty} \tilde{V}_n f_n \right\Vert_{L^2(\mathbb T_L \times [-H,0])} \\
&\leq C \epsilon(\mathtt N),
\end{aligned}
\end{equation}
where we also used $\Vb \in W^{2,\infty}([-H,0])$ and $\rhob \in W^{1,\infty}([-H,0])$.
Plugging \eqref{eqn:modes:diff:temp:3}, \eqref{eqn:modes:diff:temp:4} into \eqref{eqn:modes:diff:temp} (together with similar estimates for the other summands), and using Grönwall's inequality, we get \eqref{eqn:modes:diff} for $s=0$. \\

Now applying $\partial_x^{\tau}$ with $0 \leq \tau \leq s$ to  \eqref{eqn:modes:diff:temp}, we get that 
\begin{equation}
\label{eqn:def:Vtilde:tau}
(\partial_x^{\tau}(V_n - V^{(\mathtt N)}_n), \partial_x^{\tau}(\eta_n - \eta^{(\mathtt N)}_n))
\end{equation} 
also satisfies \eqref{eqn:modes:diff:temp}. Applying \eqref{eqn:nrj:modes:diff} with $s=0$ to these new unknowns and summing over $\tau$ yields \eqref{eqn:nrj:modes:diff} with $s > 0$.
\end{proof}
\subsection{Illustrations}
\label{subsection:illustrations}
In the \nv{absence of} shear flow ($\Vb= 0$) and when the Brunt-Väisälä frequency $N^2$ is a constant, the evolution of the modes is decoupled. Then, the dispersion relation of \eqref{eqn:euler:isopycnal:no_shear} can be computed analytically (see for instance \cite{GerkemaZimmerman2008}). Indeed, for $k \in \frac{1}{L}\mathbb Z$, the set $\mathcal C(k)$ of the phase velocities is 
$$\mathcal C(k) = \bigcup_{n \in \N^*}\{c(k,n), -c(k,n) \}, \qquad \text{ with } c(k,n) = \frac{|k| N}{\sqrt{k^2 + c_n^2}}.$$This leads to the Saint-Andrew's cross pattern, which is observed in laboratory experiments, see \cite{SutherlandLinden1999}. We now check that our numerical strategy reproduces this phenomenon. We consider the density profile
\begin{equation}
\label{eqn:def:rhob:1}
\rhob^{(1)}(r) \coloneq  e^{-2r},
\end{equation}
so that the Brunt-Väisälä frequency $N^2$ is constant. We consider the response of the system \eqref{eqn:modes:couple:finite} with $\mathtt N$ modes to a forcing that is localized in space. More precisely, we consider the system \eqref{eqn:modes:couple:finite} with $\Vb=0$ and an additional source term; dropping the superscripts $(\mathtt N)$ for readability, we get
\begin{equation}
\label{eqn:modes:finite:forcing}
\sys{
\left((1-M\partial_x^2)\partial_t V\right)_n + c_n \partial_x \eta_n &= 0,\\
\partial_t \eta_n + c_n \partial_x V_n &= F_n(t,x),
} \text{ for } n \in \{ 1, \dots, \mathtt N\},
\end{equation}
where the forcing $(F_n)_{n \in \{1, \dots, \mathtt N \}}$ is defining as follows. Let 
\begin{equation}
\label{eqn:def:forcing}
F(t,x,r) \coloneq  0.1\sin(0.4Nt)\left( 1-\tanh(|x|/0.02) \right)\left( 1-\tanh(|r+0.7|/0.02) \right)
\end{equation}
be a forcing that is oscillating in time at the frequency $0.4 N$, where $N$ is the Brunt-Väisälä frequency, and localized around $x=0$ and $r=-0.7$ in space. Then, for $n \in \{1, \dots, \mathtt N \}$, we define its first coefficients $F_n$ as the coefficients along the basis $(f_n)$, so that
$$ F(t,x,r) \approx \sum_{n=1}^{\mathtt N} F_n(t,x)f_n(r).$$
We show the response of the system under the forcing $F$ in Figure \ref{fig:saint_andrew}, showing the deviation $\rho$ from the density profile $\rhob^{(1)}$. We use $\mathtt N = 100$ vertical modes, computed with $\mathtt m_r = 5000$ vertical points, and $\mathtt K = 512$ horizontal Fourier modes. We need an extensive number of modes to accurately represent $F(t, \cdot)$, whose coefficients decay slowly, as it is localized in a small region in space. The time interval we consider is $[0,80]$, discretized with $500$ points; the variable $t$ in Figure \ref{fig:saint_andrew} is renormalized by $80$. We also used $H=L=1$. \\ 
\begin{figure}
\includegraphics[width=0.9\textwidth]{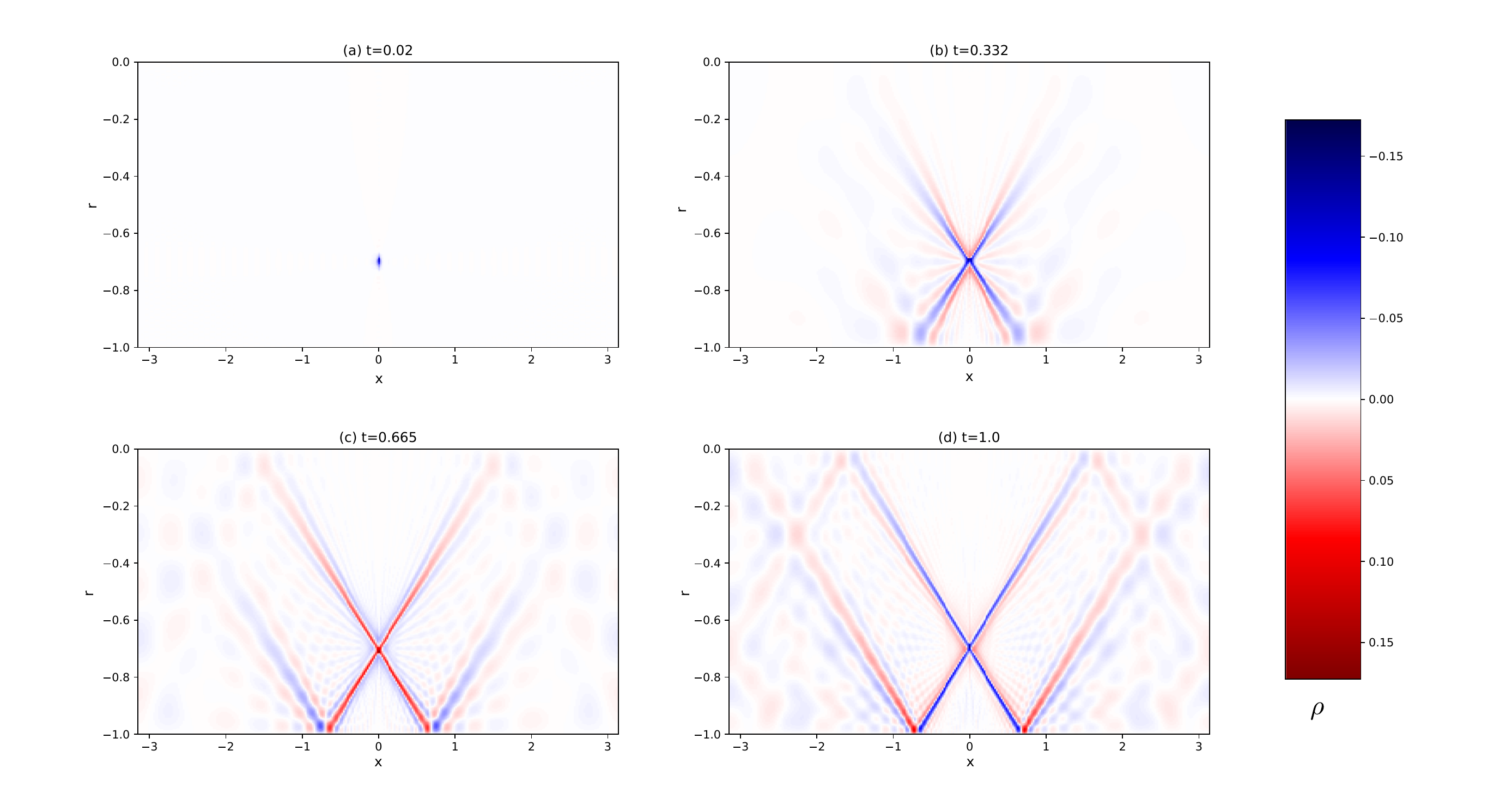}
\caption{\label{fig:saint_andrew} Response of the system \eqref{eqn:modes:finite:forcing} to a localized forcing. Internal waves propagate following Saint-Andrew's cross pattern, in a uniformly stratified fluid. The fluid starts at rest, with an oscillating localized source term at $x=0$ and $r=-0.7$, visible in $(a)$. Colors show the size of the density deviation $\rho$ from the density profile $\rhob^{(1)}$. }
\end{figure}
\FloatBarrier

\nv{Still assuming} $\Vb = 0$, \nv{we now} consider the case where the Brunt-Väisälä frequency $N^2$ is not constant. We rather consider two stably stratified layers, separated by a pycnocline of thickness $\delta > 0$, and use the profile
\begin{equation}
\label{eqn:def:rhob:2}
\rhob^{(2)}_{\delta}(r) \coloneq  0.75 +  0.75\frac{\exp(-2 r)-1 - 5\left(\frac{1}{\pi}\arctan\left(\frac{r+0.5}{0.5\delta}\right)+0.5\right)}{\exp(2) -1 - 5\left(\frac{1}{\pi}\arctan\left(\frac{-1}{\delta}\right) +0.5\right)},
\end{equation}
represented in Figure \ref{fig:waveguide:strat}. The first modes associated to $\rhob^{(2)}_{5.10^{-2}}$ are shown in Figure \ref{fig:waveguide:modes}. The evolution of the first coefficients $(V_n,\eta_n)_{n\in \{1, \dots, \mathtt N\}}$ is given by the system \eqref{eqn:modes:finite:forcing}. In particular, the matrix $M$ is not diagonal and the modes are coupled.\\
\begin{figure}
\centering
\begin{minipage}[t]{0.49\textwidth}
\includegraphics[width = \textwidth]{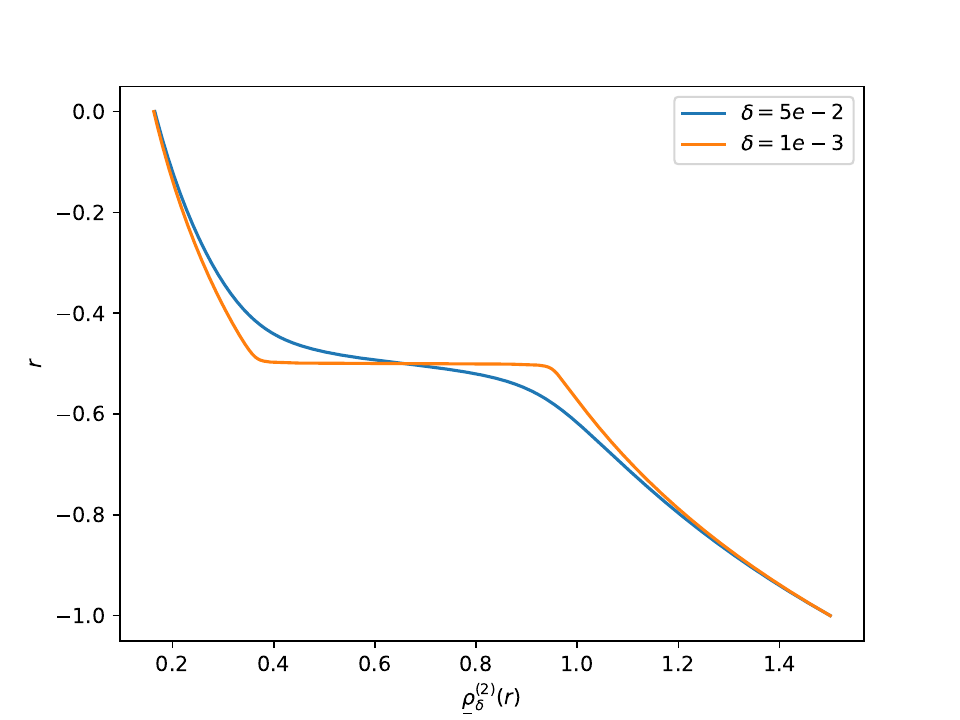}
\captionof{figure}{\label{fig:waveguide:strat}Two stratification profiles with an upper layer ($r > -0.5+\delta$), a pycnocline ($r \in [-0.5-\delta,-0.5+\delta]$) where the density increases sharply and a lower layer ($r < -0.5-\delta$), with $\delta = 5 . 10^{-2}$ and $\delta = 10^{-3}$. Both the upper and lower layers are stably stratified with a roughly constant Brunt-Väisälä frequency.}
\end{minipage}
\begin{minipage}[t]{0.49\textwidth}
\includegraphics[width = \textwidth]{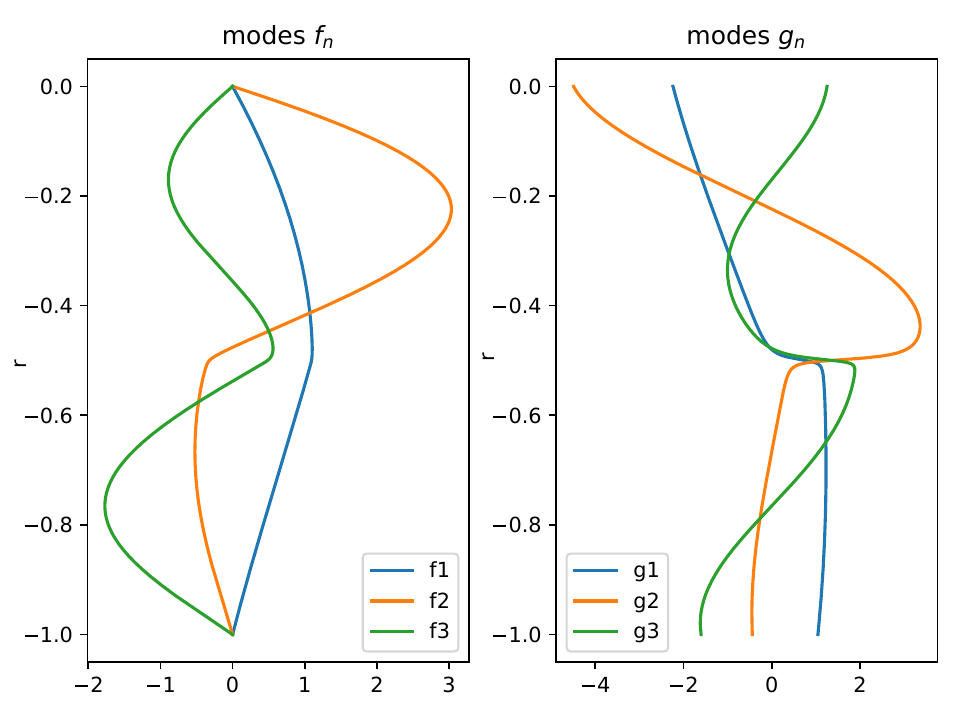}
\captionof{figure}{First modes associated to the stratification profile $\rhob^{(2)}_{\delta}$ with $\delta = 5.10^{-2}$.}
\label{fig:waveguide:modes}
\end{minipage}
\end{figure}
\begin{figure}
\centering
\begin{minipage}[t]{0.49\textwidth}
\includegraphics[width=\textwidth]{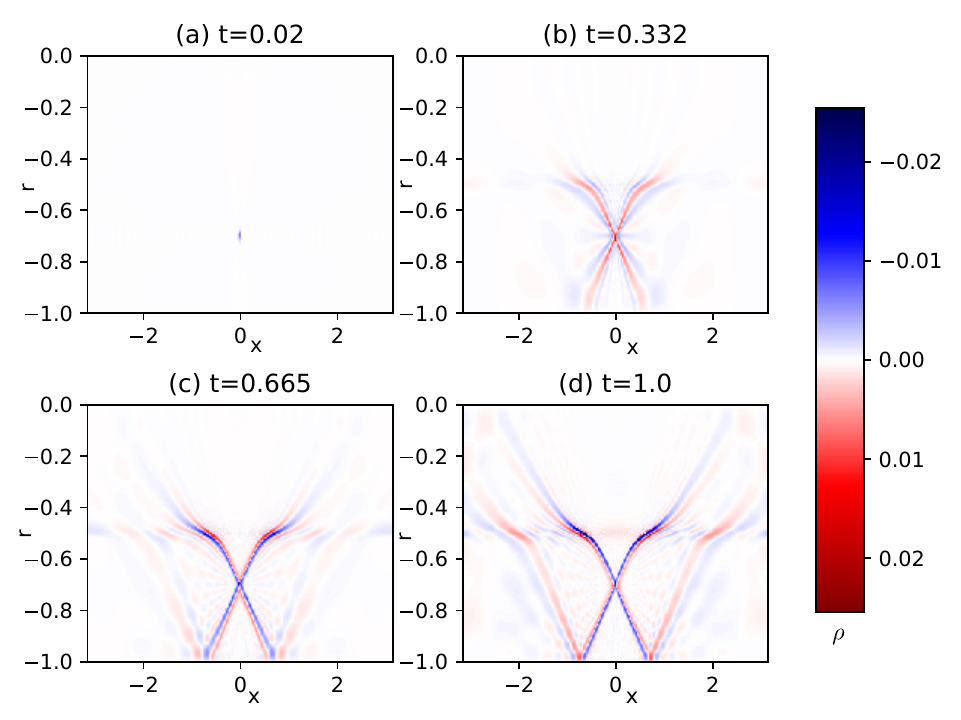}
\captionof{figure}{\label{fig:waveguide:snapshots}Response of the system \eqref{eqn:modes:finite:forcing} to the localized forcing $F$ defined in \eqref{eqn:def:forcing}. Internal waves propagate following Saint-Andrew's cross pattern, in the lower layer, and are deviated in the pycnocline, before being transmitted to the upper layer. The fluid starts at rest, with an oscillating localized source term at $x=0$ and $r=-0.7$, visible in $(a)$. Colors show the size of the density deviation $\rho$ from the density profile $\rhob^{(2)}_{\delta}$ with $\delta = 5.10^{-2}$. }
\end{minipage}
\begin{minipage}[t]{0.49\textwidth}
\includegraphics[width=\textwidth]{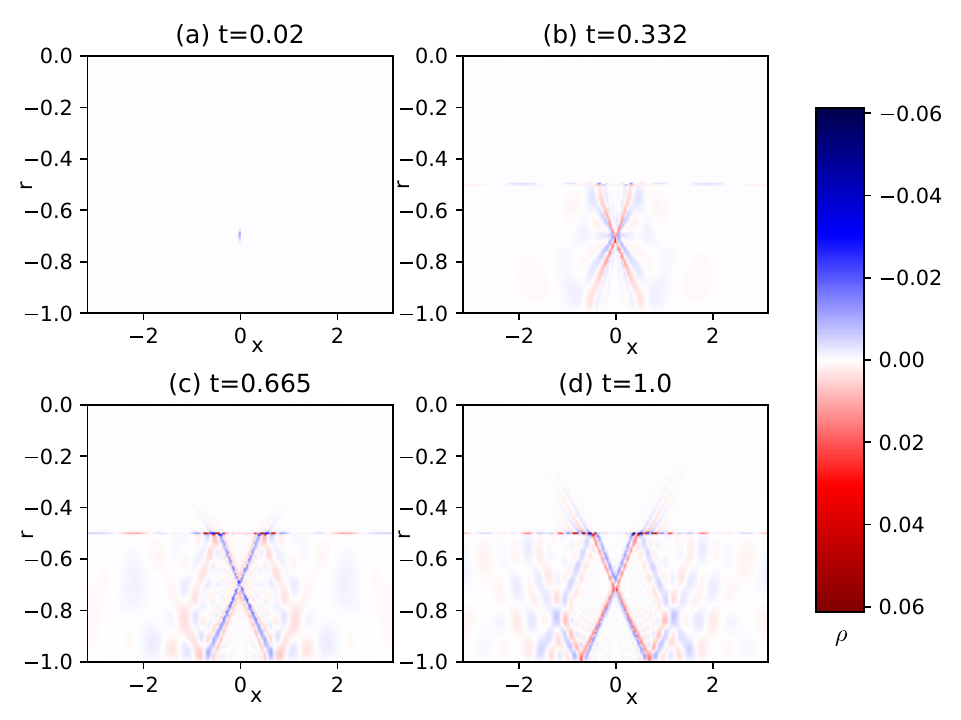}
\captionof{figure}{\label{fig:waveguide:reflection:snapshots}Response of the system \eqref{eqn:modes:finite:forcing} to the localized forcing $F$ defined in \eqref{eqn:def:forcing}. Internal waves propagate following Saint-Andrew's cross pattern, in the lower layer, and are partially reflected on the pycnocline. The fluid starts at rest, with an oscillating localized source term at $x=0$ and $r=-0.7$, visible in $(a)$. Colors show the size of the density deviation $\rho$ from the density profile $\rhob^{(2)}_{\delta}$ with $\delta = 10^{-3}$. }
\end{minipage}
\end{figure}

In Figure \ref{fig:waveguide:snapshots}, the propagation of a deviation of the density from the profile $\rhob^{(2)}_{\delta}$ with $\delta = 5.10^{-2}$ shown in Figure \ref{fig:waveguide:strat} is represented. Except for the density profile, the other parameters are the same as for the simulation shown in Figure \ref{fig:saint_andrew}. The fluid is at rest at the initial time, but is excited by the source term $F$ defined in \eqref{eqn:def:forcing}, localized in space and oscillating in time, visible in Figure \ref{fig:waveguide:snapshots}(a). Then the propagation follows Saint-Andrew's cross pattern in the lower layer, which is stably stratified and with Brunt-Väisälä frequency roughly constant, see Figure \ref{fig:waveguide:snapshots}(b). When the perturbation reaches the pycnocline, it is strongly deviated along the pycnocline, which acts as a waveguide, see Figure \ref{fig:waveguide:snapshots}(c). Finally some of the perturbation is transmitted in the upper layer (see Figure \ref{fig:waveguide:snapshots}(d)), but with a lower intensity. Note also the reflection of the waves against the bottom, visible in Figure \ref{fig:waveguide:snapshots}(c) and (d). However, no reflection against the pycnocline is observed, as the pycnocline is too thick, i.e. the transition between the two layers is not sharp enough. This is in contrast with the simulations from Figure \ref{fig:waveguide:reflection:snapshots}, performed with $\delta = 10^{-3}$, with the same parameters otherwise. Then, the pycnocline is thinner (see Figure \ref{fig:waveguide:strat}), and some reflection against the pycnocline is observed in Figure \ref{fig:waveguide:reflection:snapshots}(c) and (d). This is consistent with the simulations from \cite[Section 5]{GerkemaZimmerman2008} performed in the bilayer setting (corresponding to $\delta = 0$, i.e. when the pycnocline is approximated by a jump in the density profile).

\subsection{Numerical scheme for the irrotational bilayer Euler equations}
\label{subsection:num:bl}
We now explain how we discretize the irrotational bilayer Euler equations \eqref{eqn:interfacial_waves}. Let $\mathtt K \in \N^*$.
\paragraph{\bf Initial conditions and time evolution}
On the one hand, we compute solutions to \eqref{eqn:interfacial_waves} for $k \in \{ -\mathtt K , \dots, \mathtt K\}$, in the case where there is no shear flow (i.e. $v=0$). To this end, let $\zeta_{\ini}, \psi_{\ini}$ be two functions defined on $\mathbb T_L$; we start by computing their first Fourier coefficients $(\hat{\zeta}_k, \hat{\psi}_{k})$, for $|k| \leq \mathtt K$. Then, we \nouveau{could} solve the ODEs \eqref{eqn:interfacial_waves} with initial conditions $(\hat{\zeta}_k, \hat{\psi}_{k})$ with a Runge-Kutta method of order 4, for $|k|\leq \mathtt K$\nouveau{, as for our numerical strategy to solve the stratified Euler equations, described in Section \ref{subsection:num_scheme_description}. However, in view of the explicit formulas \eqref{eqn:interfacial_waves}, it is more efficient to integrate directly, and to get, for $t \geq 0$ and $\Delta_t > 0$:
$$\begin{pmatrix}
\hat{\zeta}_k(t+\Delta_t) \\ \hat{\psi}_k(t+\Delta_t)
\end{pmatrix} = \exp \left( \Delta_t 
\begin{pmatrix}
-ik d(k) & b(k) \\
-a(k) & -ikd(k)
\end{pmatrix} \right) \begin{pmatrix}
\hat{\zeta}_k(t) \\ \hat{\psi}_k(t)
\end{pmatrix},$$
where the exponential of the matrix can be computed numerically.}

\paragraph{\bf Dispersion relation}
Let $k \in \N$. We look for solutions to \eqref{eqn:interfacial_waves}, with $v \neq 0$ a priori, of the form $e^{ik(x-c(k)t)}$. We get that the set of phase velocities of \eqref{eqn:interfacial_waves} at frequency $k$ is
\begin{equation}
\label{eqn:dispersion:ebc:1}
\mathcal C_0(k) \coloneq  \{c(k), \tilde{c}(k)\},
\end{equation}
where $c(k)$ and $\tilde{c}(k)$ are the two (not necessarily distinct) complex solutions to
\begin{equation}
\label{eqn:dispersion:ebc:2}
(c - d(k))^2 = a(k) b(k)/|k|^2.
\end{equation}

\section{The sharp stratification limit}
\label{section:sharp_strat}
In this section we investigate the link between the stratified Euler equations, with a pycnocline of size $\delta > 0$, and the bilayer equations, both with and without a shear flow $\Vb$. 
More precisely, we specify the background stratification to one of the form
\begin{equation}
\label{eqn:def:sharp_strat}
\rhob_{\delta}(r) \coloneq  \rho_- + (\rho_+ - \rho_-) \left(\frac{1}{\pi}\arctan((r-r_*)/\delta) + \frac12\right),
\end{equation}
for some $r_* \in (-H,0)$, $\delta > 0$ and $0 < \rho_+ < \rho_-$. As $\delta \to 0$, the density \nv{approaches} the constants $\rho_+$ in the upper layer ($r > r_* + \delta$) and $\rho_-$ in the lower layer ($r < r_* - \delta$), and varies continuously but \nv{rapidly} in the pycnocline ($r \in [r_*-\delta, r_* + \delta]$)\nouveau{, see Figure \ref{fig:sharp_strat}}. The statements in Section \ref{subsection:sharp_strat:sans_shear} extend to more general profiles, see Remark \ref{rk:rhobdelta_plus_gen}. We then study \eqref{eqn:euler:isopycnal:Vb} with $\rhob = \rhob_{\delta}$.\\
\begin{figure}
\includegraphics[width=0.4\textwidth]{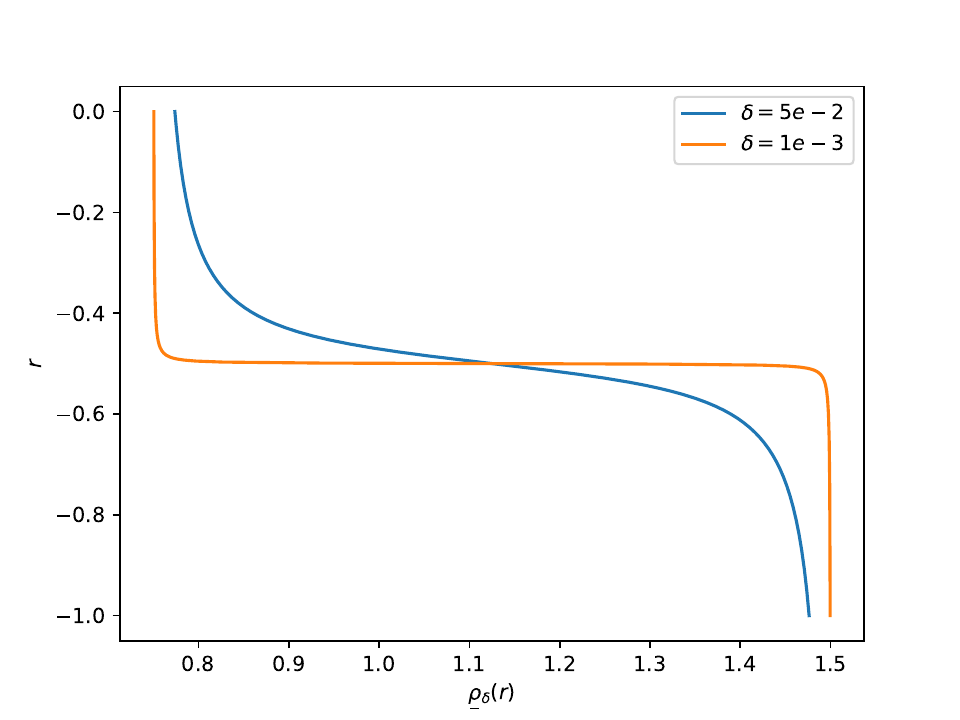}
\caption{The sharp stratification profile $\rhob_{\delta}$, used in this section. Here, $\rho_+ = 0.75$, $\rho_- = 1.5$. }
\label{fig:sharp_strat}
\end{figure}

We start with the case $\Vb =0$ in Section \ref{subsection:sharp_strat:sans_shear}. In this setting, we first show that from the system \eqref{eqn:euler:isopycnal:no_shear}, we can derive suitable energy estimates, which are uniform in $\delta$. This allows us to compare the solutions of \eqref{eqn:euler:isopycnal:no_shear} in this setting with the solutions of the linear bilayer Euler equations without shear flow \eqref{eqn:euler:bilayer} in Proposition \ref{prop:cvce}. The latter equations are well-posed (recall Proposition \ref{prop:nrj:bilayer}). The proofs of these results are postponed to  Appendix~\ref{section:apdx:preuve}. We also illustrate numerically this convergence, both on the solution itself but also on the modes $f_n$ and their speed $c_n$ defined in Section \ref{section:normal_modes}. \\
In Section \ref{subsection:sharp_strat:shear}, we specify a continuous profile $\Vb_{\delta}$ that approximates a piecewise constant shear flow, see \eqref{eqn:def:sharp:shear}. Then, one might expect that the system \eqref{eqn:euler:isopycnal:Vb}
 is well approximated by the bilayer Euler equations linearized around a shear flow. However, these equations are ill-posed in Sobolev spaces due to the Kelvin-Helmholtz instabilities (see for instance \cite{IguchiTanakaTani1997}, \cite{LebeauKamotski2005} for a setting close to ours, as well as \cite[Chapter A, Section 3.5]{Duchene2022a} and references within). \nv{Therefore}, the analysis of Section \ref{subsection:sharp_strat:sans_shear} for comparing solutions of both systems cannot be performed in Section \ref{subsection:sharp_strat:shear}. We thus study numerically the dispersion relation of \eqref{eqn:euler:isopycnal:Vb} with the help of the modal decomposition defined in Section \ref{subsection:normal_modes}.
Finally, we emphasize in Section \ref{subsection:sharp_strat:shear:csq} that the numerical evidence displayed in Section \ref{subsubsection:sharp_strat:shear:quantitative} prevents the standard method from Section \ref{subsection:sharp_strat:sans_shear} in the case $\Vb =0$ to be applied in the case $\Vb =\Vb_{\delta}$ to justify bilayer models.


\subsection{The stable case : $\Vb = 0$}
\label{subsection:sharp_strat:sans_shear}
In this setting of a sharp stratification, one expects that the motion of the fluid density can be described solely by the motion of the interface between the upper layer and the lower layer, whose evolution follows the bilayer Euler equations, see Section \ref{subsection:ebc:no_shear}. In the present section, we study the convergence from the stratified Euler equations to the bilayer equations. We state in Proposition \ref{prop:cvce} that the solution of the bilayer Euler equations is a good approximation of the stratified Euler equations \eqref{eqn:euler:isopycnal:no_shear} relative to the small parameter $\delta$, which is the size of the pycnocline. This quantitative comparison of solutions of the stratified Euler equations and the linear bilayer Euler equations is referred to as {\it full justification} - following for instance the nomenclature detailed in \cite[Appendix C.1]{Lannes2013}. The proof of Proposition \ref{prop:cvce} is postponed to Appendix \ref{section:apdx:preuve}.\\

Let us mention once again that we assume that there is no shear flow ($\Vb = 0$) in this section. The reason is that the bilayer Euler equations linearized around a non-zero shear flow are ill-posed in the present setting of Sobolev spaces, see Section \ref{subsection:sharp_strat:shear}. This restricted framework is also the reason why the well-posedness of the equations \eqref{eqn:euler:bilayer} and \eqref{eqn:euler:isopycnal:no_shear} is global, as opposition to the local well-posedness results for the non-linear stratified Euler equations from \cite{Desjardins2019, Duchene2022,Fradin2024}.\\

The following result is the convergence of the solution of \eqref{eqn:euler:isopycnal:no_shear} towards the solution of \eqref{eqn:euler:bilayer}, both stemming from the same initial data, as $\delta \to 0$. Indeed, in Proposition \ref{prop:nrj} on the well-posedness of the stratified Euler equations \eqref{eqn:euler:isopycnal:no_shear}, we allow $V$ to be discontinuous so that the initial data \eqref{eqn:euler:bilayer:ci} are admissible for the system \eqref{eqn:euler:isopycnal:no_shear}. In order to compare both systems, we define 
\begin{equation}
\label{eqn:def:etastar}
\eta_* \coloneq  \eta_{r=r_*}.
\end{equation}
We now state the following proposition, whose proof is postponed to Appendix \ref{subsection:strat_to_bilayer}.
\begin{proposition}
\label{prop:cvce}
Let $s \geq 0$,  $(V_{\ini},w_{\ini}, \eta_{\ini}) \in H^{s+3,0}(S) \times H^{s+3,1}(S) \times H^{s+3,1}(S)$ satisfying the incompressibility condition in \eqref{eqn:euler:isopycnal:no_shear} and the boundary conditions \eqref{eqn:euler:iso:bc}. Then there exists a unique solution 
$$(V_{\delta},w_{\delta},\eta_{\delta}) \in C^0([0,\infty), H^{s+3,0}(S) \times H^{s+3,1}(S) \times H^{s+3,1}(S))$$
to \eqref{eqn:euler:isopycnal:no_shear} with the density profile $\rhob_{\delta}$ given in \eqref{eqn:def:sharp_strat}, with boundary conditions \eqref{eqn:euler:iso:bc} and initial conditions $(V_{\ini},w_{\ini}, \eta_{\ini})$.\\

Moreover, let $\zeta_{\ini} \coloneq  \eta_{\ini |r=r_*}$, and 
$$(V_{\bl},w_{\bl},\zeta) \in C^0([0,\infty),H^{s+3,0}(S)\times H^{s+3,1}(S) \times H^{s+3}(\mathbb T_L))$$
be the solution of the bilayer Euler equations \eqref{eqn:euler:bilayer} with boundary conditions \eqref{eqn:euler:bilayer:bc}, with the initial data $(V_{\ini},w_{\ini}, \zeta_{\ini})$ \nv{provided} by Proposition \ref{prop:nrj:bilayer}. Then there exists $C > 0$, which depends only on an upperbound of the $H^{s+3}$ norm of $(V_{\ini},w_{\ini}, \eta_{\ini})$, as well as $\rho_-, \rho_+,H,L$, such that for any $t \geq 0$:
\begin{equation}
\label{eqn:erreur:bl}
\Vert (V_{\delta} - V_{\bl})(t,\cdot), (w_{\delta} - w_{\bl})(t,\cdot)\Vert_{H^{s,0}} + \vert \zeta(t,\cdot) - \eta_{\delta|r=r_*}(t,\cdot) \vert_{H^{s}} \leq C \delta^{\frac12} |\log \delta | (1+t).
\end{equation}
\end{proposition}
\begin{remark}
Proposition \ref{prop:cvce} holds for more general profiles than $\rhob_{\delta}$ given in \eqref{eqn:def:sharp_strat}, see Remark \ref{rk:rhobdelta_plus_gen}.
\end{remark}
\begin{remark}
\label{rk:cvce:perte}
The constant $C$ in Proposition \ref{prop:cvce} depends on an upper bound on the $H^{s+3,0}(S)$-norm of the unknowns $V,w,\eta$ and $V_{\bl}, w_{\bl}, \eta_{\bl}$ and the $H^s(\mathbb T_L)$-norm of $\zeta$. However, we only estimate the difference of the two solutions in $H^{s,0}$ norm. This is due to the fact that we reformulate the systems \eqref{eqn:euler:isopycnal:no_shear} and \eqref{eqn:euler:bilayer} in such a way that they share a common structure, see \eqref{eqn:euler:diff}. We then allow a loss of derivatives in order obtain a better dependency of the estimates with respect to $\delta$. 
\end{remark}

\nv{We now} illustrate the convergence stated in Proposition \ref{prop:cvce}. In Figure \ref{fig:strat_to_bc:1e2}, we plot the density computed from the stratified Euler equations with the profile $\rhob_{\delta}$ described in \eqref{eqn:def:sharp_strat}, with $ \delta = 10^{-2}$ (as well as $r_*=-0.5$, $\rho_- = 1.5$, $\rho_+=0.75$), for four different times. We use $\mathtt N=40$ vertical modes, computed on a vertical grid with $\mathtt m_r=50001$ points, and  $\mathtt m_x=32$ Fourier modes, on a time interval $[0,T]$ with $T=10$ discretized with $\mathtt m_t = 2000$ points. We use $H=L=1$. For the initial data, we consider no initial velocity, and choose for $\eta_{\ini}$:
\begin{equation}
\label{eqn:eta_in}
\eta_{\ini}(x,r) = 0.4r(r+1)\left((\frac{2}{\pi}(\arctan((x+0.8)/0.2)-1)\mathbb{1}_{\R_-}(x) + (-1-\frac{2}{\pi}\arctan((x-0.8)/0.2))\mathbb{1}_{\R_+}(x)  +2)\right),
\end{equation}
where $\mathbb{1}_{\R_+}(x) = 1$ if $x > 0$ and $0$ otherwise (and $\mathbb{1}_{\R_-}(x) = 1- \mathbb{1}_{\R_+}(x)$).
We superimpose the interface computed from the irrotational bilayer Euler equations, computed with the method described in Section \ref{subsection:num:bl}. We choose zero initial condition for the velocity potential $\psi$, and 
$$ \zeta_{\ini} \coloneq  \eta_{\ini|r=-0.5}.$$

\nv{We} define the error between the interface $\zeta$ computed from the bilayer Euler equations and the isopycnal $\eta_*$ of density $\rhob(r_*)$, computed from the stratified Euler equations with $\delta > 0$\nv{, by}
\begin{equation}
\label{eqn:error:bc}
\mathrm{Err} \coloneq  \Vert \zeta - \eta_* \Vert_{L^{\infty}([0,T]\times \mathbb{T}_L)},
\end{equation}
for some $T > 0$. Note that our numerical strategy defines (the approximations of) $\zeta$ and $\eta_*$ as piecewise constant in time and continuous in space functions, so that we can compute $\mathrm{Err}$ using \eqref{eqn:error:bc} numerically. In Figure \ref{fig:erreur_bc}, we plot the logarithm of Err versus the logarithm of $\delta$ for $\delta$ between $10^{-2}$ and $6.10^{-3}$. Apart from $\delta$, the other parameters are the same as for Figure \ref{fig:strat_to_bc:1e2}. We observe a clear decrease of the error, \nouveau{at a faster rate than the theoretical upper bound from Proposition \ref{prop:cvce}. However, independently of the (high enough) number of modes and time and space resolutions, the points are not well aligned so that the measured rate (roughly $\delta^{0.59}$) should be only be interpreted as an indicative value. Moreover, the apparent discrepancy between the measured rate ($\delta^{0.59}$) and the theoretical upper bound from Proposition \ref{prop:cvce} ($\delta^{0.5} |\log \delta |$) may reflect either that the initial data \eqref{eqn:eta_in} does not realize the optimal convergence rate, or that the upper bound from Proposition \ref{prop:cvce} is not optimal.}\\

\nv{We} now give some numerical results on the modes from the Sturm-Liouville problem \eqref{eqn:SL}. The speeds $c_1,c_2,c_3$ associated to the first three modes are computed in Figure \ref{fig:c}, for $\delta$ between $10^{-3}$ and $10^{-4}$. We observe that $c_1$ is independent of $\delta$, in accordance with \cite[Proposition 6]{Desjardins2019}. It appears that $c_n$ decays as $\sqrt{\delta}$ for $n \geq 2$.   
\begin{figure}
\begin{minipage}[t]{0.59\textwidth}
\vspace*{0pt}
\includegraphics[width = \textwidth]{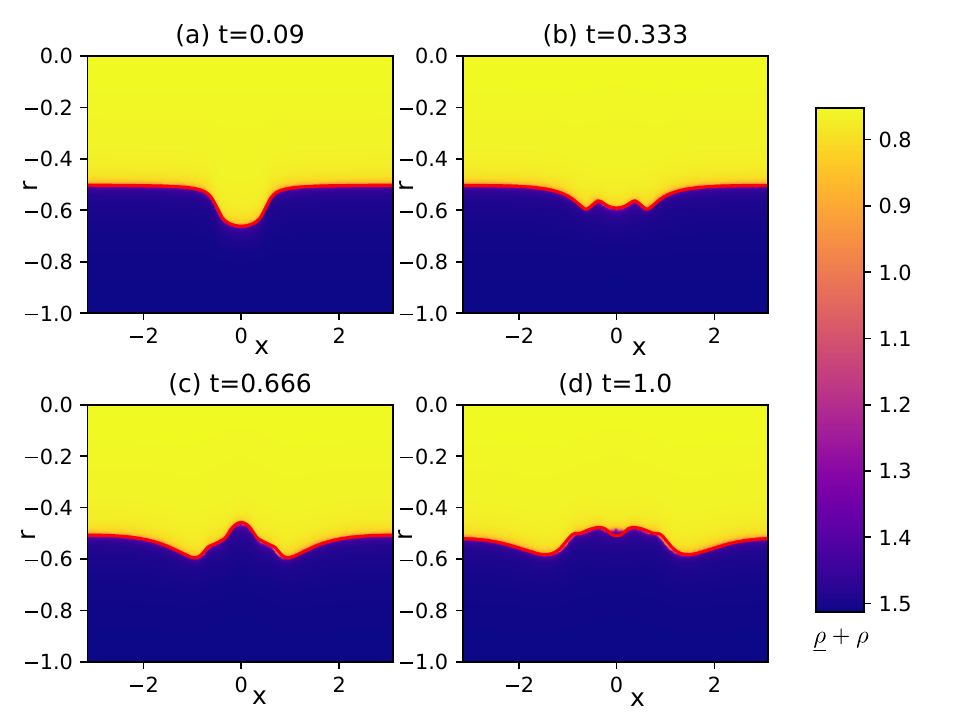}
\captionof{figure}{\label{fig:strat_to_bc:1e2}Comparison of the solutions of the linear stratified Euler equations \eqref{eqn:euler:isopycnal:no_shear} with $\Vb = 0$ and with the sharp stratification profile \eqref{eqn:def:sharp_strat} with $\delta = 10^{-2}$ and the bilayer Euler equations \eqref{eqn:euler:bilayer}. The colors show the value of the density computed from the solution of the stratified Euler equations, and the red line is the interface computed from the bilayer Euler equation. The interface is an approximation of the diffuse interface, when $\delta$ is small enough.}
\end{minipage}
\begin{minipage}[t]{0.39\textwidth}
\vspace*{0pt}
\includegraphics[width = \textwidth]{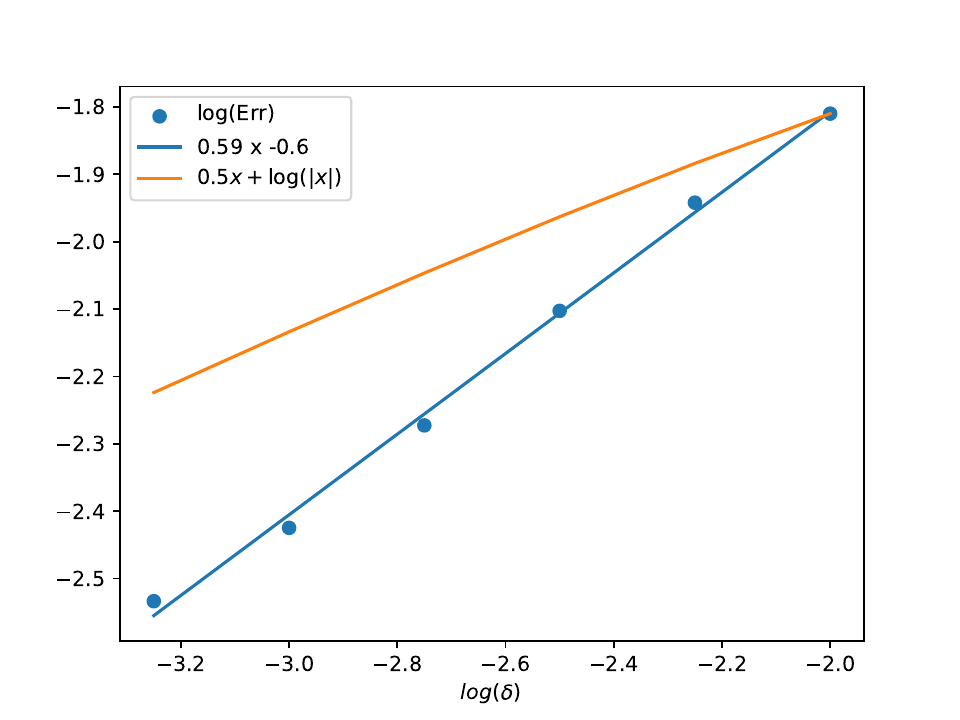}
\captionof{figure}{\label{fig:erreur_bc} $\log(\mathrm{Err})$ versus $\log(\delta)$\nouveau{, compared with the theoretical upper bound from Proposition~\ref{prop:cvce}}.}
\end{minipage}
\end{figure}
\begin{figure}
\includegraphics[width = 0.6\textwidth]{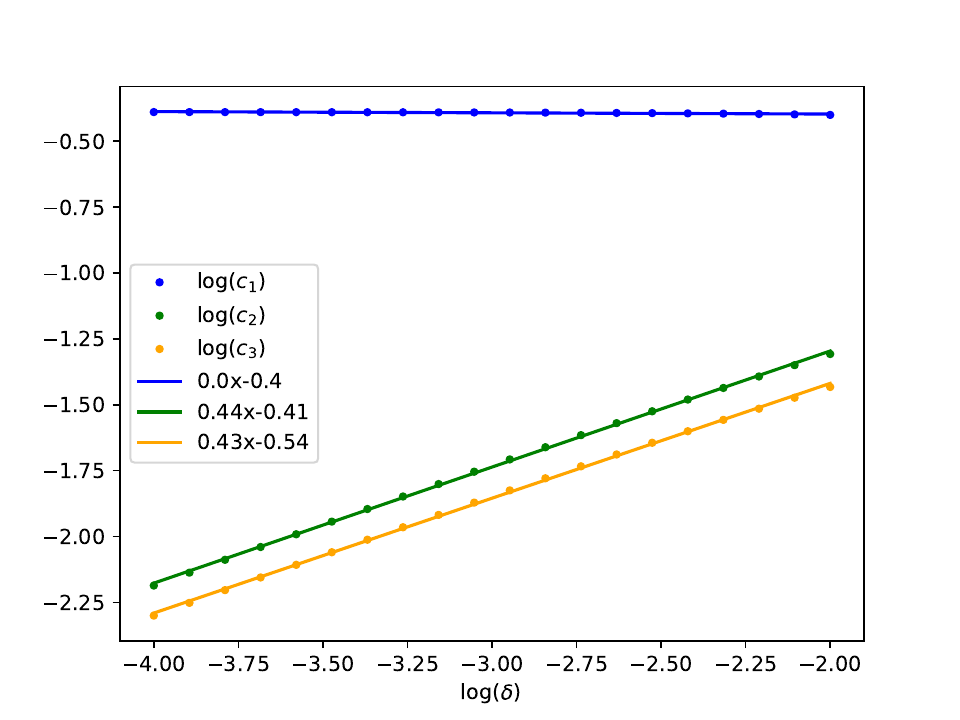}
\caption{\label{fig:c} $c_n$ versus $\log(\delta)$, for $n \in \{1,2,3\}$. The first mode speed $c_1$ does not decay to $0$, but $c_2$ and $c_3$ seem to go to zero as $\sqrt{\delta}$.}
\end{figure}
\FloatBarrier
\subsection{The unstable case : $\Vb = \Vb_{\delta}$}
\label{subsection:sharp_strat:shear}
In the presence of a shear flow, the bilayer Euler equations are ill-posed in Sobolev spaces, due to the Kelvin-Helmholtz instabilities (see for instance \cite{IguchiTanakaTani1997}, \cite{LebeauKamotski2005} for a setting close to ours, as well as \cite[Chapter A, Section 3.5]{Duchene2022a} and references within). In the irrotational case, the dispersion relation can be computed analytically, see \eqref{eqn:dispersion:ebc:1} and \eqref{eqn:dispersion:ebc:2} for an expression of the phase velocities. We show these phase velocities in Figure \ref{fig:dispersion:ebc}, with $\rho_-=1.5,$ $\rho_+ = 0.75$ and $v=0.25$. From the dispersion relation \eqref{eqn:dispersion:ebc:1} - \eqref{eqn:dispersion:ebc:2} of \eqref{eqn:interfacial_waves}, we can deduce the following result.

\begin{proposition}
\label{prop:dispersion:ebc}
Let $\rho_- > \rho_+ > 0$, and $v > 0$. Then the sets of the phase velocities $\mathcal{C}_0(k)$ for $k \in \frac{1}{L} \mathbb Z$ of \eqref{eqn:euler:bilayer:shear} (see \eqref{eqn:dispersion:ebc:1} and \eqref{eqn:dispersion:ebc:2})  satisfy the following statements, for some $k_{\min,\mathrm{BL}} \geq 0$.
\begin{itemize}
\item \underline{Stability of the low frequencies.} Let $|k|\leq k_{\min,\mathrm{BL}}$, then 
	$$ \Im(c(k)) = 0, \text{ for any } c(k) \in \mathcal C(k).$$
\item \underline{Instability of the high frequencies.}  Let $|k| > k_{\min,\mathrm{BL}}$. Then the two elements $c(k), \tilde{c}(k)$ in $\mathcal C(k)$ are complex conjugates with non-zero imaginary part; by convention we denote by $c(k)$ the element of $\mathcal C(k)$ with positive imaginary part. Then
	$$ \Im(c(k)) \underset{|k| \to \infty}{\longrightarrow} \frac{\sqrt{\rho_-\rho_+}}{\rho_-+\rho_+} v > 0.$$ 
\end{itemize}
\end{proposition}
The instability of the high-frequencies is called the Kelvin-Helmholtz instability.\\ 

\nouveau{
The presence of the Kelvin-Helmholtz instabilities in a continuously stratified fluid with a thin pycnocline has already been invested, both with laboratory experiments (see for instance \cite{Grue1999}), and with direct numerical simulations at the non-linear level: in \cite{AlmgrenCamassaTiron2012}, the authors study the emergence of Kelvin-Helmholtz roll-ups, starting from a perturbation of a solitary wave. They also include a numerical study of the Taylor-Goldstein equation for the stratified Euler equations, linearized around the density and velocity profiles of the solitary wave that they consider. However, they do not provide a complete description of the dispersion relation of the stratified Euler equation, and in particular do not describe its dependency in the thickness of the pycnocline $\delta$, which is the main goal of the present section. See also \cite{CaulfieldVieweg2025}, where the authors perform direct numerical simulations for the non-linear density-dependent Euler equations, starting from a stably stratified continuous density profile, with a thin pycnocline. Focusing on the large-time dynamics, the authors exhibit a thickening of the pycnocline over time. \\
}

\begin{figure}
\includegraphics[width = 0.6\textwidth]{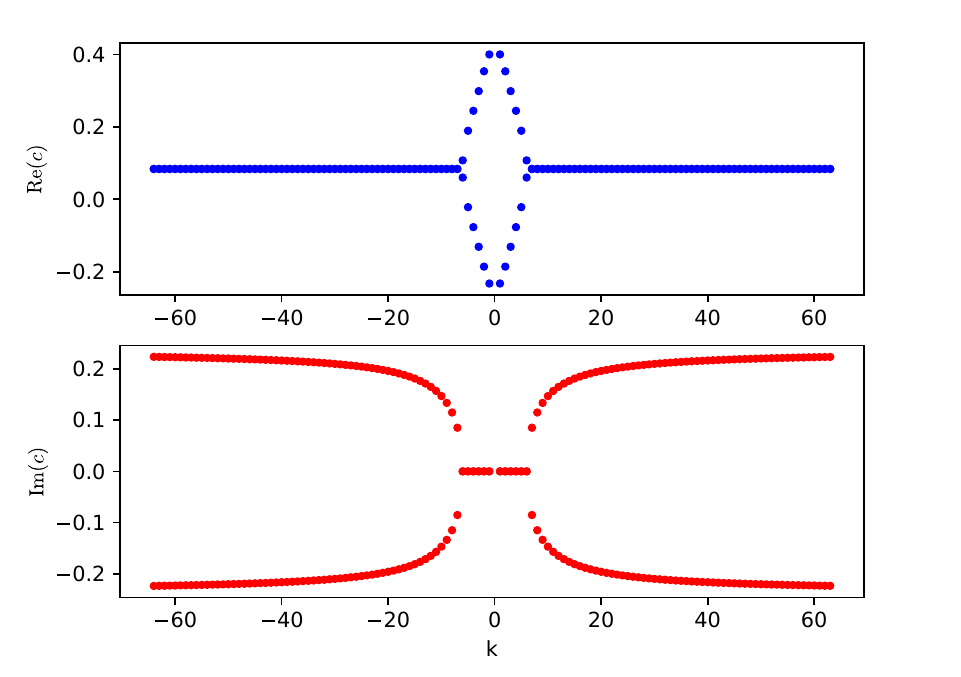}
\caption{Dispersion relation of the irrotational bilayer Euler equations with a non-zero shear flow. We plot the phase velocities, i.e. elements of $\mathcal C(k)$ versus $k$.}
\label{fig:dispersion:ebc}
\end{figure}
The goal of this section is to obtain a numerical approximation of the dispersion relation in phase-velocity formulation of the linear stratified Euler equations with a shear flow \eqref{eqn:euler:isopycnal:Vb} (see Definition \ref{def:dispersion_relation}), and to compare the result with Proposition \ref{prop:dispersion:ebc}. To this end, we define the sharp shear flow as
\begin{equation}
\label{eqn:def:sharp:shear}
\Vb_{\delta}(r) \coloneq   \frac{v}{2\pi} \arctan((r-r_*)/(0.5\delta)),
\end{equation}
with $v > 0$ and $\delta > 0$, so that $\Vb_{\delta}$ is roughly constant equal to $v$ (resp. $-v$) in the upper (resp. lower) layer, with a continuous but fast transition in the pycnocline, roughly of size $\delta$. In all the simulations shown in this section, we use $v=0.25$, $\rho_- = 1.5$, $\rho_+=0.75$, $H=L=1$, $r_* = -0.5$, and the profiles $\Vb_{\delta}$ and $\rhob_{\delta}$ defined in \eqref{eqn:def:sharp:shear} and \eqref{eqn:def:sharp_strat}.
\subsubsection{Qualitative description of the relation dispersion in the unstable case}
\label{subsubsection:sharp_strat:shear:qualitative}
In the presence of a non-zero shear flow $\Vb$, the evolution of the modes $(V_n,\eta_n)$ is given by \eqref{eqn:modes:couple:shear}. \nv{We} use the numerical strategy described in Section \ref{subsection:num_scheme_description} to find an approximation of the dispersion relation (more precisely, of the phase velocities, see Definition \ref{def:dispersion_relation}) of the stratified Euler equations \eqref{eqn:euler:isopycnal:Vb}, which can be seen in Figure \ref{fig:dispersion:strat}. We use $\mathtt N = 40$ vertical modes, computed with $\mathtt m_r = 150 \ 000$ grid points, and $\mathtt K = 128$ Fourier modes. See also Figure \ref{fig:dispersion:long}, which shows the same simulation as in Figure \ref{fig:dispersion:strat} with $\delta = 10^{-2}$, only with $\mathtt K = 450$. \nv{In} Figure \ref{fig:dispersion:superimposed}, we superimpose the imaginary part of the dispersions relations displayed in Figure \ref{fig:dispersion:strat} for $\delta = 5.10^{-2}$ and $\delta = 10^{-2}$ with the imaginary part of the dispersion relation of the bilayer Euler equations, displayed in Figure \ref{fig:dispersion:ebc}. \\

\nv{We} now describe the dispersion relations shown in Figure \ref{fig:dispersion:strat}. First, we observe that, regardless of the value of $\delta$ or $k$, any real number in $[-v,v]$ (that is, in the range of $\Vb_{\delta}$) is in $\mathcal C_{\delta}(k)$. This is linked with the presence of continuous spectrum for \eqref{eqn:euler:isopycnal:Vb}, which is already known (see \cite[Section 3]{CotiZelatiNualart2025} for a complete, analytical description of the spectrum of \eqref{eqn:euler:isopycnal:Vb}, and some additional references on the spectral study of \eqref{eqn:euler:isopycnal:Vb}. Note however that this latter work is restricted to the density profile $\rhob(r) = \rho_0 -\beta r$ for $\rho_0,\beta > 0$ some constants, and $\Vb(r) = r$, in the Boussinesq approximation). \\

\nv{Focusing} on the imaginary part of the phase velocities in Figure \ref{fig:dispersion:strat}\nv{, we} find that there exists a threshold $k_{\min,\delta}$ such that for $|k|\leq k_{\min,\delta}$, the imaginary parts of the elements in $\mathcal{C}_{\delta}(k)$ is small - more precisely, it seems to go to $0$ when $\delta \to 0$; let us call this configuration {\it asymptotic stability of the low frequencies}. \\

\nv{We} now turn to the behavior of the high frequencies. To this end, we show in Figure \ref{fig:dispersion:long} the results of the same simulation as in Figure \ref{fig:dispersion:strat} with $\delta = 10^{-2}$, although with $\mathtt K = 450$, to study the behavior of the large frequencies. We observe that there exists $k_{\max,\delta} > 0$ such that, for $k \in [k_{\min,\delta}, k_{\max,\delta}]$, there exist elements with a non-zero imaginary part in $\mathcal{C}_{\delta}(k)$. Moreover, we observe that
$$k_{\max,\delta} \underset{\delta \to 0}{\longrightarrow} +\infty.$$ 
This is consistent with the Kelvin-Helmholtz instabilities described in the bilayer case, where the high-frequencies are unstable. We conclude that the Kelvin-Helmholtz instabilities are present in the dispersion relation of the stratified Euler equations, with sharp density and shear profiles.\\

It should be noted that, in Figures \ref{fig:dispersion:strat} and \ref{fig:dispersion:long}, the frequencies $k$ with $|k| \geq k_{\max,\delta}$ are stable, i.e. elements of $\mathcal{C}_{\delta}(k)$ have zero imaginary part, up to numerical inconsistencies. This is in sharp contrast with the bilayer setting, where all the frequencies $|k| \geq k_{\min,\mathrm{BL}}$ are unstable. 
\begin{figure}
\includegraphics[width = \textwidth]{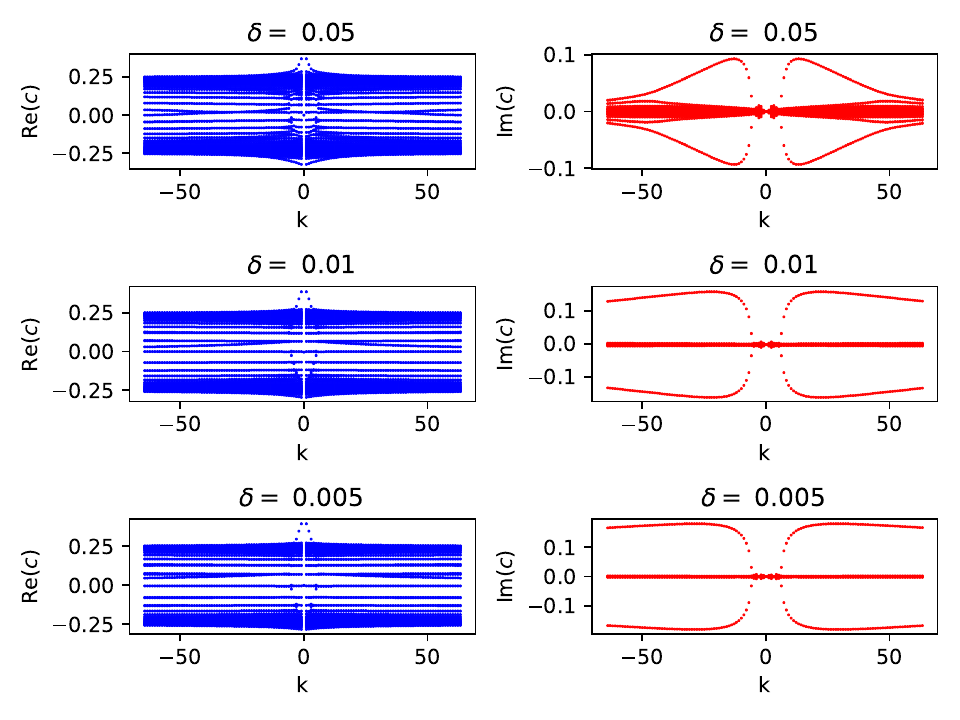}
\caption{Dispersion relations for the stratified Euler equations with a shear flow of the form of \eqref{eqn:def:sharp:shear}, obtained numerically. We plot the numerical approximation of phase velocities, i.e. elements of $\mathcal C^{(\mathtt N)}_{\delta}(k)$ versus $k$.} 
\label{fig:dispersion:strat}
\end{figure}
\begin{figure}
\includegraphics[width = 0.8\textwidth]{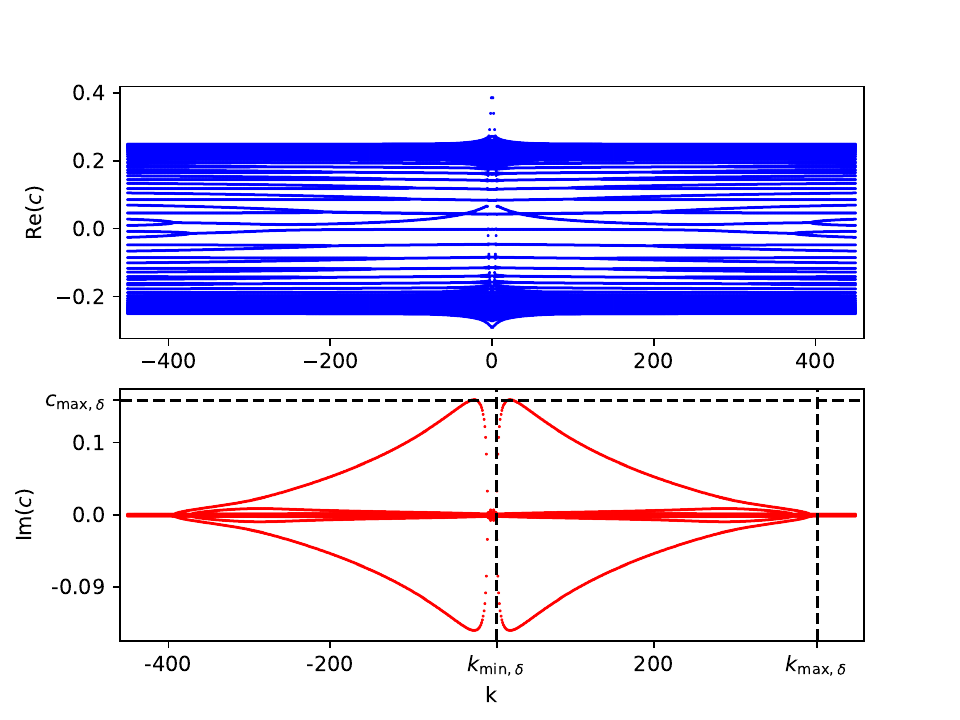}
\caption{Dispersion relation for $\delta = 10^{-2}$ for the stratified Euler equations obtained numerically. We plot the numerical approximation of phase velocities, i.e. elements of $\mathcal C^{(\mathtt N)}_{\delta}(k)$ versus $k$. After the threshold $k_{\max,\delta}$, the high frequencies are stable.}
\label{fig:dispersion:long}
\end{figure}
\begin{figure}
\includegraphics[width = 0.8\textwidth]{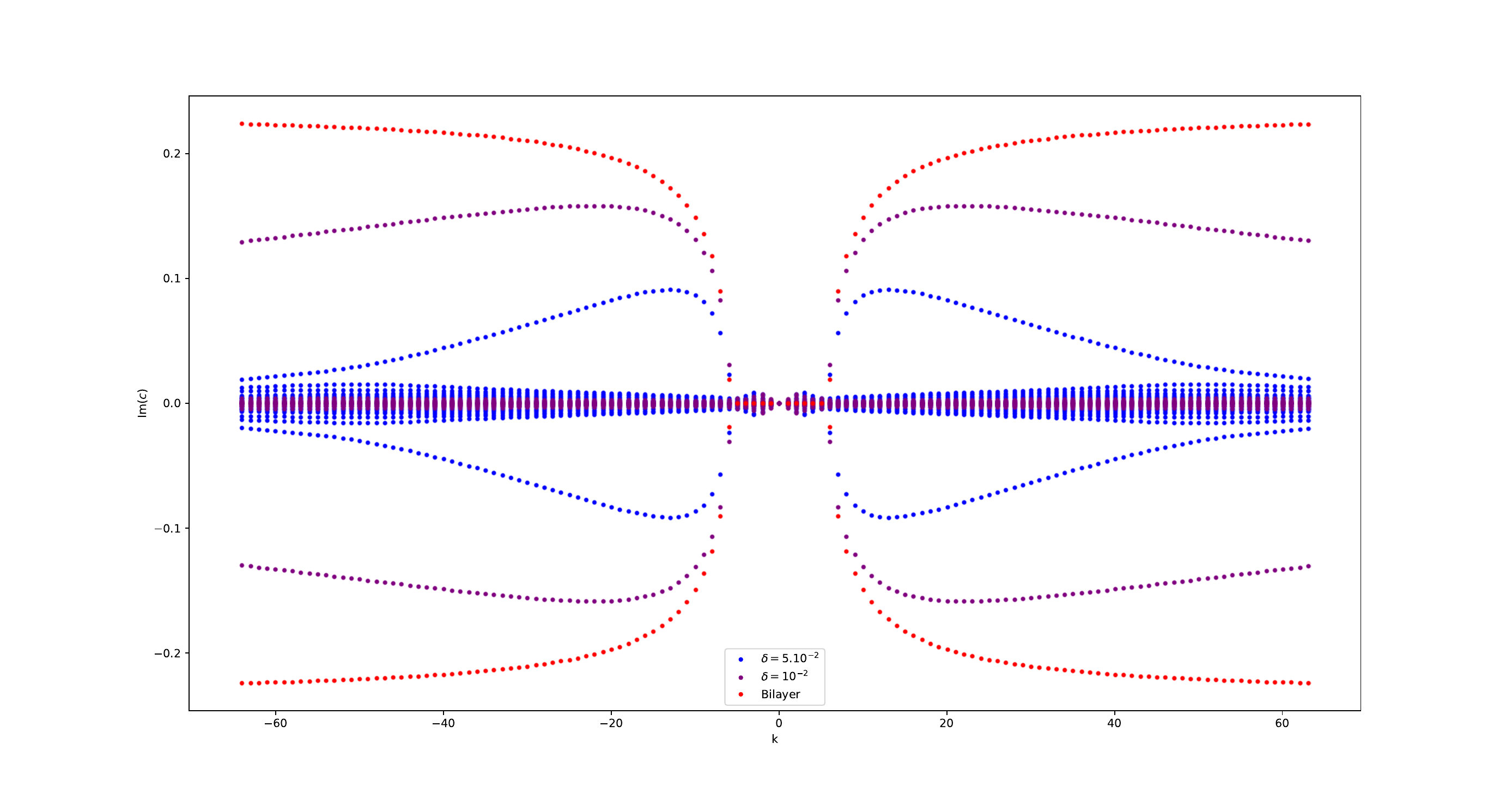}
\caption{Imaginary part of the dispersion relation for $\delta = 5.10^{-2}$ and $\delta = 10^{-2}$ for the stratified Euler equations, superimposed with the imaginary part of the dispersion relation of the bilayer Euler equations. We plot the phase velocities, i.e. elements of $\mathcal C(k)$ versus $k$.}
\label{fig:dispersion:superimposed}
\end{figure}
\nouveau{
\begin{remark}
Several instabilities are present in the dispersion relations from Figures \ref{fig:dispersion:strat}, \ref{fig:dispersion:long}. The strongest instabilities are the Kelvin-Helmholtz instabilities, which are our main focus. We provide in Appendix \ref{section:apdx:approx} some consistency checks towards the fact that these instabilities are well approximated by our numerical strategy. Additionally, the presence of a stable (but poorly resolved) continuous spectrum for \eqref{eqn:euler:isopycnal:Vb} induces small instabilities, due to numerical inconsistencies. These instabilities are artificial, but can clearly be distinguished from the Kelvin-Helmholtz instabilities in Figures \ref{fig:dispersion:strat}, \ref{fig:dispersion:long}. See however the paragraph {\it The upper threshold $k_{\max,\delta}$ for the Kelvin-Helmholtz instabilities} below. Finally, some additional instabilities are present, in particular in the low-frequency region. We conjecture that these instabilities disappear in the sharp stratification limit ($\delta \to 0$). However, when $\delta$ is too small, these instabilities become of size comparable with the numerical inconsistencies due to the presence of continuous spectrum, preventing a quantitative numerical assessment of their persistence - see the paragraph {\it Asymptotic stability of the low frequencies} below. 
\end{remark}
}
\subsubsection{Quantitative description of the dispersion relation in the unstable case}
\label{subsubsection:sharp_strat:shear:quantitative}
In this section we provide quantitative results on the features of the dispersion relation displayed in Figure \ref{fig:dispersion:long}.\\

\paragraph{\bf The upper threshold $k_{\max,\delta}$ for the Kelvin-Helmholtz instabilities.} As shown in Figure \ref{fig:dispersion:long}, there is a threshold frequency $k_{\max,\delta}$ after which the higher frequencies are stable, up to numerical inconsistencies. We now investigate numerically the dependency of $k_{\max,\delta}$ with respect to $\delta$. First, we compute the dispersion relation of \eqref{eqn:euler:isopycnal:Vb} for several $\delta \in [10^{-3}, 5.10^{-2}]$. We use $\mathtt K = 1200$, so that $k_{\max,\delta} \leq \mathtt K$ for the values of $\delta$ that we consider. We use $\mathtt m_r = 200 \ 000$ points on the vertical axis. We perform two sets of simulations, one with $\mathtt N = 80$ vertical modes, and one with $\mathtt N = 90$ vertical modes.\\

Some elements $c$ of the dispersion relation with $\mathrm{Re}(c)$ in the range of $\Vb$ (that correspond to continuous spectrum of the stratified Euler equations \eqref{eqn:euler:isopycnal:Vb}) have zero imaginary part but are computed with less precision (see Appendix \ref{section:apdx:approx}), so that there is a band of size $2.10^{-3}$ in which we cannot distinguish the Kelvin-Helmholtz instabilities from these numerical inconsistencies. This is why we set the threshold $3.10^{-3}$, above which the imaginary part of the computed phase velocities coming from the dispersion relation of \eqref{eqn:euler:isopycnal:Vb} and numerical inconsistencies can be distinguished. Then we find the highest frequency $k_{\max,\delta}$, for which an element of $\mathcal C_{\delta}(k_{\max,\delta})$ has an imaginary part larger than $3.10^{-3}$; see Figure \ref{fig:dispersion_k_delta} for a representation of this process. We then plot in Figure \ref{fig:k_delta} the obtained values for $k_{\max,\delta}$. There are two estimated $k_{\max,\delta }$ for each value of $\delta$, corresponding to two number of modes ($N=80$ and $N=90$) used to compute the dispersion relation. We conjecture that 
\begin{equation}
\label{eqn:conj:kmax}
k_{\max,\delta} \approx \delta^{-1}.
\end{equation} 
\begin{figure}
\centering
\begin{minipage}{0.49\textwidth}
\includegraphics[width = \textwidth]{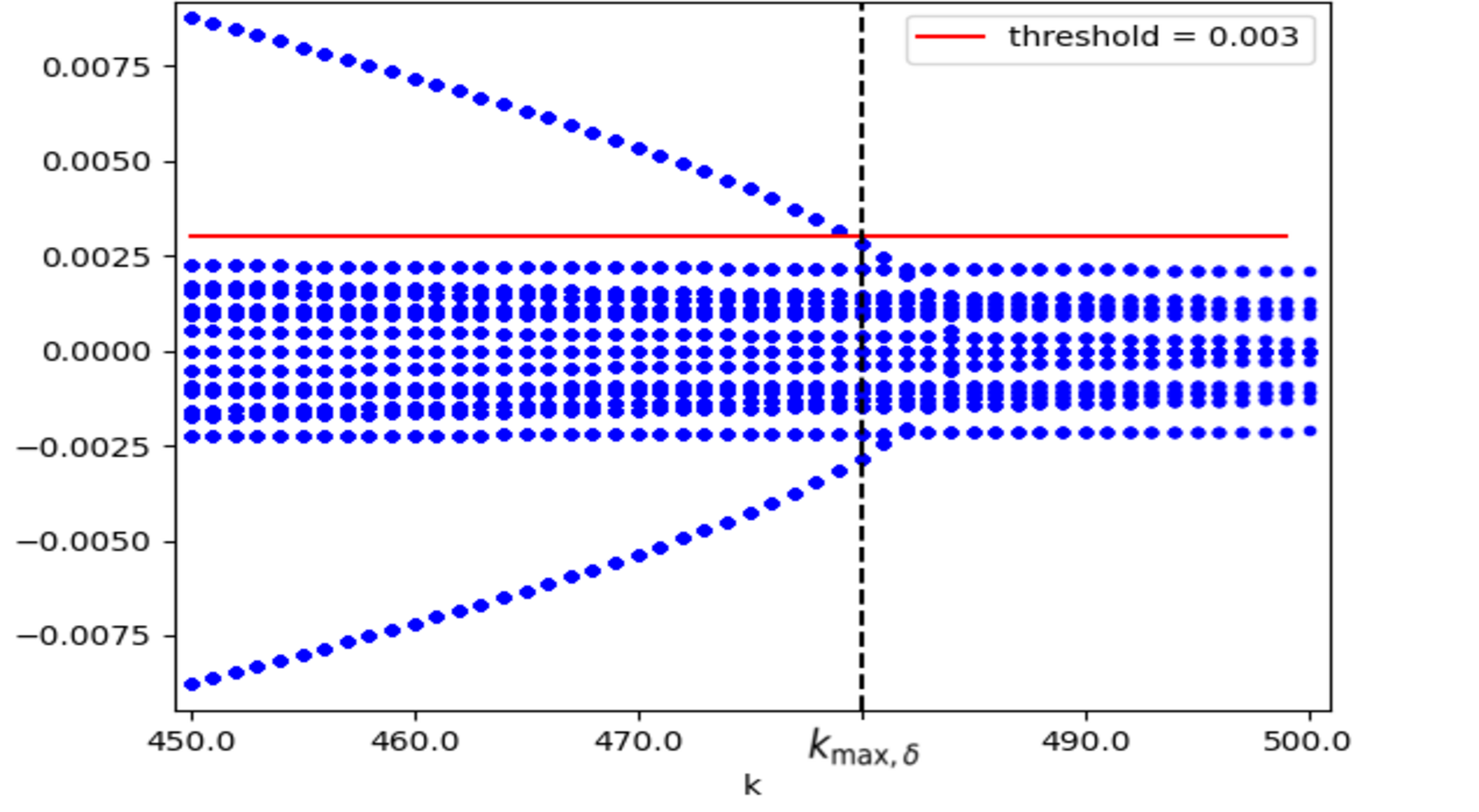}
\caption{Dispersion relation with $\mathtt N=80$, $\delta = 8.10^{-3}$. With the threshold $3.10^{-3}$, we get $k_{\max,\delta} = 480$. We plot the numerical approximation of phase velocities, i.e. elements of $\mathcal C^{(\mathtt N)}_{\delta}(k)$ versus $k$.}
\label{fig:dispersion_k_delta} 
\end{minipage}
\begin{minipage}{0.49\textwidth}

\includegraphics[width = \textwidth]{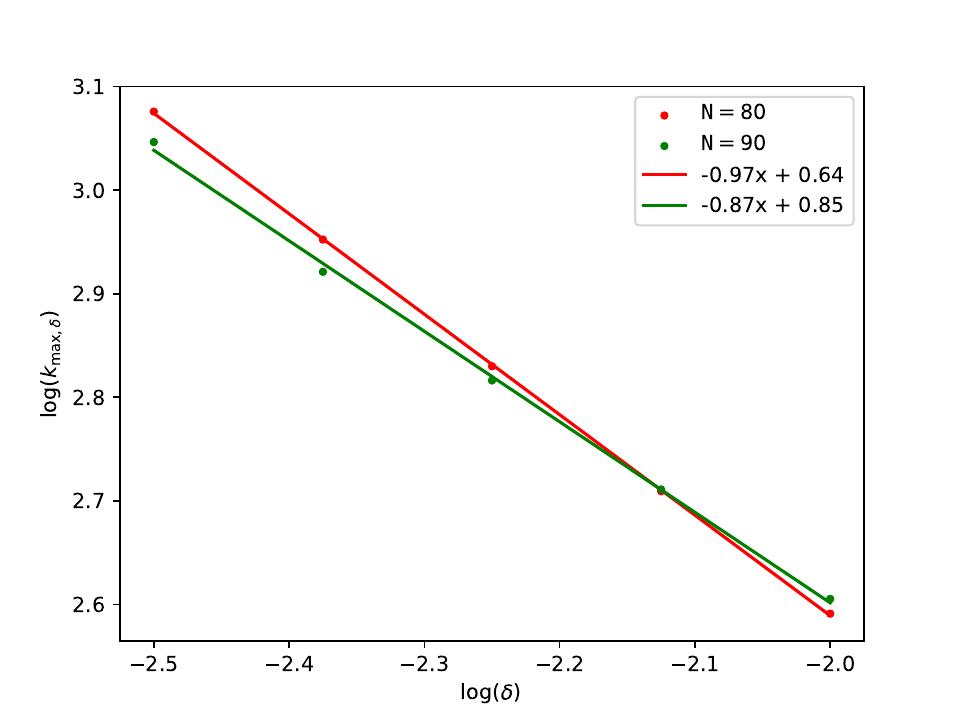}
\caption{$\log(k_{\max,\delta})$ versus $\log(\delta)$. For each $\delta$, the phase velocities are computed with $\mathtt N=80$ and $\mathtt N=90$ number of modes respectively, in order to find $k_{\max, \delta}$.}
\label{fig:k_delta} 
\end{minipage}
\end{figure}

\paragraph{\bf Strength of the Kelvin-Helmholtz instabilities, Part I} Let us now turn to the estimation of 
\begin{equation}
\label{eqn:def:cmax}
c_{\max,\delta} \coloneq  \argmax_{c\in \mathcal{C}_{\delta}(k), |k| \in [k_{\min,\delta},k_{\max,\delta}]\cap \frac{1}{L}\mathbb{Z}} \Im(c),
\end{equation}
which is the phase velocity with largest imaginary part, obtained numerically. The imaginary part of $c_{\max,\delta}$ is represented on Figure \ref{fig:dispersion:long}. This value is to be compared with the value 
$$\Im(c_{\max,\mathrm{BL}}) = \lim_{|k|\to\infty} \Im(c(k)) = \frac{\sqrt{\rho_- \rho_+}}{\rho_- + \rho_+} v,$$
where $c(k)$ refers here to the element of positive imaginary part in the dispersion relation of the bilayer equations, see Proposition \ref{prop:dispersion:ebc}. Indeed, from Figure \ref{fig:dispersion:superimposed}, we conjecture
\begin{equation}
\label{eqn:conj:khiI}
|\Im(c_{\max,\delta}) - \Im(c_{\max,\mathrm{BL}})| \underset{\delta \to 0}{\longrightarrow} 0.
\end{equation}
We compute the dispersion relation of \eqref{eqn:euler:isopycnal:Vb}, with $\mathtt N = 50$ vertical modes, computed with $\mathtt m_r = 200 \ 000$ points on the vertical axis, and $\mathtt K = 1000$, for several values of $\delta$ in $[10^{-3}, 5.10^{-2}]$. We show in Figure \ref{fig:cstar} the logarithm of the difference between $\Im(c_{\max,\delta})$ and $\Im(c_{\max,\mathrm{BL}})$. The difference seems to decrease polynomially as $\delta$ goes to $0$, however it decreases with a rather slow rate; this is consistent with the error at $\delta = 10^{-3}$, which is still non-negligible, namely
\begin{equation*}
\frac{|\Im(c_{\max,10^{-3}}) - \Im(c_{\max,\mathrm{BL}})|}{\Im(c_{\max,\mathrm{BL}})} \approx 10\%.
\end{equation*}

\begin{figure}
\begin{minipage}{0.49\textwidth}
\includegraphics[width =\textwidth]{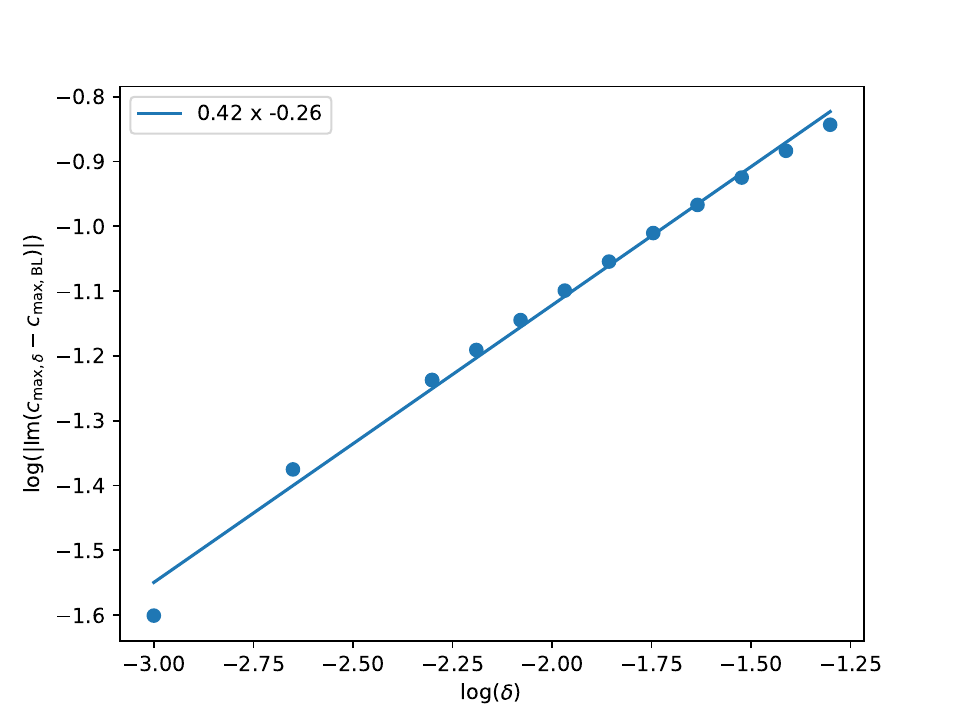}
\captionof{figure}{Value of \text{$\log(|\Im(c_{\max,\delta})-\Im(c_{\max,\mathrm{BL}})|)$} versus $\log(\delta)$, for $\delta$ between $5.10^{-2}$ and $10^{-3}$. The error $\frac{|\Im(c_{\max,10^{-3}}) - \Im(c_{\max,\mathrm{BL}})|}{\Im(c_{\max,\mathrm{BL}})}$ is approximately $10\%$.}
\label{fig:cstar}
\end{minipage}
\begin{minipage}{0.49\textwidth}
\includegraphics[width =\textwidth]{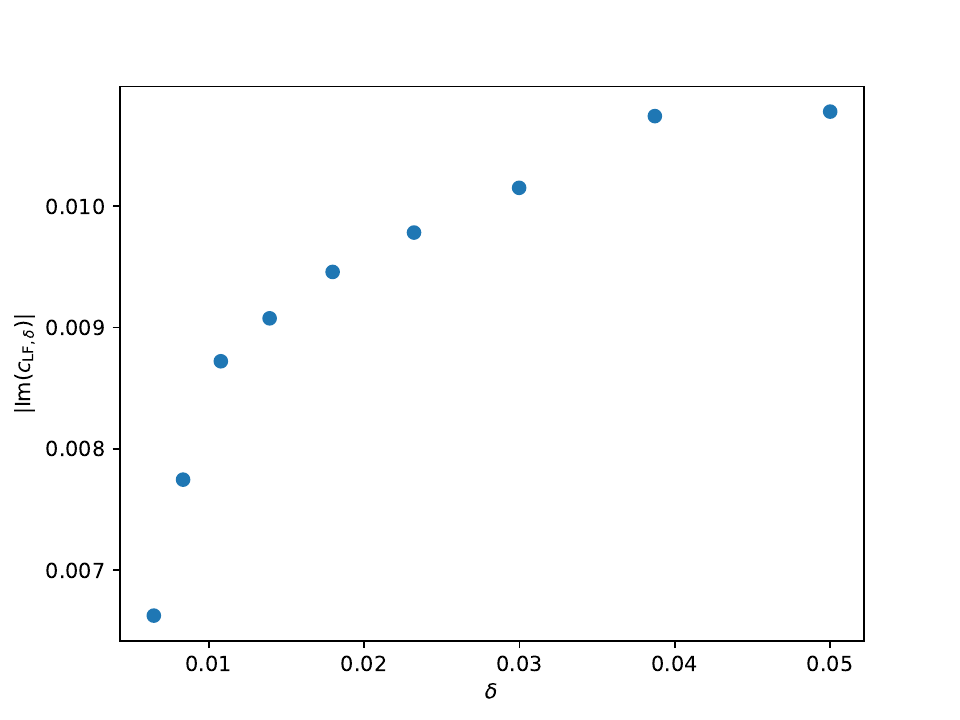}
\captionof{figure}{$|\Im(c_{\mathrm{LF},\delta})|$ versus $\delta$, for $\delta$ between $5.10^{-2}$ and $7.10^{-3}$, and $\mathtt N=50$ number of modes.}
\label{fig:clf}
\end{minipage}
\end{figure}

\paragraph{\bf Strength of the Kelvin-Helmholtz instabilities, Part II} We investigate the dependency of the quantity
\begin{equation}
\label{eqn:def:kc}
\omega^*_{\delta} \coloneq  \argmax_{c\in \mathcal{C}_{\delta}(k),|k| \in [k_{\min,\delta},k_{\max,\delta}]\cap\frac{1}{L}\N} \left(k \Im(c)\right).
\end{equation}
We are also interested in a value $k^*_{\delta}$ such that the maximum in \eqref{eqn:def:kc} is reached, namely such that
\begin{equation}
\label{eqn:def:kc:k}
\omega^*_{\delta} = k^*_{\delta} \argmax_{c \in \mathcal{C}(k^*_{\delta})} (\Im(c)),
\end{equation}
These quantities are related to the growth rate of growing modes, see Section \ref{subsection:sharp_strat:shear:csq}. We compute the dispersion relation of \eqref{eqn:euler:isopycnal:Vb}, with $\mathtt N = 50$ vertical modes, with $\mathtt m_r = 200 \ 000$ points on the vertical axis, and $\mathtt K = 1000$, for several values of $\delta$ in $[10^{-3}, 5.10^{-2}]$. We show the dependency of $\Im(\omega^*_{\delta})$ and $k^*_{\delta}$ with respect to $\delta$ in Figures \ref{fig:omegamax} and \ref{fig:kstar}. We conjecture:
\begin{equation}
\label{eqn:conj:khiII}
\sys{ \Im(\omega^*_{\delta}) &\approx \delta^{-1},\\
		k^*_{\delta} &\approx \delta^{-1}.}
\end{equation}
\begin{figure}
\begin{minipage}{0.49\textwidth}
\includegraphics[width=\textwidth]{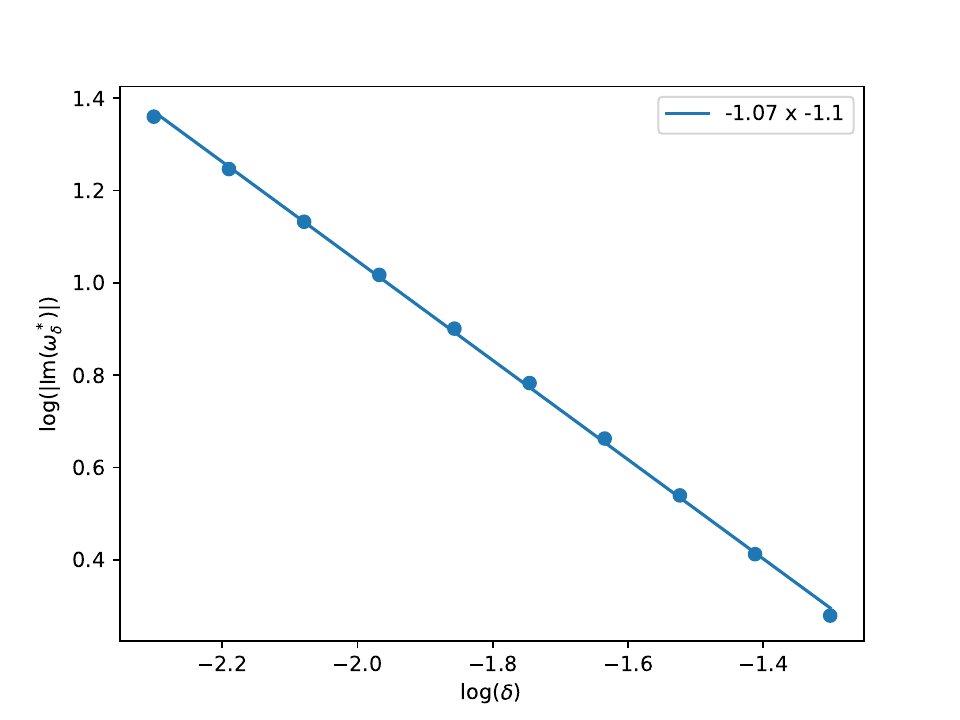}
\captionof{figure}{$\log(|\Im(\omega^*_{\delta})|)$ versus $\log(\delta)$, computed for $\delta$ between $5.10^{-2}$ and $5.10^{-3}$, with $\mathtt N = 50$ modes.}
\label{fig:omegamax}
\end{minipage}
\begin{minipage}{0.49\textwidth}
\includegraphics[width=\textwidth]{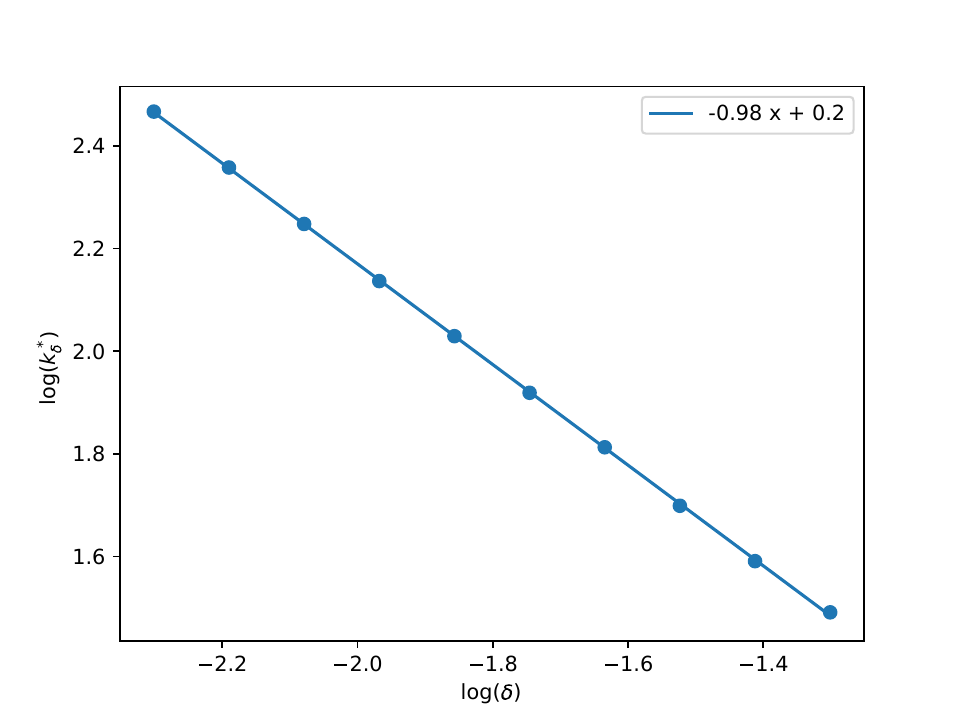}
\captionof{figure}{$\log(k^*_{\delta})$ versus $\log(\delta)$, computed for $\delta$ between $5.10^{-2}$ and $5.10^{-3}$, with $\mathtt N = 50$ modes.}
\label{fig:kstar}
\end{minipage}
\end{figure}

\paragraph{\bf Asymptotic stability of the low frequencies} Let us now turn to the behavior at low frequency. In the bilayer case, \nv{we define} 
\begin{equation*}
c_{\mathrm{LF},\mathrm{BL}} \coloneq  \argmax_{c(k), |k| < k_{\min,\mathrm{BL}}} |\Im(c(k)) |,
\end{equation*}
where $k_{\min,\mathrm{BL}}$ is the threshold of onset of the Kelvin-Helmholtz instabilities. This threshold can be computed from the dispersion relation of the bilayer Euler equations, see \eqref{eqn:dispersion:ebc:1}-\eqref{eqn:dispersion:ebc:2}, \nv{which also proves that} the low frequencies are stable in the bilayer case, i.e.
$$\Im(c_{\mathrm{LF},\mathrm{BL}}) =0.$$
In the continuously stratified case, we observe the same general behavior, with low-frequencies that are more stable than the high frequencies, and the threshold of the onset of the Kelvin-Helmholtz instabilities is roughly the same, see Figure \ref{fig:dispersion:superimposed}. Due to the fact that the low frequencies exhibit small but non-zero instabilities, this threshold is not as clearly defined as in the bilayer case. 
Finally we find that the low frequencies are stabilized in the limit $\delta \to 0$. More precisely, we define
\begin{equation*}
c_{\mathrm{LF},\delta} \coloneq  \argmax_{c \in \mathcal{C}(k), |k| \leq k_{\min,\mathrm{BL}}-1} |\Im(c(k))|,
\end{equation*}
which is the maximum instability in the low-frequencies, which we define as the frequencies below the threshold given by the analysis in the bilayer case. We compute the dispersion relation of \eqref{eqn:euler:isopycnal:Vb}, with $\mathtt N = 50$ vertical modes, computed with $\mathtt m_r = 200 \ 000$ points on the vertical axis, and $\mathtt K = 10$, for several values of $\delta$ in $[7.10^{-3}, 5.10^{-2}]$. We show the dependency of $c_{\mathrm{LF},\delta}$ with respect to $\delta$ in Figure \ref{fig:clf}. Note that we do not use a logarithmic scale, contrary to most of the plots in this study: this is because no polynomial, logarithmic or exponential pattern emerges from the data. Note also that the range of $\delta$ does not include as small values as in the rest of this study. The reason is that two factors may contribute to instabilities of the low frequencies: on the one hand, the low frequencies might exhibit instabilities. On the other hand, the continuous spectrum exhibits small instabilities, due to numerical inconsistencies. For lower values of $\delta$, the quantity $c_{\mathrm{LF},\delta}$ is of the same order as the (artificial) instabilities associated to the continuous spectrum. This range might be the reason why no general pattern emerges from the data. \\

\paragraph{\bf Conjecture on the form of the dispersion relation of \eqref{eqn:euler:isopycnal:Vb}} Our findings, described in the present section, suggest the following conjecture, which is an analogue of Proposition \ref{prop:dispersion:ebc}, for the stratified Euler equations.
\begin{conjecture}
\label{conj}
Let $\delta > 0$, $\rhob_{\delta}, \Vb_{\delta}$ as in \eqref{eqn:def:sharp_strat} and \eqref{eqn:def:sharp:shear}. There exist $k_{\min,\delta}, k_{\max,\delta}$, such that the following holds.
\begin{itemize}
\item \underline{Asymptotic stability of the low frequencies.} The low frequencies are asymptotically stable, i.e.
$$ \max_{c \in \mathcal{C}_{\delta}(k), |k| \leq k_{\min,\delta}} |\Im(c(k))| \underset{\delta \to 0}{\longrightarrow} 0,$$
and, with $k_{\min,\mathrm{BL}}$ defined in Proposition \ref{prop:dispersion:ebc}:
$$k_{\min,\delta} \underset{\delta \to 0}{\longrightarrow} k_{\min,\mathrm{BL}}.$$

\item \underline{Presence of the Kelvin-Helmholtz instabilities.} All the frequencies $k \in [k_{\min,\delta}, k_{\max,\delta}]$ are unstable. Moreover,
$$\max_{c\in \mathcal{C}_{\delta}(k),|k| \in [k_{\min,\delta},k_{\max,\delta}]\cap\frac{1}{L}\N} \left(k \Im(c)\right) \approx \delta^{-1},$$
and this maximum is reached for some frequency $k^*_{\delta}$ satisfying
$$k^*_{\delta} \approx \delta^{-1}.$$

\item \underline{Stability of the frequencies beyond $k_{\max,\delta}$.} All the frequencies $k$ with $|k| > k_{\max,\delta}$ are stable, and moreover
$$k_{\max,\delta} \approx \delta^{-1}.$$
\end{itemize}
\end{conjecture}
\nouveau{
\begin{remark}
The scaling $\delta^{-1}$ appearing in Conjecture \ref{conj} can be understood from a (formal) scaling argument in the Taylor-Goldstein equation \eqref{eqn:taylor_goldstein}. Indeed, our density and velocity profiles $\rhob_{\delta},\Vb_{\delta}$ are defined from (regular) profiles $\rhob,\Vb$ through
$$\rhob_{\delta}(r) = \rhob\left(\frac{r-r_*}{\delta}\right), \qquad \Vb_{\delta}(r) = \Vb\left(\frac{r-r_*}{\delta}\right),$$
which imply
$$ N^2_{\delta} = \frac{1}{\delta} N^2, \qquad N^2 = \frac{-g\rhob'}{\rhob}.$$
We apply the change of variable $s \coloneq \frac{r-r_*}{\delta}$ to \eqref{eqn:taylor_goldstein}, to write a rescaled Taylor-Goldstein equation:
\begin{equation}
\label{eqn:taylor_goldstein:rescaled:delta}
(\Vb - c(k))^2 \left( (\rhob w')' - (\delta k)^2 \rhob w \right) - (\Vb - c(k)) (\rhob \Vb')'w + \delta \rhob N^2 w =0.
\end{equation}
Now the only dependency in $\delta$ is through the product $\delta k$, and the factor $\delta$ in the last term in \eqref{eqn:taylor_goldstein:rescaled:delta}. At main order in $\delta$, we can therefore expect that this last term is negligible. Dropping this last term and defining $\tilde{k} \coloneq \delta k$, we get
\begin{equation}
\label{eqn:taylor_goldstein:rescaled:delta:bis}
(\Vb - c)^2 \left( (\rhob w')' - \tilde k^2 \rhob w \right) - (\Vb - c) (\rhob \Vb')'w \approx 0,
\end{equation}
which is an equation that does not depend on $\delta$. Let $w^{*}$, $ c^*\in \mathbb C, \tilde{k}^* \in \R_+$ be a solution of \eqref{eqn:taylor_goldstein:rescaled:delta:bis}, with the maximum value of $\Im(c)$ among solutions of \eqref{eqn:taylor_goldstein:rescaled:delta:bis} - although both the existence and uniqueness of such a solution are far from being trivial, and require a proof. Note in particular that $c^*$ and $\tilde k^*$ are independent of $\delta$, since \eqref{eqn:taylor_goldstein:rescaled:delta:bis} is independent of $\delta$. This leads to the following asymptotic for the maximum instability in the Taylor-Goldstein equation in the original variable $r$:
$$ k_{\delta}^* = \frac{\tilde{k}^*}{\delta}, \qquad \Im (c^*(k_{\delta}^*)) = \Im(c^*), \qquad \Im(\omega_{\delta}^*) = \frac{\Im(c^*) \tilde{k}^*}{\delta},$$
where we used the definition \eqref{eqn:def:kc} of $\omega_{\delta}^*$. Similarly, if there exists $\tilde{k}_{\max}$ such that any $(w,c,\tilde{k})$ solution to \eqref{eqn:taylor_goldstein:rescaled:delta} with $\tilde k \geq \tilde k_{\max}$ satisfies $\Im(c) = 0$, then
$$ k_{\max,\delta} = \frac{\tilde{k}_{\max}}{\delta},$$
with the notation of Conjecture \ref{conj}. This is exactly the scaling conjectured in Conjecture \ref{conj}, for the second and third points. \\

This reasoning requires two additional arguments to be made rigorous. First, the rescaling $s = \frac{r-r_*}{\delta}$ implies that the rescaled Taylor-Goldstein equation \eqref{eqn:taylor_goldstein:rescaled:delta} is defined on an interval of length of order $\frac{1}{\delta}$, so that it is not independent of $\delta$. In \cite{KumarOzanski2025}, the authors study the Taylor-Goldstein equation \eqref{eqn:taylor_goldstein} with $\rhob = 1$ (constant density), and a shear profile similar to $\Vb_{\delta}$. They manage to make the scaling argument rigorous, by using a more involved change of variables. Second, and most importantly, discarding the last term in \eqref{eqn:taylor_goldstein:rescaled:delta} requires a careful justification. 
\end{remark}}
\subsection{Consequences for the justification of bilayer models}
\label{subsection:sharp_strat:shear:csq}
Let us now point out the consequence of our findings on the full justification of bilayer models. \\
In many applications, bilayer models are used as reduced models for a continuously stratified fluid with a thin pycnocline. This is particularly the case for the study of geophysical flows, for which the bilayer shallow-water equations are commonly used; we comment on this later in this subsection. The derivation of these models relies on formal arguments. However, it is possible in some cases to  provide a rigorous comparison of the solutions of both the continuous and the bilayer models, stemming from close initial data. This is done in Theorem \ref{prop:cvce} in the case $\Vb = 0$, for the comparison with the bilayer Euler equations, and it relies on energy estimates; this is what we call {\it full justification}. In the case where $\Vb$ is of the form given in \eqref{eqn:def:sharp:shear}, our findings indicate that such a result does not seem to be true in general, due to the presence of the Kelvin-Helmholtz instabilities. Let us take two particular bilayer models.

\subsubsection{The bilayer Euler equations}
\label{subsection:csq:ebc}
Let us compare the solutions of the linear stratified Euler equations with a shear flow \eqref{eqn:euler:isopycnal:Vb}. We consider the density and shear profiles given in \eqref{eqn:def:sharp_strat} and \eqref{eqn:def:sharp:shear}. Let $\omega^*_{\delta}$ and $k^*_{\delta}$ be defined in \eqref{eqn:def:kc} and \eqref{eqn:def:kc:k}. Let $(V_n,\eta_n)_{n\in \N^*}$ be the eigenvector associated to $\omega^*_{\delta}$, for which we find  a numerical approximation with the numerical scheme described in Section \ref{subsection:num_scheme_description}, and let
$$V = \sum_{n\in \N^*} V_n g_n, \qquad w = \sum_{n\in \N^*} c_n ik^*_{\delta} V_n f_n, \qquad \eta = \sum_{n\in \N^*} \eta_n f_n,$$
and
$$f^* \coloneq  (V,w,\eta).$$
We thus find that, for any constant $\mathcal{U}^*$,
\begin{equation}
\label{eqn:csq:temp:0}\mathcal{U}^*_{\delta}(t,x,r) \coloneq \mathcal{U}^* e^{-it\omega^*_{\delta} + ik^*_{\delta}x} f^*(r)
\end{equation}
is a particular solution of \eqref{eqn:euler:isopycnal:Vb}, whose $L^2(S)$ norm grows like $e^{t\Im(\omega^*_{\delta})}$. We conjecture $\Im(\omega^*_{\delta}) \approx \delta^{-1}$ (see Figure~\ref{fig:omegamax} and Conjecture \ref{conj}), so that
\begin{equation}
\label{eqn:csq:temp}
\Vert \mathcal{U}^*_{\delta}(t,\cdot,\cdot)\Vert_{L^2(S)} \approx |\mathcal{U}^*_0| e^{t \delta^{-1}}.
\end{equation}
Such a pathological behavior prevents to define solutions of \eqref{eqn:euler:isopycnal:Vb} on a time interval $[0,T]$ independent of $\delta$ for generic initial data in $L^2$. More generally, the same argument holds for initial data in any Sobolev space $H^s$: from Figure \ref{fig:kstar} and Conjecture \ref{conj}, we conjecture $k^*_{\delta} \approx \delta^{-1}$. Then, let us assume that $\mathcal U^*_{\delta|t=0}$ is bounded in $H^s$ uniformly in $\delta$. Because of \eqref{eqn:csq:temp:0} and $$k^*_{\delta} \approx \delta^{-1},$$ this implies
$$ \mathcal U^* \leq C \delta^s,$$
for some constant $C$. Then:
\begin{equation}
\label{eqn:exp_growth_delta}
|\mathcal{U}^*_{\delta}(t,\cdot,\cdot)|_{L^2} \leq C \delta^s e^{t\delta^{-1}}.
\end{equation}
Thus, the assumption $\mathcal U^*_{\delta|t=0}$ uniformly bounded in $H^s$ is not enough to ensure that $\mathcal{U}^*_{\delta}$ remains bounded uniformly in $L^2$ on a time interval independent of $\delta$. This is consistent with ill-posedness results on the bilayer Euler equations (see \cite{IguchiTanakaTani1997}, \cite{LebeauKamotski2005} for a framework close to ours, as well as \cite[Chapter 1, Section 3.5]{Duchene2022a} and references within).\\
\begin{remark}
\label{rk:analytic}
Analytic functions are associated with exponentially decreasing Fourier coefficients. In our context, analytic initial data imply $|\mathcal{U}^*_0| \sim e^{-k^*_{\delta}}$, so that \eqref{eqn:csq:temp} becomes
$$\Vert \mathcal{U}^*(t,\cdot,\cdot)\Vert_{L^2(S)} \sim e^{t \delta^{-1} - \delta^{-1}},$$%
which could allow for bounded solution on a time interval uniform in $\delta$. Therefore, our findings do not prevent a full justification of the bilayer Euler equations from the stratified Euler equations in the sharp stratification limit, in analytical spaces. This is consistent with well-posedness results in analytic spaces for the vortex sheet problem (see \cite{SulemSulemBardos1981}, \cite{CaflischOrellana1986}), which is closely related to the bilayer Euler equations, and in particular suffers from the Kelvin-Helmholtz instability as well.  Such a result is, to the best of our knowledge, absent from the literature. 
\end{remark}
\subsubsection{The bilayer shallow-water equations}
\label{subsection:csq:blsw}
We now introduce the shallow-water parameter $\mu \coloneq  H^2/L^2$, where $H$ is the depth of the strip $S$ and $L$ is the length of the torus. In most applications, it is a small parameter, so that we are interested in the limit $\mu \to 0$. Formally, the equations one expects to recover from the stratified Euler equations in the double limit $\mu,\delta \to 0$ are the bilayer shallow-water equations. These equations are ubiquitous in oceanography; one of their advantages compared to the bilayer Euler equations is that they do not suffer from the Kelvin-Helmholtz instabilities. They are even well-posed at the non-linear level, see \cite[Chapter B, Section 6]{Duchene2022a} \nouveau{and \cite{GuyenneLannesSaut2010}}. We now show that our findings also provide numerical evidence that a full justification of the bilayer shallow-water equations from the stratified Euler equations cannot hold. \\

One can rescale the variables and unknowns to obtain the non-dimensionalized stratified Euler equations, similar to \eqref{eqn:euler:isopycnal:Vb} but that makes appear $\mu$, and posed in the non-dimensionalized strip $\mathbb{T}_1 \times [-1,0]$. The non-dimensionalized  Taylor-Goldstein equation then reads
\begin{equation}
\label{eqn:taylor_goldstein:rescaled}
(\Vb - c)^2 \left( (\rhob w')' - \mu k^2 \rhob w \right) - (\Vb - c) (\rhob \Vb')'w + \rhob N^2 w =0,
\end{equation}
with Dirichlet boundary conditions. We thus see that the non-dimensionalization process consists in replacing the Fourier variable $k$ by $\sqrt{\mu}k$. Hence, Figure \ref{fig:dispersion:long} still describes the phase velocities in the non-dimensional case, and we can specify the dependency in the shallow-water parameter
$$k_{\min,\delta,\mu} = \frac{1}{\sqrt{\mu}} k_{\min,\delta}, \qquad k_{\max,\delta,\mu} = \frac{1}{\sqrt{\mu}} k_{\max,\delta}, \qquad c_{\max,\delta,\mu} = c_{\max,\delta}.$$
We also find, with the same rescaling argument,
\begin{equation}
\label{eqn:def:kc:mu}
\omega^*_{\delta,\mu} \coloneq  \argmax_{c\in \mathcal{C}_{\delta,\mu}(k), |k| \in [k_{\min,\delta,\mu},k_{\max,\delta,\mu}]\cap\N} \left(k \Im(c)\right)= \frac{1}{\sqrt{\mu}} \omega^*_{\delta} \qquad \text{and} \qquad k^*_{\mu,\delta}  = \frac{1}{\sqrt{\mu}} k^*_{\delta},
\end{equation}
where $\mathcal{C}_{\delta,\mu}(k)$ is defined as in \eqref{eqn:def:mathcalC} but with $\mu$ possibly different from $1$. Now the same discussion as in Section \ref{subsection:csq:ebc} holds: a solution $\mathcal{U}^*_{\delta,\mu}(t,x,r)$ of the form
$$ \mathcal{U}^*_{\delta,\mu}(t,x,r) =\mathcal{U}^* e^{-it\omega^*_{\delta,\mu} + ik^*_{\delta,\mu}x} f^*(r),$$
which is in $H^{s}$ at $t=0$ exhibits exponential growth, even in $L^2$ norm:
\begin{equation}
\label{eqn:exp_growth_delta_mu}
\Vert \mathcal{U}^*_{\delta,\mu}(t,\cdot,\cdot)\Vert_{L^2(S)} \approx \mu^{\frac{s}{2}}\delta^{s} e^{t \mu^{-\frac12}\delta^{-1}}.
\end{equation}
In other words, because of the Kelvin-Helmholtz instabilities, it does not seem to be possible to define solutions of (the non-dimensional form of) \eqref{eqn:euler:isopycnal:Vb}, on a time interval $[0,T]$ independent of $\delta$ and $\mu$, for generic initial data in a Sobolev space. In fact, the smaller the value of $\mu$, the faster the exponential growth is. This is in sharp contrast with the expected reduced model at $\delta = \mu = 0$: indeed, the bilayer shallow-water equations are well-posed in Sobolev spaces, even at the non-linear level (under a hyperbolicity condition, see \cite[Chapter B, Section 6]{Duchene2022a} and references \nv{therein}). This can be understood through the following formal argument: our findings show that the low frequencies, i.e. $|k|\leq k_{\min,\delta,\mu}$ are stable, i.e. $\Im(c) \approx 0 $ - this last approximation seems to hold only asymptotically in $\delta$, see Figure \ref{fig:clf} and Conjecture \ref{conj}. However, it appears that $k_{\min,\delta,\mu} \approx \frac{1}{\sqrt{\mu}}$, so that when formally setting $\mu = 0$, all frequencies are stable. 
\FloatBarrier
\section{Conclusion and perspectives}
\label{section:conclusion}
In this study we investigate the sharp stratification limit, in the framework of the linear stratified Euler equations \eqref{eqn:euler:isopycnal:Vb} with a particular background density profile $\rhob_{\delta}$ that converges to a piecewise constant profile when $\delta$ goes to $0$; here, $\delta$ is the size of the pycnocline, the transition layer between the upper and lower layers. We also consider a background velocity profile $\Vb$ that is a shear flow, i.e. depending only on the vertical variable. We compare this system with the bilayer Euler equations \eqref{eqn:euler:bilayer}.\\

In the case $\Vb = 0$, the linear bilayer Euler equations do not exhibit instabilities. More precisely we derive energy estimates for these equations and for the stratified Euler equations \eqref{eqn:euler:isopycnal:no_shear}, to prove in Proposition \ref{prop:cvce} that solutions to these systems starting from close initial data, remain close on some time interval $[0,T]$; we measure their difference with respect to the parameter $\delta$. For this we find that isopycnal coordinates (see Section \ref{subsection:isopyc}) are particularly useful. \\

In the general non-linear Euler equations (in the continuously stratified or bilayer settings), there is a non-zero horizontal velocity. We introduce the shear velocity $\Vb_{\delta}$ in \eqref{eqn:def:sharp:shear} to study its influence on the behavior of solutions. It is well-known (see \cite[Chapter 1, Section 3]{Duchene2022a}) that the bilayer Euler equations linearized around a non-zero shear velocity suffer from the Kelvin-Helmholtz instabilities (see Figure \ref{fig:dispersion:ebc}), which imply arbitrarily fast growing modes. This implies ill-posedness of the bilayer Euler equations in Sobolev spaces. It should be noted that the non-linear stratified Euler equations are proved to be well-posed in Sobolev spaces (see \cite{Desjardins2019}, \cite{Fradin2024}), however on a time interval that may shrink to the singleton $\{0\}$ as $\delta$ goes to $0$.\\

In this unstable case, we provide some numerical results. To this end, we develop a numerical strategy in Section \ref{section:num_scheme}, based on a modal decomposition, adapted to the stratification profile at hand, and in particular allows to treat the case where the Brunt-Väisälä frequency depends on the vertical variable. In this case, we provide some numerical illustrations (see Figures \ref{fig:waveguide:snapshots} and \ref{fig:waveguide:reflection:snapshots}) that show the effect of a pycnocline of non-zero thickness on the propagation of internal waves.\\

The numerical strategy developed in Section \ref{section:num_scheme} allows us to compute the dispersion relation (more precisely, an approximation of the phase velocities) of the stratified Euler equations \eqref{eqn:euler:isopycnal:Vb}, shown in Figures \ref{fig:dispersion:strat}, \ref{fig:dispersion:long}, \ref{fig:dispersion:superimposed}. We find a good agreement with the dispersion relation in the bilayer case: namely, stability of the low frequencies, a threshold $k_{\min,\delta}$ for which higher frequencies are unstable. The main difference is the presence of a threshold $k_{\max,\delta}$ after which higher frequencies are stable again, see Figure \ref{fig:dispersion:long}. We also provide some quantitative estimations of these features in Section \ref{subsubsection:sharp_strat:shear:quantitative}. This allows us to formulate Conjecture \ref{conj} on the general features of the phase velocities of the stratified Euler equations, in the sharp stratification limit. Overall, we find that the bilayer equations, which we see as a reduced model that can be formally derived from the stratified Euler equations, reproduce the Kelvin-Helmholtz instabilities with good accuracy.\\

Our findings also indicate that there exist exponential growing modes, with rate $\Im(\omega_{\max,\delta})\approx \delta^{-1}$ (see Figure \ref{fig:omegamax}). We explain in Section \ref{subsection:sharp_strat:shear:csq} that, according to these findings, it is not possible to define solutions in Sobolev spaces on a time interval uniform in $\delta$, which is the crucial starting point to prove the convergence of solutions of the stratified Euler equations \eqref{eqn:euler:isopycnal:Vb} towards a bilayer model. This is consistent when considering the limit $\delta \to 0$ alone, as in this case the reduced model (the bilayer Euler equations) \nv{is} ill-posed. We find that this conclusion still holds when also allowing smallness of the shallow-water parameter $\mu$. This is a striking fact: when considering the (formal) double limit $\delta, \mu \to 0$, the reduced model is the bilayer shallow-water equations, which are well-posed in Sobolev spaces. Yet, it does not appear to be possible to justify the shallow-water equations from the stratified Euler equations, in the Sobolev framework.\\

The main perspective of the present article is to prove Conjecture \ref{conj}, i.e. rigorously describe the dispersion relation shown in Figure \ref{fig:dispersion:long}, and more precisely the various features (asymptotic stability of the low frequencies $|k| \leq k_{\min,\delta}$, and high enough frequencies $|k| \geq k_{\max,\delta}$, Kelvin-Helmholtz instabilities for $k_{\min,\delta} \leq |k| \leq k_{\max,\delta}$), as well as asymptotics in the regime $\delta \to 0$ for some key elements (for instance, $k_{\min,\delta}, k_{\max,\delta}$ and $\omega_{\max,\delta}$, defined in Figure \ref{fig:dispersion:long} and \eqref{eqn:def:kc}). A first step towards this line of research is to study the Sturm-Liouville problem \eqref{eqn:SL} with $\rhob$ given by \eqref{eqn:def:sharp_strat}. This would for instance provide a justification of the asymptotic of the mode speeds $c_n$ displayed in Figure \ref{fig:c}. It would also shed light onto the behavior of the various coupling operators in \eqref{eqn:modes:couple} and \eqref{eqn:modes:couple:shear}, in the sharp stratification limit.\\

Another perspective from the numerical point of view is to extend the numerical strategy described in Section \ref{section:num_scheme} to the non-linear setting. In order to keep an efficient modal approach as the one we performed in the present study, a fast Sturm-Liouville transform would be required, in order to adapt the standard pseudo-spectral methods to our framework.


\bibliographystyle{amsalpha}
\bibliography{../../../Seafile/refs/refs} 
\appendix

\section{The Taylor Goldstein equation and the modal decomposition in the presence of a shear flow}
\label{section:apdx:derivation}
In this appendix we show how to derive the Taylor-Goldstein equation \eqref{eqn:taylor_goldstein} and the system \eqref{eqn:modes:couple:shear}, for the sake of completeness.
\subsection{Derivation of \eqref{eqn:taylor_goldstein}}
We now show how to derive the Taylor-Goldstein equation \eqref{eqn:taylor_goldstein} (and therefore, \eqref{eqn:taylor_goldstein:reduit}, which is exactly \eqref{eqn:taylor_goldstein} with $\Vb = 0$). Let $V,w,P,\eta$ be a smooth solution of \eqref{eqn:euler:isopycnal:Vb}. Applying $\rhob \partial_x$ to the first equation in \eqref{eqn:euler:isopycnal:Vb} and using the incompressibility condition, we get
\begin{equation}
\label{eqn:tg:temp:1}
- \rhob (\partial_t + \Vb \partial_x) \partial_r w + \rhob \Vb' \partial_x w + \partial_x^2 P = 0.
\end{equation}%
Applying $\partial_r$ to \eqref{eqn:tg:temp:1} yields
\begin{equation}
\label{eqn:tg:temp:2}
- (\partial_t + \Vb \partial_x) \partial_r (\rhob \partial_r w) - \Vb' \rhob \partial_x \partial_r w + (\rhob \Vb')' \partial_x w + \rhob \Vb'\partial_r  \partial_x w+ \partial_r \partial_x^2 P = 0;
\end{equation}
notice that the second and fourth terms in \eqref{eqn:tg:temp:2} cancel out each other. Now applying $\rhob \partial_x^2$ to the second equation in \eqref{eqn:euler:isopycnal:Vb} and substracting the result from \eqref{eqn:tg:temp:2}, we get
\begin{equation}
\label{eqn:tg:temp:3}
- (\partial_t + \Vb \partial_x) \left(\partial_r (\rhob \partial_r w) + \rhob \partial_x^2 w\right) + (\rhob \Vb')' \partial_x w - \rhob N^2 \partial_x^2 \eta = 0.
\end{equation}
We now apply $-\frac{1}{\rhob}(\partial_t + \Vb \partial_x)$ to \eqref{eqn:tg:temp:3} and use the third equation in \eqref{eqn:euler:isopycnal:Vb} to get the Taylor-Goldstein equation \eqref{eqn:taylor_goldstein}.

\subsection{Derivation of \eqref{eqn:modes:couple:shear}}
We now show how to obtain \eqref{eqn:modes:couple:shear}. Let $V,w,P,\eta$ be a smooth solution of \eqref{eqn:euler:isopycnal:Vb}. We decompose $V$ along $(g_n)_{n\in \N}$, which is a basis of $L^2_{\rhob}$, see Lemma \ref{lemma:SL}. As noted in Remark \ref{rk:coef_zero:shear}, $V_0$ is constant in time and space, so that we focus on the evolution of $V_n$ for $n \geq 1$. Similarly, we decompose $P$ along $(\rhob g_n)_{n\in\N}$, which is a basis of $L^2_{\rhob^{-1}}$. We decompose $w$ and $\eta$ along $(f_n)_{n\in \N^*}$, which is a basis of $L^2_{\rhob N^2}$. We write $V_n,P_n,w_n,\eta_n$ for their coefficients, respectively. Plugging the decompositions of $V,w$ and $P$ into the first equation of \eqref{eqn:euler:isopycnal:Vb} and taking the $L^2_{\rhob}$ scalar product with $g_n$, we get, for $n \in \N^*$:
\begin{equation}
\label{eqn:modes:temp:1}
\partial_t V_n + \sum_{m \in \N^*}\partial_x V_m \int_{-H}^0 \Vb g_m g_n \rhob  + \sum_{m \in \N^*} w_m \int_{-H}^0 \Vb' f_m g_n \rhob  + \partial_x P_n=0.
\end{equation}
We integrate by parts in the last integral in \eqref{eqn:modes:temp:1} and use \eqref{eqn:SL} and \eqref{eqn:def:gn} to get
\begin{equation}
\label{eqn:modes:temp:2}
\partial_t V_n + \sum_{m \in \N^*}\partial_x V_m \int_{-H}^0 \Vb g_m g_n \rhob  - \sum_{m \in \N^*} w_m \left(\frac{1}{c_m}\int_{-H}^0 \Vb g_m g_n \rhob - \frac{1}{c_n}\int_{-H}^0 \Vb f_m f_n \rhob N^2 \right) + \partial_x P_n=0.
\end{equation}
We now plug the modal decomposition of $V$ and $w$ into the incompressibility condition in \eqref{eqn:euler:isopycnal:Vb} to write, for $n \in \N^*$:
\begin{equation}
\label{eqn:modes:temp:3}
w_n = - c_n \partial_x V_n,
\end{equation}
as $c_n f_n' =g_n$. Plugging \eqref{eqn:modes:temp:3} into \eqref{eqn:modes:temp:2} and using \eqref{eqn:def:A}, we get
\begin{equation}
\label{eqn:modes:temp:4}
\partial_t V_n + \left((2A^1 + A^2)\partial_x V\right)_n + \partial_x P_n=0.
\end{equation}
Now notice that
\begin{equation}
\label{eqn:modes:temp:5}
\partial_r P = \partial_r \left(\sum_{n \in \N} P_n \rhob g_n\right) = - \sum_{n \in \N^*} P_n \frac{\rhob N^2}{c_n} f_n. 
\end{equation}
Note that we used that $\rhob g_0$ is constant. Now plugging the modal decomposition into the second equation in \eqref{eqn:euler:isopycnal:Vb}, and taking the $L^2_{\rhob}$ scalar product with $f_n$, we get
\begin{equation}
\label{eqn:modes:temp:6}
\sum_{m\in \N^*} \partial_t w_m \int_{-H}^0 f_m f_n \rhob  + \sum_{m\in \N^*} \partial_x w_m \int_{-H}^0 \Vb f_n f_m \rhob - \frac{1}{c_n} P_n + \eta_n = 0.
\end{equation}
Applying $c_n \partial_x$ to \eqref{eqn:modes:temp:6}, using \eqref{eqn:modes:temp:3} and adding the result to \eqref{eqn:modes:temp:4}, we get the first equation in \eqref{eqn:modes:couple:shear}.\\

Now plugging the modal decomposition in the third equation in  \eqref{eqn:euler:isopycnal:Vb} and using \eqref{eqn:modes:temp:3}, we get the second equation in \eqref{eqn:modes:couple:shear}.

\section{The sharp stratification limit with $\Vb = 0$}
\label{section:apdx:preuve}
In this appendix we prove Propositions \ref{prop:nrj:bilayer} and \ref{prop:cvce}.
To this end, we start in Section \ref{subsection:euler:strat:sharp} by showing how the linear stratified Euler equations \eqref{eqn:euler:isopycnal:no_shear} are globally well-posed in Proposition \ref{prop:nrj}. Then in Section \ref{subsection:euler:bilayer} we show how the linear bilayer Euler equations \eqref{eqn:euler:bilayer} are also globally well-posed. Then, in Section \ref{subsection:strat_to_bilayer} we perform estimates on the difference between the solutions of \eqref{eqn:euler:isopycnal:no_shear} and \eqref{eqn:euler:bilayer}, stemming from the same initial data; that is to say, we prove Proposition \ref{prop:cvce}. \\

\subsection{The stratified Euler equations with a sharp stratification}
\label{subsection:euler:strat:sharp}
In this section, we briefly show that the system \eqref{eqn:euler:isopycnal:no_shear} with a sharp stratification as in \eqref{eqn:def:sharp_strat} is well-posed, at the linear level and without a shear flow. 
\begin{proposition}
\label{prop:nrj}
Let $s \geq 1$, and $(V_{\ini},w_{\ini},\eta_{\ini}) \in H^{s,0}(S) \times H^{s,1}(S)^2$ satisfying the incompressibility condition in \eqref{eqn:euler:isopycnal:no_shear} and the boundary conditions \eqref{eqn:euler:iso:bc}. There exists a unique solution $(V,w,\eta)\in C^0([0,+\infty),H^{s,0}(S) \times H^{s,1}(S)^2)$ to \eqref{eqn:euler:isopycnal:no_shear}, and it satisfies, for $t > 0$:
\begin{equation}
\label{eqn:nrj:delta}
\Vert \sqrt{\rhob_{\delta}} V(t,\cdot)\Vert_{H^{s,0}}^2 + \Vert \sqrt{\rhob_{\delta}} w(t,\cdot)\Vert_{H^{s,1}}^2 + \Vert \sqrt{-g\rhob'_{\delta}} \eta(t,\cdot) \Vert_{H^{s,0}}^2 \leq \Vert \sqrt{\rhob_{\delta}} V_{\ini}\Vert_{H^{s,0}}^2 + \Vert \sqrt{\rhob_{\delta}} w_{\ini}\Vert_{H^{s,1}}^2 + \Vert \sqrt{-g\rhob'_{\delta}} \eta_{\ini} \Vert_{H^{s,0}}^2.
\end{equation}
Moreover, it holds
\begin{equation}
\label{eqn:nrj:delta:dreta}
\Vert \partial_r \eta(t,\cdot) \Vert_{H^{s-1,0}} \leq \Vert \partial_r \eta_{\ini} \Vert_{H^{s-1,0}}+ t \Vert V_{\ini} \Vert_{H^{s,0}}.
\end{equation}
\end{proposition}
Note that the statement of Proposition \ref{prop:nrj} holds for any smooth profile $\rhob$ satisfying $c_* \leq \rhob \leq c^*$ and $-\rhob' \geq c_*$ for some $c_*,c^* > 0$, and is not restricted to the profile $\rhob_{\delta}$ defined in \eqref{eqn:def:sharp_strat}. We keep however this notation, to insist on the fact that Proposition \ref{prop:nrj} provides bounds on the solution of \eqref{eqn:euler:isopycnal:no_shear} that are uniform in $\delta$. 
\begin{proof}
Let us first mention that the pressure is defined from an elliptic problem: taking the divergence of the equations on $V$ and $w$ in \eqref{eqn:euler:isopycnal:no_shear} yields
\begin{equation}
\label{eqn:elliptic:1}
 \partial_x \frac{1}{\rhob_{\delta}} \partial_x P + \partial_r \frac{1}{\rhob_{\delta}} \partial_r P = \partial_r \left(\frac{g \rhob'_{\delta}}{\rhob_{\delta}} \eta\right),
\end{equation}
completed with Neumann conditions at the flat bottom and surface, obtained by taking the trace of the equation on $w$ in \eqref{eqn:euler:isopycnal:no_shear} at the bottom and surface. Then classical elliptic theory (see for instance \cite[I, Chapter 5, Section 7]{Taylor2011}) yields the existence and uniqueness of the pressure in a suitable Sobolev space.\\
Using the standard energy estimates for \eqref{eqn:euler:isopycnal:no_shear} (see for instance \cite{Desjardins2019,Fradin2024}), we directly get
\begin{equation}
\label{eqn:nrj:delta:L2}
\frac{d}{dt}\left( \Vert \sqrt{\rhob_{\delta}} V\Vert_{L^2}^2 + \Vert \sqrt{\rhob_{\delta}} w\Vert_{L^2}^2 + \Vert \sqrt{-g\rhob'_{\delta}} \eta \Vert_{L^2}^2 \right) = 0.
\end{equation}
Applying $\partial_x^{\alpha}$ to the system \eqref{eqn:euler:isopycnal:no_shear} for $\alpha \in \N$, we get that $(\dot{V},\dot{w},\dot{\eta},\dot{P})\coloneq  (\partial_x^{\alpha} V,\partial_x^{\alpha} w,\partial_x^{\alpha} \eta,\partial_x^{\alpha} P)$ satisfies \eqref{eqn:euler:isopycnal:no_shear}. Applying the previous $L^2$ estimate to $(\dot{V},\dot{w},\dot{\eta},\dot{P})$ and summing over $|\alpha| \leq s$ yields
\begin{equation}
\label{eqn:nrj:delta:temp}
\frac{d}{dt}\left( \Vert \sqrt{\rhob_{\delta}} V\Vert_{H^{s,0}}^2 + \Vert \sqrt{\rhob_{\delta}} w\Vert_{H^{s,0}}^2 + \Vert \sqrt{-g\rhob'_{\delta}} \eta \Vert_{H^{s,0}}^2 \right) = 0.
\end{equation}
From such an estimate the uniqueness of solutions of \eqref{eqn:euler:isopycnal:no_shear} follows, and the existence classically follows.\\
Thanks to the incompressibility condition, we get $w \in C^0([0,+\infty),H^{s,1}(S))$. Then, differentiating the equation on $\eta$ in \eqref{eqn:euler:isopycnal:no_shear} and using the incompressibility condition, we get 
$$\partial_t \partial_r \eta =- \partial_x V,$$
which yields $\eta \in C^0([0,+\infty),H^{s,1}(S))$ and \eqref{eqn:nrj:delta:dreta}.
\end{proof}

\subsection{The bilayer Euler equations}
\label{subsection:euler:bilayer}
Before proving Proposition \ref{prop:nrj:bilayer}, let us briefly explain how the pressure $P_{\pm}$ in \eqref{eqn:euler:bilayer} is defined. Assume that $V_{\bl},w_{\bl},\zeta,P_{\bl}$ satisfy \eqref{eqn:euler:bilayer} with boundary conditions \eqref{eqn:euler:bilayer:bc}, as well as the continuity of $w_{\bl}$ across the interface. Then, taking the divergence of the first two equations in \eqref{eqn:euler:bilayer} yields
\begin{equation}
\label{eqn:def:P:bilayer}
\partial_x \frac{1}{\rho_{\bl}} \partial_x P_{\bl} + \partial_r \frac{1}{\rho_{\bl}} \partial_r P_{\bl} = - g\partial_x^2 \zeta \qquad \text{ in } S_{\bl}.
\end{equation} 
Let us insist on the fact that \eqref{eqn:def:P:bilayer} defines $P_+$ as the solution to an elliptic equation (with boundary conditions yet to define) in the upper layer $S_+$ (resp. $P_-$ in $S_-$). Taking the trace of the second equation in \eqref{eqn:euler:bilayer} at $r=-H$ and $r=0$ yields
\begin{equation}
\label{eqn:def:P:bilayer:bc}
\sys{\partial_r P_{+|r=0} &=0,\\
	\partial_r P_{-|r=-H} &= 0.} 
\end{equation}
At the interface $r=r_*$ between both layers, we do not impose boundary conditions, but rather transmission conditions: first, we impose the continuity of the pressure across the interface, which stems from the absence of surface tension. Taking the trace of the second equation in \eqref{eqn:euler:bilayer} at $r=r_*$ and using the continuity of $w_{\bl}$ across the interface (in \eqref{eqn:euler:bilayer:bc}), we finally write the two transmission conditions
\begin{equation}
\label{eqn:def:P:bilayer:tc}
\sys{P_{+|r=r_*} &= P_{-|r=r_*},\\
	\frac{1}{\rho_+} (\partial_r P_+)_{|r=r_*} &= \frac{1}{\rho_-} (\partial_r P_-)_{|r=r_*}.}
\end{equation}%
In order to study \eqref{eqn:def:P:bilayer}, we define the following Hilbert space
\begin{equation}
\label{eqn:def:hilbert:P:bilayer}
\mathcal H \coloneq  \{ (P_+,P_-)\in L^2_{\loc}(S_+) \times L^2_{\loc}(S_-), \nabla_{x,r} P_+ \in L^2(S_+), \nabla_{x,r} P_- \in L^2(S_-), P_{+|r=r_*} = P_{-|r=r_*} \}/\R, 
\end{equation}
where $L^2_{\loc}$ is the space of locally square integrable functions. As before, we write $P_{\bl} \in \mathcal H$ for $(P_+,P_-) \in \mathcal H$, with a slight abuse of notations. We endow $\mathcal H$ with the scalar product
$$ \langle P_{\bl} , \phi_{\bl} \rangle_{\mathcal H} \coloneq  \int_{S_+} \nabla_{x,r} P_+ \cdot \nabla_{x,r} \phi_+ + \int_{S_-} \nabla_{x,r} P_- \cdot \nabla_{x,r} \phi_-.$$%
Note that $\langle P_{\bl},P_{\bl} \rangle_{\mathcal H} = 0$ implies that $P_+$ and $P_-$ are constants. Because $P_+=P_-$ on the interface $r=r_*$, these constants are equal. Because the space $\mathcal H$ is quotiented by constants, this yields that $P_{\bl} = 0$ in $\mathcal H$. We now state the following lemma.

\begin{lemma}
Let $s \geq 0$, $\zeta \in H^{s}(\mathbb T_L)$. There exists a unique pressure $P_{\bl} \in \mathcal H$ satisfying \eqref{eqn:def:P:bilayer} with boundary conditions \eqref{eqn:def:P:bilayer:bc} and the transmission conditions \eqref{eqn:def:P:bilayer:tc}. Moreover, there exists $C > 0$ depending only on $\rho_+, \rho_-,g$ such that 
\begin{equation}
\label{eqn:P:bilayer:estimates}
\Vert \partial_x P_{\bl}\Vert_{H^{s,0}} + \Vert \partial_r P_{\bl} \Vert_{H^{s,0}} + \left\Vert \partial_r \frac{1}{\rho_{\bl}} \partial_r P_{\bl} \right\Vert_{H^{s-1,0}} \leq C \vert \partial_x \zeta \vert_{H^{s}}.
\end{equation}
\end{lemma}
\begin{proof}
Let $\phi_{\bl} \in \mathcal H$, and let $P_{\bl} $ be a smooth solution to \eqref{eqn:def:P:bilayer} with boundary conditions \eqref{eqn:def:P:bilayer:bc} and the transmission conditions \eqref{eqn:def:P:bilayer:tc}. Multiplying \eqref{eqn:def:P:bilayer} with $\phi_{\bl}$ and integrating by parts, we get the variational formulation of \eqref{eqn:def:P:bilayer}:
\begin{equation}
\label{eqn:var}
\int_{S_+} \frac{1}{\rhob_+}\nabla_{x,r} P_+ \cdot \nabla_{x,r} \phi_+ + \int_{S_-} \frac{1}{\rhob_-}\nabla_{x,r} P_- \cdot \nabla_{x,r} \phi_- = - \int_{S_+} g\partial_x \zeta \partial_x \phi_+ - \int_{S_-} g\partial_x \zeta \partial_x \phi_-.
\end{equation} 
The existence and uniqueness of solutions to the variational problem \eqref{eqn:var} is classical, we refer to \cite{Borsuk2010} for transmission problems, as well as \cite[Chapter 2]{Lannes2013} for boundary problems close to our setting. Then, the estimate \eqref{eqn:P:bilayer:estimates} is a standard elliptic estimate, see for instance \cite[Chapter 2]{Lannes2013}.
\end{proof}
\begin{proof}[Proof of Proposition \ref{prop:nrj:bilayer}]
We only sketch the energy estimate that classically leads to the well-posedness of \eqref{eqn:euler:bilayer}. Integrating the equations \eqref{eqn:euler:bilayer} against $(\rho_{\bl} V_{\bl}, \rho_{\bl} w_{\bl})$ and \eqref{eqn:euler_bilayer:zeta} against $(\rho_- - \rho_+) \zeta$, we get
\begin{equation}
\label{eqn:nrj:bilayer}
\begin{aligned}
&\frac12 \frac{d}{dt} \left( \Vert \sqrt{\rho_{\bl}} V_{\bl}\Vert_{L^2(S)}^2 + \Vert \sqrt{\rho_{\bl}} w_{\bl}\Vert_{L^2(S)}^2 + \vert \sqrt{(\rho_--\rho_+)}\zeta\vert_{L^2(\mathbb T_L)}^2 \right) \\
&= - \int_{S} \nabla_{x,r} P_{\bl} \cdot (V_{\bl},w_{\bl})^T - \int_S \rho_{\bl}\partial_x \zeta \cdot V_{\bl} + \int_{\mathbb T_L} (\rho_- - \rho_+)w_{\bl|r=r_*}\zeta.
\end{aligned}
\end{equation}
The first integral on the right-hand side of \eqref{eqn:nrj:bilayer} is zero, by integration by parts and after using the continuity of $P_{\bl} w_{\bl}$ across the interface. We also integrate by parts and use the incompressibility condition in \eqref{eqn:euler:bilayer} in the second integral on the right-hand side of \eqref{eqn:nrj:bilayer} to get
$$\int_S \rho_{\bl}\partial_x \zeta \cdot V_{\bl} = - \int_S \rho_{\bl} \zeta \partial_x V_{\bl} = \int_S \rho_{\bl} \zeta \partial_r w_{\bl} = \int_{\mathbb{T}_L} (\rho_--\rho_+)w_{\bl|r=r_*} \zeta.$$
Thus, the second and third integrals in \eqref{eqn:nrj:bilayer} cancel each other out, and we get
\begin{equation}
\label{eqn:nrj:bilayer:2}
\frac{d}{dt} \left( \Vert \sqrt{\rho_{\bl}} V_{\bl}\Vert_{L^2(S)}^2 + \Vert \sqrt{\rho_{\bl}} w_{\bl}\Vert_{L^2(S)}^2  + \vert \sqrt{(\rho_--\rho_+)}\zeta\vert_{L^2(\mathbb T_L)}^2 \right) =0.
\end{equation}
Applying $\partial_x^{\alpha}$ to \eqref{eqn:euler:bilayer} with $|\alpha| \leq s$, we get that $(\partial_x^\alpha V_{\bl}, \partial_x^\alpha w_{\bl}, \partial_x^\alpha \zeta)$ satisfies \eqref{eqn:euler:bilayer}, so that \eqref{eqn:nrj:bilayer:2} holds for the $H^{s,0}$ norm of the unknowns. From such an energy estimate, the uniqueness of solutions to \eqref{eqn:euler:bilayer} follows. The existence also follows and we omit it.
\end{proof}
\subsection{From \eqref{eqn:euler:isopycnal:no_shear} to \eqref{eqn:euler:bilayer}}
\label{subsection:strat_to_bilayer}
Before Proving Proposition \ref{prop:cvce}, let us comment on the density profile given in \eqref{eqn:def:sharp_strat}.
\begin{remark}
\label{rk:rhobdelta_plus_gen}
In the proof of Proposition \ref{prop:cvce}, we only use some properties of the specific profile $\rhob_{\delta}$ given in \eqref{eqn:def:sharp_strat}. In fact, Proposition \ref{prop:cvce} holds for profiles $\rhob_{\delta}$ of the following form.\\
Let $r_* \in (-H,0)$, $\delta > 0$, $\rho_- > \rho_+ > 0$. Then
$$\rhob_{\delta}(r) \coloneq  \rhob\left(\frac{r-r_*}{\delta}\right),$$
where $\rhob$ is a smooth, strictly decreasing function from $\R$ to $\R$ satisfying the following conditions. First, 
\begin{equation}
\begin{aligned}
\rhob(y) &\underset{y \to - \infty}{\longrightarrow} \rho_{-},\\
\rhob(y) &\underset{y \to + \infty}{\longrightarrow} \rho_{+}.
\end{aligned}
\end{equation}
Moreover:
\begin{equation}
\label{eqn:rhobdelta:plus:gen:L2}
\vert \rhob - \rho_{\pm} \vert_{L^2(\R)} < + \infty,
\end{equation}
where $y \mapsto \rho_{\pm}$ is the piecewise constant function with values $\rho_+$ for $y > 0$ and $\rho_-$ for $y < 0$. Second, the function
\begin{equation}
\label{eqn:rhobdelta:plus:gen:hyp}
y \mapsto y^{\frac32} \rhob'(y)
\end{equation}
belongs to $L^{\infty}(\R)$. Note that the profile $\rhob_{\delta}$ defined in \eqref{eqn:def:sharp_strat} does satisfy these conditions.
\end{remark}
\begin{proof}[Proof of Proposition \ref{prop:cvce}]
We write $(\tilde{V},\tilde{w},\tilde{P},\tilde{\eta}_*) \coloneq  (V - V_{\bl},w - w_{\bl},\check{P}-P_{\bl},\eta_* -  \zeta )$, where $\check{P} \coloneq  P - g\rhob_{\delta}\eta_* $ (with $P$ as in \eqref{eqn:euler:isopycnal:no_shear} and $\eta_* \coloneq  \eta_{r=r_*}$). We then take the difference between \eqref{eqn:euler:isopycnal:no_shear} and \eqref{eqn:euler:bilayer} to get
\begin{equation}
\label{eqn:euler:diff}
\sys{\partial_t \tilde{V} + \frac{1}{\rhob_{\delta}} \partial_x \tilde{P} + g\frac{\rho_{\bl}}{\rhob_{\delta}} \partial_x \tilde{\eta}_*  &= -\left(\frac{1}{\rhob_{\delta}} - \frac{1}{\rho_{\bl}}\right)\partial_x P_{\bl} + g\left(\frac{\rho_{\bl}}{\rhob_{\delta}} -1 \right) \partial_x \tilde{\eta}_*,\\
	\partial_t \tilde{w} + \frac{1}{\rhob_{\delta}} \partial_r \tilde{P} &= -\left(\frac{1}{\rhob_{\delta}} - \frac{1}{\rho_{\bl}}\right)\partial_r P_{\bl} + g\frac{\rhob_{\delta}'}{\rhob_{\delta}} (\eta-\eta_*),\\
	\partial_t \tilde{\eta}_* - \tilde{w}_{|r=r_*} &=0.\\
	}
\end{equation}
We integrate against $\rhob_{\delta} \tilde{V}, \rhob_{\delta} \tilde{w}, g(\rho_- - \rho_+) \tilde{\eta}_*$. We get
\begin{equation}
\label{eqn:nrj:diff}
\begin{aligned}
&\frac12 \frac{d}{dt}\left( \Vert \sqrt{\rhob_{\delta}} \tilde{V}\Vert_{L^2}^2 + \Vert \sqrt{ \rhob_{\delta}} \tilde{w}\Vert_{L^2}^2+ \vert \sqrt{g(\rho_--\rho_+)} \tilde{\eta}_* \vert_{L^2}^2 \right) + \int_S g\rho_{\bl}\partial_x\tilde{\eta}_* \tilde{V} - \int_{\mathbb{T}_L} g(\rho_--\rho_+) \tilde{w}_{r=r_*}  \tilde{\eta}_* \\
&\leq  C \left\vert \rho_{\bl} - \rhob_{\delta}\right\vert_{L^2([-H,0])}\left(\left\Vert \frac{1}{\rho_{\bl}}\nabla_{x,r} P_{\bl} \right\Vert_{L^{\infty}}  +  \vert \partial_x \tilde{\eta}_* \vert_{L^{\infty}} \right) \left( \Vert \tilde{V} \Vert_{L^2} + \Vert \tilde{w} \Vert_{L^2} \right) + \Vert \rhob_{\delta}' (\eta - \eta_*) \Vert_{L^1} \Vert \tilde{w}\Vert_{L^{\infty}},
\end{aligned}
\end{equation}
where we used that $(\tilde{V},\tilde{w})$ is divergence-free so that the contribution of the pressure term $\nabla_{x,r} \tilde{P}$ in \eqref{eqn:nrj:diff} vanishes. Similarly to \eqref{eqn:nrj:bilayer}, we integrate by parts in the first integral on the left-hand side of \eqref{eqn:nrj:diff}, so that it cancels out with the second integral on the left-hand side of \eqref{eqn:nrj:diff}. 
We now bound from above the terms on the right-hand side of \eqref{eqn:nrj:diff}.
For the first term on the right-hand side of \eqref{eqn:nrj:diff}, we use \eqref{eqn:P:bilayer:estimates}, and the Sobolev embedding $H^{1,1}(S) \hookrightarrow L^{\infty}(S)$ (see for instance \cite[Proposition 2.12]{Lannes2013}), to write
\begin{equation}
\left\vert \rho_{\bl} - \rhob_{\delta}\right\vert_{L^2([-H,0])}\left(\left\Vert \frac{1}{\rho_{\bl}}\nabla_{x,r} P_{\bl} \right\Vert_{L^{\infty}}  +  \vert \partial_x \tilde{\eta}_* \vert_{L^{\infty}} \right) \leq C\delta \left( \vert \eta_* \vert_{H^{2}} + \vert \zeta \vert_{H^2} \right),
\end{equation}
where we also used that 
$$\left\vert \rho_{\bl} - \rhob_{\delta}\right\vert_{L^2([-H,0])} \lesssim \delta,$$%
from a change of variable and \eqref{eqn:rhobdelta:plus:gen:L2}. For the second term in \eqref{eqn:nrj:diff}, we first write
\begin{equation}
\label{eqn:nrj:diff:temp:11}
\Vert \rhob_{\delta}' (\eta-\eta_*)\Vert_{L^1} \leq \left\Vert \rhob_{\delta}' \int_{r_*}^r \partial_r \eta \right\Vert_{L^1} \leq C \Vert (r-r_*)^{\frac32 - \epsilon} \rhob_{\delta}' \Vert_{L^{\infty}} \Vert \partial_r \eta \Vert_{L^2} \vert (r-r_*)^{ -1+\epsilon}\vert_{L^1([-H,0])} \leq C \frac{\delta^{\frac12 - \epsilon} }{\epsilon}(1+t), 
\end{equation}
for any $ \epsilon \in (0,\frac12)$, where we used the property \eqref{eqn:rhobdelta:plus:gen:hyp} satisfied by $\rhob_{\delta}$, as well as Proposition \ref{prop:nrj}. Optimizing the value of $\epsilon$, we choose $\epsilon = \frac{1}{|\log \delta|}$, and note that $\delta^{\frac{1}{|\log \delta|}} = e$ to get
\begin{equation}
\label{eqn:nrj:diff:temp:12}
\Vert \rhob_{\delta}' (\eta-\eta_*)\Vert_{L^1} \lesssim \delta^{\frac12} |\log \delta |(1+t). 
\end{equation}
Thus, \eqref{eqn:nrj:diff} yields the desired \eqref{eqn:erreur:bl} with $s=0$. Differentiating \eqref{eqn:euler:diff} with respect to $x$ and applying this $L^2$ estimate allows us to get \eqref{eqn:erreur:bl} with $s > 0$.
\end{proof}
\section{On the numerical approximation in Section \ref{section:sharp_strat}}
\label{section:apdx:approx}
Let us now comment on the precision of the numerical approximation of the phase velocities of \eqref{eqn:euler:isopycnal:Vb}, that gives rise to Figures \ref{fig:dispersion:strat} - \ref{fig:kstar}, and in particular Figure \ref{fig:dispersion:long}. In this appendix we check numerically that the functions $w_{\mathtt N}$ and phase velocities $c_{\mathtt N}$ given by the diagonalization of the  matrices $B_k$ as in \eqref{eqn:EDO:modes}, approximately satisfy the Taylor-Goldstein equation \eqref{eqn:taylor_goldstein}.\\

More precisely, with the notations of Section \ref{section:num_scheme}, let $\lambda$ be an eigenvalue of $B_k$ as in \eqref{eqn:EDO:modes}, associated to the eigenvector $(\hat{V}_1, \dots, \hat{V}_{\mathtt N}, \hat{\eta}_1, \dots, \hat{\eta}_{\mathtt N})$. Then, we define $\hat{w}$ thanks to the divergence-free equation, namely
\begin{equation}
\label{eqn:def:w:div_free}
\hat{w} \coloneq  - i k \sum_{n=1}^{\mathtt{N}} c_n \hat{V}_n f_n,
\end{equation} 
with $(f_n){n\in \N^*}$ defined in \eqref{eqn:SL}. Then, we plug \eqref{eqn:def:w:div_free} into \eqref{eqn:taylor_goldstein}. To this end, we define $c \coloneq  -i \lambda /k$, and note that differentiating \eqref{eqn:def:w:div_free} with respect to $r$ can be done using \eqref{eqn:SL}. We define:
\begin{equation}
\mathrm{TG} \coloneq  (\Vb - c)^2 \left( ik \sum_{n=1}^{\mathtt{N}} V_n \frac{N^2}{c_n} f_n - k^2 \rhob \hat{w}  \right) -(\Vb - c) (\rhob \Vb')'\hat{w}+ \rhob N^2 \hat{w};
\end{equation}
it depends only on $r \in [-H,0]$. Observe that $\hat{w} \mapsto \mathrm{TG}$ can be seen as an operator, depending on the parameter $c$, and bounded from $W^{2,\infty}$ to $L^{\infty}$, though with a norm that depends on $\delta$. It is thus natural to renormalize the error which measures how well $(c,w)$ is an approximate solution (i.e. an element of the kernel of TG) of \eqref{eqn:taylor_goldstein} by the $W^{2,\infty}$ norm of $\hat{w}$, which can be computed from \eqref{eqn:def:w:div_free} and \eqref{eqn:SL}, \eqref{eqn:def:gn}:
\begin{equation}
\label{eqn:W2infty:w}
\vert \hat{w} \vert_{W^{2,\infty}} = \max_{r\in [-H,0]} \left( \left| k \sum_{n=1}^{\mathtt N} c_n \hat{V}_n f_n(r) \right|  + \left|  k \sum_{n=1}^{\mathtt N} \hat{V}_n g_n(r) \right| + \left|  k \rhob_{\delta}'(r) \sum_{n=1}^{\mathtt N} \hat{V}_n (g_n(r) + \frac{1}{c_n} f_n(r)) \right| \right)
\end{equation}
We now define the error which measures how well $(c,w)$ is an approximate solution of \eqref{eqn:taylor_goldstein}:
\begin{equation}
\label{eqn:def:error:TG}
\mathtt{Err}_{c} \coloneq \frac{\max_{r\in [-H,0]}|\mathrm{TG}(r) |}{\vert \hat{w} \vert_{W^{2,\infty}}}.
\end{equation}
%
We then plot the phase velocities from Figure \ref{fig:dispersion:long} with $\delta = 10^{-2}$, and represent $\mathtt{Err}_c$ for each point $c$.\\
We perform the same simulation as for Figure \ref{fig:dispersion:long}, however with $\mathtt K = 100$, and show the result in Figure \ref{fig:dispersion:Errc}. We observe that the error is greater for some elements $c$ corresponding to continuous spectrum, i.e. $\mathrm{Re}(c)$ is in the range of $\Vb$.
\begin{remark}
\label{rk:cvce:spectre}
Proving that our numerical approximation gives rise to sequences $(w_{\mathtt N})$ and $(c_{\mathtt N})$ (where $\mathtt N$ denotes the number of modes used for the computations) which converge to some solution $w,c$  of the Taylor-Goldstein equation \eqref{eqn:taylor_goldstein} when $\mathtt N \to +\infty$ is out of the scope of the present study. In particular, because \eqref{eqn:taylor_goldstein} is a Quadratic Eigenvalue Problem, this is not immediate. Moreover, let
\begin{equation}
\label{eqn:def:T}
T \begin{pmatrix} V \\ w \\ \eta \end{pmatrix} = 
\begin{pmatrix} 
\Vb \cdot \nabla_x V + w \Vb' + \frac{1}{\rhob} \nabla_x P \\
\Vb \cdot \nabla_x w + \frac{1}{\rhob} \partial_r P - \frac{g \rhob'}{\rhob} \eta \\
\Vb \cdot \nabla_x V - w
\end{pmatrix}
\end{equation}
be the unbounded operator defined on the space of $L^2$ vector fields with $(V,w)$ divergence-free, such that the stratified Euler equations \eqref{eqn:euler:isopycnal:Vb} read
$$ \partial_t \begin{pmatrix} V \\ w \\ \eta \end{pmatrix} + T \begin{pmatrix} V \\ w \\ \eta \end{pmatrix} =0.$$%
The operator $T$ is neither self-adjoint nor skew self-adjoint. Therefore, if $f,\lambda$ are such that $Tf - \lambda f$ is small, one cannot conclude that $\lambda$ is close to an element of the spectrum of $T$. Thus, the evaluation of the error $\mathtt{Err}_{c}$ defined in \eqref{eqn:def:error:TG} provides an estimation of the consistency of our scheme, but not on its convergence, and it is not enough to conclude rigorously that the phase velocities that we show in Figure \ref{fig:dispersion:Errc} associated with a small error are indeed close to the phase velocities of the stratified Euler equations \eqref{eqn:euler:isopycnal:Vb}. Such a proof falls outside of the scope of this study, and we simply consider the error $\mathtt{Err}_c$ as an indication of the reliability of the associated $c$ as being close to a point of the phase velocities of \eqref{eqn:euler:isopycnal:Vb}.
\end{remark}
We are particularly interested in the consistency of the eigenvalues corresponding to the Kelvin-Helmholtz instability. To this end, we define 
\begin{equation}
\label{eqn:def:error:TG:KH}
\mathtt{Err}_{KH} \coloneq  \max_{c, |\Im(c)| \geq 10^{-2}, \mathrm{Re}(c) \geq 10} \mathtt{Err}_{c}.
\end{equation}
The range of points $c$ upon which the supremum is taken in \eqref{eqn:def:error:TG:KH} is chosen to coincide with the Kelvin-Helmholtz instabilities, see Figure \ref{fig:dispersion:Errc}.  We show $\log(\mathtt{Err}_{KH})$ computed for different number of modes $\mathtt N$ in Figure \ref{fig:dispersion:ErrKH}. This indicates that indeed, the couples $(c,w)$ are approximate solutions of the Taylor-Goldstein equation \eqref{eqn:taylor_goldstein}.
\begin{figure}
\centering
\begin{minipage}[t]{0.49\textwidth}
\vspace{0pt}
\includegraphics[width = \textwidth]{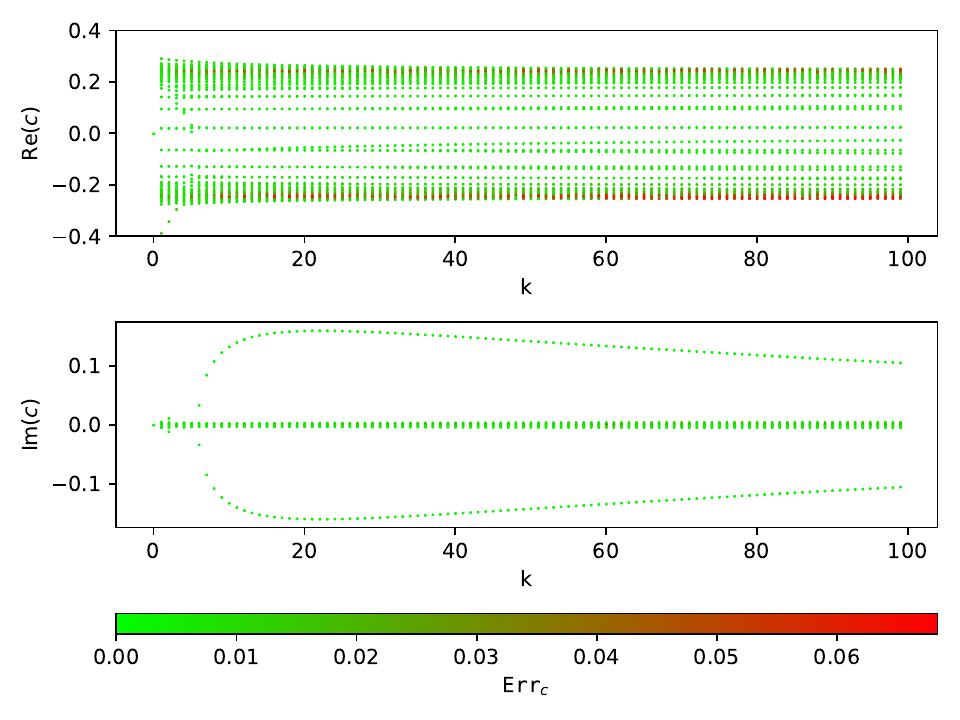}
\captionof{figure}{Dispersion relation for $\delta =  10^{-2}$ for the stratified Euler equations obtained numerically with $\mathtt N = 40$. We plot the numerical approximation of phase velocities, i.e. elements of $\mathcal C^{(\mathtt N)}_{\delta}(k)$ versus $k$. The color of a point $c$ shows $\mathtt{Err}_c$. Because of symmetry, we only show the positive frequencies $k$.}
\label{fig:dispersion:Errc}
\end{minipage}
\begin{minipage}[t]{0.49\textwidth}
\vspace{-15pt}
\includegraphics[width = \textwidth]{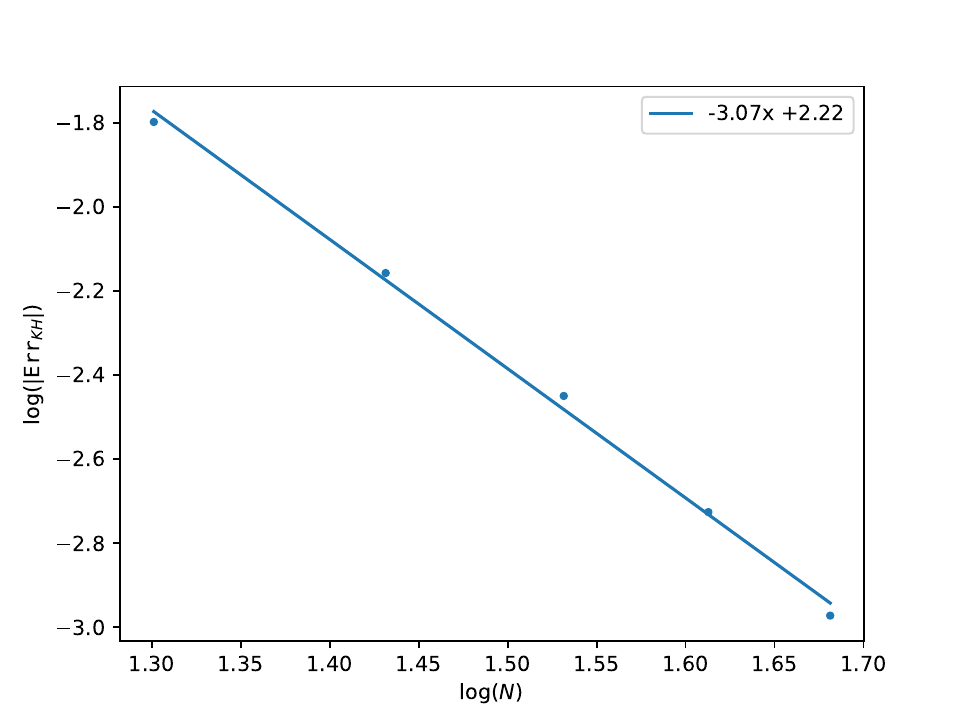}
\captionof{figure}{Plot of $\mathtt{Err}_{KH}$ versus the number of modes $\mathtt N$, with $\delta = 10^{-2}$. The number of modes ranges from $\mathtt N=20$ to $\mathtt N=50$. }
\label{fig:dispersion:ErrKH}
\end{minipage}
\end{figure}

\section*{Acknowledgments}
The author thanks his PhD supervisors Vincent Duchêne and David Lannes for their valuable advice.

\section*{Data availability}
The scripts used to generate the figures in the present study are available in \cite{Fradin2025num} (Zenodo upload 17941779).

\section*{Funding}
This work was supported by the BOURGEONS project, grant ANR-23-CE40-0014-01 of the French National Research Agency (ANR).

\section*{Conflict of interest}
The author declares that there are no conflicts of interest to disclose. 

\section*{Other declarations}
Other declarations such as ethics approval statement, patient consent statement, permission to reproduce material from other sources, clinical trial registration do not apply to the present work.\\

This work is licensed under \href{https://creativecommons.org/licenses/by/4.0/}{ CC BY 4.0}.

\end{document}